\documentclass[a4paper,11pt]{article}
\usepackage[margin=25mm]{geometry}
\usepackage{amsmath,amssymb,amsthm,booktabs,graphicx,xcolor}
\usepackage[unicode,colorlinks=true,linkcolor=blue!45!black,urlcolor=blue!45!black,citecolor=blue!45!black]{hyperref}
\usepackage{enumitem}
\usepackage{float}
\usepackage{xurl}
\setlist{nosep}

\newtheorem{theorem}{Theorem}
\newtheorem{proposition}[theorem]{Proposition}
\newtheorem{lemma}[theorem]{Lemma}
\newtheorem{corollary}[theorem]{Corollary}
\newtheorem{remark}[theorem]{Remark}
\newcommand{\E}{\mathbb E}
\newcommand{\Var}{\operatorname{Var}}
\newcommand{\Cov}{\operatorname{Cov}}
\newcommand{\Pois}{\operatorname{Pois}}
\newcommand{\given}{\,\middle|\,}
\title{\large Decoding Failure Biases IBLT Count Measurements: A Gaussian-Selection Law at the Peeling Threshold}
\author{Min Wu}
\date{October 2026}
\begin{document}
\maketitle
\begin{abstract}
In a self-sizing set-reconciliation protocol, a small first-round Invertible Bloom Lookup Table (IBLT)
sketch attempts peeling decoding; after failure, the protocol estimates the difference count $d$ from
the count energy and sizes the second round from a lower quantile of the measurement. This quantile must
be computed conditionally on failure: the measurement and peeling success are determined by the same
random mapping, so conditioning can shift the measurement law.

In the ideal peeling model, inside the critical window where failure has nondegenerate probability, the
standardized energy and a statistic of the unsigned mapping that governs the peeling outcome (the
precritical score) have a joint Gaussian limit with correlation $\eta_k\rho^2$, $\rho$ being the sign
imbalance and $\eta_k$ a closed-form constant. The measurement conditional on failure is thus a
bivariate normal truncated to a half-plane, which we call the Gaussian-selection law. Since $\eta_k$
changes sign between $k=4$ and $k=5$, under one-sided imbalance the unconditional threshold undercovers
exactly when $k\ge5$ and is conservative for $k=3,4$.

The key tool is a quenched CLT for a signed quadratic form on a random overlap graph; under
asymptotically balanced signs, conditioning on any unsigned mapping event that is permutation-invariant
and has probability bounded away from zero leaves the standard normal limit unchanged. After failure the
protocol observes the total sign; inverting the selection law yields an upper confidence bound for $d$,
asymptotically conservative in the critical window and, when $d/\sqrt M\to\infty$, along the whole load
axis. The experiments use analytic constants, none fitted; the unconditional threshold carries a
selection error that does not vanish with scale, while the corrected selection-law threshold matches
per-configuration empirical quantiles.
\end{abstract}

\section{Introduction}
\subsection{Background: reading the measurement after failure}
Set reconciliation requires the two parties to recover the symmetric difference of their sets $A$ and $B$.
Database consistency verification is a canonical scenario: each side scans the whole table, compresses every row into a short fingerprint, and comparing the two tables then reduces to computing the symmetric difference of the two fingerprint sets. The number of differing elements $d=|A\triangle B|$ is unknown in advance,
and can vary by several orders of magnitude at different points in time for the same table. The protocol aims to finish in as few rounds as possible while not transmitting a sketch much larger than necessary.

The Invertible Bloom Lookup Table (IBLT) is a common structure for this task. It has $M$ cells, and each cell
stores a count, a key sum, and a checksum. Each element hashes to $k$ distinct cells and adds itself to those cells. Each side builds a
sketch and then the two sketches are subtracted cell by cell; the elements shared by both sides cancel out, leaving only the differing elements: an element only in $A$ carries sign $+1$, and an element only in $B$
carries sign $-1$. Decoding uses peeling: repeatedly find a cell containing a single differing element, read that element out, and delete it from the $k$ cells
to which it belongs. Viewing cells as vertices and the $k$ cells selected by each element as a hyperedge, peeling deletes edges incident to degree-one vertices on this random $k$-uniform
hypergraph. When the load $d/M$ is below the critical value $\alpha_c$ (approximately $0.818$ for $k=3$\cite{molloy,dw}), peeling deletes every edge with high probability; once the load passes $\alpha_c$,
peeling stalls at a nonempty residual 2-core and decoding fails. The failure probability rises from near zero to near one over an interval of width $M^{-1/2}$
around $\alpha_c$.

The capacity $M$ of the IBLT must be fixed before $d$ is known. Choose it too small and decoding fails, requiring another round; choose it too large and bandwidth and memory are wasted. Reference~\cite{prior}
proposes a self-sizing protocol that proceeds in two rounds. The first round sends a fixed small sketch; when there are few differences peeling succeeds and the round ends,
which is the case in most runs in actual deployments. Peeling failure only means the differing elements cannot be listed one by one; the count array held by the receiver is not lost, and it still records $d$.
Among the $d$ differing elements the numbers carrying $+1$ and $-1$ are written $d_+$ and $d_-$; this composition is called the sign composition, and it is summarized by the
sign imbalance $\rho=(d_+-d_-)/d$: $\rho=\pm1$ means all differences come from one side, and $\rho=0$ means the two sides each contribute half. Write $C_v$ for the count of cell $v$, the sum of the signs of the differing elements in that cell. The count energy
\begin{equation*}
T=\sum_{v}\bigl(C_v-\overline C\bigr)^2,\qquad \overline C=\frac1M\sum_vC_v,
\end{equation*}
can be computed before peeling. Each differing element contributes roughly $k$ to the energy, so $\widehat d=T/[k(1-k/M)]$ is an estimate of $d$;
for any sign composition it is exactly unbiased, with variance $2d(d-1)/(M-1)$. When the first round fails, the protocol reads
$\widehat d$, builds the second-round sketch with a fresh hash, and sets its capacity from $\widehat d$.

The estimate $\widehat d$ has random error, and sizing the capacity directly from $\widehat d$ understates $d$ with probability about one half. The protocol therefore introduces a lower
quantile $q_u<1$, takes $\widehat d/q_u$ as an upper bound on $d$, and sizes the second-round sketch from it:
\begin{equation*}
\Pr\bigl[\widehat d\ge q_u d\bigm| F\bigr]\ge1-u,
\end{equation*}
where $F$ is the event that the first-round peeling fails, and deployments take $u=0.01$. A smaller $q_u$ covers more conservatively but wastes space, and a larger one gives insufficient coverage probability,
so $q_u$ must be computed accurately.

Reference~\cite{prior} gives two ways to obtain $q_u$. The first is the chi-square approximation under deep overload:
when $d$ far exceeds the first-round capacity, failure occurs almost surely, conditioning on failure hardly changes the distribution, and the unconditional $\widehat d/d$ with $M$ fixed and
$d\to\infty$ tends to a standardized chi-square $(M-1)$ distribution. The second estimates an empirical
quantile by simulation for each configuration. The former applies only under deep overload; the latter works at any load but must be sampled separately for each configuration and yields no cross-configuration pattern.
What neither answers is the \emph{critical window}\cite{dm}, where failure and success have comparable probability.

\subsection{The problem: selection of the mapping by failure}
The lower quantile $q_u$ of the previous section comes from the unconditional distribution of $\widehat d$, but the protocol uses $\widehat d$ only after the first round has failed. If conditioning on failure changes the distribution of $\widehat d$, then a $q_u$ calibrated from the unconditional distribution no longer gives the correct coverage probability, and the second-round capacity systematically deviates from its target.

Intuitively, the count array is fixed before peeling begins, so the outcome of peeling seems unable to affect $\widehat d$. This intuition misses one point: $\widehat d$ and the success or failure of peeling are both determined by the same random mapping. The protocol reads $\widehat d$ only on mappings that leave a nonempty 2-core, which is equivalent to selecting a subset of all mappings and then measuring. Selection changes the distribution of the mapping, and may change the distribution of $\widehat d$.

Figure~\ref{fig:demo} shows this, and also reveals that the strength of the selection effect depends on the sign composition. Two sign configurations are compared on the same set of mappings (so peeling succeeds or fails identically): under one-sided signs ($\rho=1$) the $\widehat d$ of failing samples shifts right as a whole ($\mathbb E[X\mid F]=+0.609$), while under balanced signs ($\rho=0$) it is nearly unchanged ($\mathbb E[X\mid F]=+0.008$).
\begin{figure}[tbp]
\centering\includegraphics[width=\linewidth]{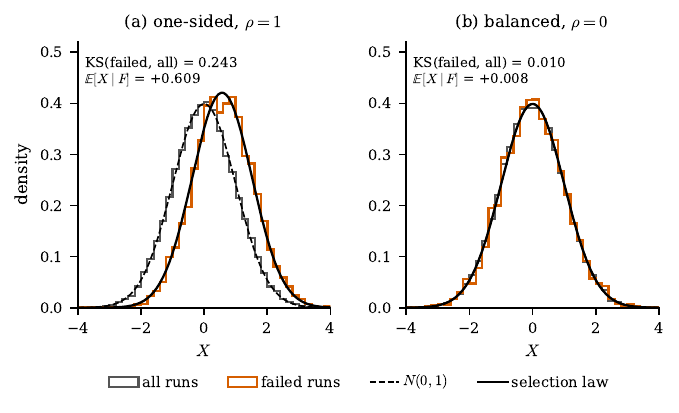}
\caption{Selection of the mapping by failure ($k=3$, $M=16384$, failure rate $0.125$, 40,000 Plain hypergraphs). (a) Distribution of $X$ over all samples and over failing samples under one-sided signs, with the solid line the selection law; (b) the same set of hypergraphs with balanced signs.}
\label{fig:demo}
\end{figure}

This dependence has a clear mechanism. The success or failure of peeling depends only on the topology of the hypergraph, not on the signs carried by the elements, but the signs determine how the count energy responds to changes in the topology. When two elements share a cell, equal signs raise the energy and opposite signs lower it. Failure selects the degree profile of the hypergraph and thereby shifts the amount of overlap between elements; for $k=3$, the case of Figure~\ref{fig:demo}, the failing hypergraphs carry more overlap. When same-sign pairs dominate ($\rho^2$ near one), this change in overlap moves the energy consistently in one direction, and the selection effect passes fully to $\widehat d$; when the signs are nearly balanced ($\rho^2$ near zero), the positive and negative contributions cancel, and the selection effect vanishes.

The selection effect is most pronounced inside the critical window. Under deep overload failure is nearly certain, so conditioning on failure hardly selects the mapping; inside the critical window the failure probability rises from zero to one, and the selection is strongest. The magnitude of the bias is controlled by $\rho^2$ and its direction by $k$, so a capacity set from the unconditional $q_u$ may be too large or too small. In database verification the differences produced by replication delay are almost all of one sign, with $\rho^2$ near one, exactly the case of maximal bias.

This paper answers two questions inside the critical window. After the first-round peeling fails, the protocol holds the estimate $\widehat d$ while the true $d$ is unknown. The distribution question asks what the limit distribution of $\widehat d$ conditional on failure is, and how it depends on $k$ and the sign composition. The inference question asks how to construct an upper bound on $d$ from the observables available after failure, so that the second-round capacity covers the true difference count with high probability.

\subsection{Main results}

\paragraph{Gaussian-selection law.}
The main result of this paper (Theorem~\ref{thm:main}) gives the limit distribution of the measurement conditional on failure inside the critical window.
We construct a precritical score from the degree histogram shortly before peeling reaches the bottleneck; it predicts how many degree-one vertices remain when peeling
reaches the bottleneck, and its sign asymptotically determines the success or failure of peeling. Inside the critical window,
the standardized measurement and the precritical score are asymptotically joint Gaussian, with correlation coefficient proportional to $\rho^2$. The measurement conditional on failure is therefore the marginal
of a bivariate normal truncated on a half-plane, which we call the Gaussian-selection law.

The selection effect decays as $\rho^2$: when the signs are asymptotically balanced the conditional distribution returns to the standard normal. The direction of the bias is determined by $k$: for $k=3,4$
the conditional mean is raised, and keeping the unconditional threshold after failure stays on the conservative side; for $k\ge5$
the bias reverses direction, and the same calibration systematically undercovers. The sign change is not obvious. Under one-sided signs more overlap always raises the count energy, so the direction of the bias is the direction in which failure moves the overlap. Peeling failure depends on the full degree profile, and the correlation between the overlap and the failure score changes sign between $k=4$ and $k=5$: failing hypergraphs carry more overlap than average for $k=3,4$ and less for $k\ge5$. This sign change means that for different values of $k$ the calibration after failure
must be corrected in opposite directions, and cannot be handled uniformly by a single empirical rule.

\paragraph{Separation of the signs.}
Theorem~\ref{thm:sign} reduces the general sign composition to the all-same-sign case. The measurement decomposes into two asymptotically independent parts: one depends only on the mapping and has magnitude proportional to $\rho^2$; the other comes only from the random assignment of signs on a fixed mapping. Failure selects only the mapping, so it affects only the first part, and the selection effect scales with $\rho^2$.

\paragraph{Upper confidence bound.}
The selection law gives the conditional distribution, but its parameters depend on the unknown $d$ and cannot be used directly to set the capacity. After failure the protocol can observe $b=d_+-d_-$; testing each candidate difference count against the selection law and inverting yields an upper confidence bound on $d$ (Theorem~\ref{thm:upper}). Compared with the two methods described above, this bound is valid inside the critical window, all its parameters are given by analytic constants, it needs no per-configuration simulation, and it does not rely on the deep-overload asymptotic approximation. After adding the mixture law of the two failure mechanisms, the bound covers the whole load axis when $d/\sqrt M\to\infty$ (Proposition~\ref{prop:axis}), and the protocol can calibrate the second-round capacity with a single formula.

\paragraph{Organization.}
Section~\ref{sec:model} fixes the model and the measurement and separates the role of the sign composition from that of the mapping.
Section~\ref{sec:window} establishes the selection law inside the critical window and analyzes how the direction of the bias depends on $k$.
Section~\ref{sec:upper} constructs the upper confidence bound from the selection law, gives the
finite-size implementation rule, and extends it to the whole load axis under $d/\sqrt M\to\infty$. Section~\ref{sec:numerics} checks the theoretical predictions one by one and concludes with a post-failure coverage
test with unknown $d$. Section~\ref{sec:related} discusses related work. Appendix~\ref{app:notation} collects the notation used throughout; Appendices~\ref{sec:main-proof}--\ref{app:loadaxis} give the complete proofs; Appendix~\ref{app:exp} contains supplementary experiments.

\section{Model, measurement, and separation of signs}\label{sec:model}
This section fixes the model and the measurement, and then separates the role of the sign composition from that of the mapping.
After the separation, the part of the measurement that varies with the mapping is scaled by $\rho^2$, and the rest is
asymptotically independent of the mapping.
Section~\ref{sec:window} then only needs to answer one question: how the failure event selects the mapping in the
all-same-sign case.
The notation used throughout is collected in Appendix~\ref{app:notation}. Except when different scales must be compared
explicitly, the dependence of every quantity on $M$ is omitted; the finite-size sign
imbalance is written $\rho$, and its limit is written $\rho_\infty$.

\subsection{The Plain mapping and ideal peeling}
The IBLT of the introduction can be abstracted into a random hypergraph model: only the counts are kept, the key sums and
checksums are ignored, and decoding uses only the information of which cells
contain exactly one element.

Let the set of cells be $[M]=\{1,\ldots,M\}$, and let $d$ differing elements remain after the two sketches are subtracted. Element $i$ independently and
uniformly chooses a $k$-subset $S_i\subseteq[M]$ as its hyperedge in the hypergraph. The endpoints of a single edge are
distinct, but different elements may choose exactly the same subset. We call this model the
Plain model. Ideal peeling deletes at each step one edge incident to a degree-one vertex, until all edges are deleted or no
degree-one vertex remains. When the final residual hypergraph is nonempty, we say the failure event $F$ occurs. This event is
equivalent to the existence of a nonempty residual
2-core, hence independent of which degree-one vertex is chosen at each step.

Each element carries a fixed sign $s_i\in\{-1,+1\}$, indicating which set contains the element exclusively; the signs are independent of the random mapping. Write
$d_+=\#\{i:s_i=+1\}$ and $d_-=\#\{i:s_i=-1\}$, with $d_++d_-=d$. The ordered pair $(d_+,d_-)$
is called the sign composition, and it determines how many of the $d$ differing elements come from each side. Sign imbalance
\begin{equation*}
\rho=\frac{d_+-d_-}{d}=\frac1d\sum_{i=1}^d s_i
\end{equation*}
summarizes the asymmetry of the sign composition in a single scalar: $\rho=\pm1$ means all differences come from one side, and $\rho=0$
means the two sides each contribute half. The count of cell $v$ after the two sketches are subtracted is
\begin{equation*}
C_v=\sum_{i:v\in S_i}s_i.
\end{equation*}
The ideal model identifies a cell containing a single element directly from its unsigned degree. A real IBLT uses the key sums and checksums to decide whether a cell contains exactly one element,
and we assume this decision is free of error.

\subsection{Count energy and overlap}
The count energy measures the total amount by which the cell counts deviate from their mean:
\begin{equation}
T=\sum_{v=1}^M\left(C_v-\frac{kd\rho}{M}\right)^2.\label{eq:T}
\end{equation}
The centering term $kd\rho/M$ is exactly the mean of all counts. Each differing element adds its own sign to each of the $k$
cells it selects; if the cells selected by different elements do not overlap, each element contributes exactly $k$ to the sum of squared counts $\sum_vC_v^2$, and the centering subtracts $(kd\rho)^2/M$, a term fixed by the sign composition. When two elements share a cell, the contributions of same-sign elements add up on the shared cell
and the energy is raised; the contributions of opposite-sign elements cancel and the energy is lowered.

Write $H_{ij}=|S_i\cap S_j|$ for the number of cells shared by elements $i$ and $j$. This defines the standardized
measurement
\begin{equation}
\widehat d=\frac{T}{k(1-k/M)},\qquad
X=\frac{\widehat d-d}{s},\qquad
s^2=\frac{2d(d-1)}{M-1}.\label{eq:measurement}
\end{equation}
For $M>k$ and $d\ge2$, regardless of the sign composition, $\E X=0$ and $\Var X=1$; this normalization follows exactly from
the hypergeometric moments of $H_{ij}$, and the derivation is given in Appendix~\ref{app:normalization}. The fluctuation of $\widehat d$
comes entirely from the overlap between elements, and the overlap is determined by the mapping, which is precisely why the measurement and the success or failure of peeling are correlated.

\subsection{Separation of signs and mapping}\label{sec:sign}
Peeling looks only at the unsigned mapping, while the measurement looks at both the mapping and the signs. How much of the
selection of the mapping by failure passes to the measurement depends on whether the overlap falls between same-sign or
opposite-sign pairs. This subsection first answers this question exactly at the level of the mean, and then gives the
decomposition at the level of the distribution.

\paragraph{The mean.}
To understand how failure shifts the mean, it suffices to consider two things: how much the overlap of the mappings selected by
failure differs from the unconditional case, and whether this overlap falls between same-sign or opposite-sign pairs. Define
\begin{equation*}
c_F=\frac{\E[H_{12}\mid F]-k^2/M}{k(1-k/M)},
\end{equation*}
where $k^2/M$ is the expected number of overlapping cells of two elements without conditioning. The quantity $c_F$ measures the excess overlap after selection by the failure event relative to the unconditional case; it can be negative, and Section~\ref{sec:flip} shows that its sign inside the critical window depends on $k$.
Element pairs split into two classes by sign: the overlap of a same-sign pair ($s_is_j=+1$) raises the energy, and the overlap of an opposite-sign pair ($s_is_j=-1$) lowers it. Using the condition that the numbers of positive and negative signs among the $d$
signs are fixed, the number of same-sign pairs minus the number of opposite-sign pairs equals
$\binom{d_+}2+\binom{d_-}2-d_+d_-=(d^2\rho^2-d)/2$. By exchangeability
one obtains exactly
\begin{equation}
\frac{\E[\widehat d\mid F]}d-1=c_F(d\rho^2-1).\label{eq:projection}
\end{equation}
The derivation uses only two properties of $F$: it depends only on the unsigned mapping, and it is invariant under element-label permutation. Hence \eqref{eq:projection} holds verbatim with $F$ replaced by any event of positive probability that has these two properties.

Equation~\eqref{eq:projection} decomposes the shift of the estimate by failure into two factors. The first factor $c_F$ depends only on the mapping and describes the effect of selection on overlap;
the second factor $d\rho^2-1$ depends only on the sign composition and is the difference between the same-sign and opposite-sign pairs. Under one-sided signs almost all element pairs are same-sign,
so more overlap directly raises the estimate; under balanced signs the same-sign and opposite-sign pairs are roughly equal in number, and the changes in overlap cancel in the mean.

\paragraph{The distribution.}
The separation at the level of the mean also holds at the level of the distribution: the measurement $X$ decomposes into a component that varies with the mapping and a component that depends only on the assignment of signs, the latter being asymptotically independent of the mapping.

To state this decomposition, split the generation of the model into two steps. The first step generates the unsigned Plain hypergraph $\mathcal G$ with element labels:
recording only which cells each element selects and assigning no signs. The second step, independent of $\mathcal G$, assigns the given numbers of positive and negative signs to the elements uniformly.
The element mapping is exchangeable, so for any statistic invariant under element-label permutation (including $F$ and
the precritical score of Section~\ref{sec:window}), its joint law with the measurement under this two-step generation
equals its joint law under the original model, where the sign of each label is fixed in advance.

Define two quantities in this framework. Given the mapping $\mathcal G$, change all $d$ signs to the same sign (all $+1$), and write $X_{\rm one}$ for the measurement computed from it; it is a function of $\mathcal G$ and is determined once the mapping is fixed. Now let
\begin{equation*}
\gamma=\frac{d\rho^2-1}{d-1}=\frac{2\sum_{i<j}s_is_j}{d(d-1)},
\end{equation*}
that is, the number of same-sign pairs minus the number of opposite-sign pairs, divided by the total number of element pairs. The quantity $\rho$ measures the average sign of a single element, and $\gamma$ measures the net same-sign proportion of element pairs: when all signs are the same every pair is same-sign and $\gamma=1$; when the two sides each contribute half, the same-sign and opposite-sign pairs are comparable and $\gamma$ is near zero. For $d$ sufficiently large $\gamma\approx\rho^2$, and in general $\gamma\to\rho_\infty^2$.

\begin{theorem}[Graph-conditional sign CLT]\label{thm:sign}
Fix $k\ge3$. Suppose the Plain model satisfies $d/M\to\alpha\in(0,\infty)$, the fixed numbers of positive and negative signs satisfy
$\rho\to\rho_\infty\in[-1,1]$, and the signs are generated by the two steps above. Then
\begin{equation*}
X=\gamma X_{\rm one}+V,\qquad
\E[e^{iuV}\mid\mathcal G]\xrightarrow{p}\exp[-u^2(1-\rho_\infty^4)/2]\quad(u\in\mathbb R).
\end{equation*}
Under this representation, $V$ is asymptotically independent of any graph-measurable vector that has a joint limit. Returning to the original model,
this joint law transfers to graph statistics that are invariant under element-label permutation.
\end{theorem}

Intuitively, $\gamma X_{\rm one}$ is the component of the measurement that moves together with the mapping: the overlap of same-sign pairs raises the energy,
the overlap of opposite-sign pairs lowers it, and the difference $\gamma$ between the two classes of pairs sets the size of this component. The term $V$ is the remaining component, coming from the way signs are assigned on a fixed mapping:
even with the mapping fixed, reassigning the given positive and negative signs to the elements makes the energy fluctuate. The theorem states that $V$ is asymptotically normal conditionally on the mapping and asymptotically independent of the mapping. Since
$\gamma\to\rho_\infty^2$, the limiting correlation of the measurement with any permutation-invariant graph statistic that has a joint limit
with $X_{\rm one}$ and a uniformly integrable square is $\rho_\infty^2$ times its value in the all-same-sign case.

Figure~\ref{fig:decomp} illustrates this decomposition by simulation. Take $k=3$, $M=16384$, $d=13314$, draw 3000 Plain hypergraphs independently,
and on each graph compute $X_{\rm one}$ and $X$ for six levels of sign imbalance. After binning by $X_{\rm one}$, the mean of $X$ falls on the line
$\gamma X_{\rm one}$, whose slope drops from $1$ to $0$ with $\rho^2$. Take 5 further graphs with fixed mappings and reassign signs 3000 times on each.
The resulting vertical bars are the distribution of $V$ given the mapping: within the same bin the five bars have equal length, their centers all lie on the line, and their width agrees with the theoretical interval from $\sqrt{1-\rho^4}$,
showing that the conditional law of $V$ does not change with the mapping. The variances of the two components always sum to $1$; as $\rho$ drops from $1$ to $0$,
the fluctuation of $X$ moves gradually from the mapping-dependent component to the sign-assignment component; at $\rho=0.4$ the value $\gamma$ is already only $0.16$.
The long axis of the point cloud always looks to be along the 45-degree direction, because the variances of both $X$ and $X_{\rm one}$ equal $1$; the slope $\gamma$
given by the theorem is the slope of the conditional mean, captured by the binned means in the figure.
\begin{figure}[!htbp]
\centering\includegraphics[width=\linewidth]{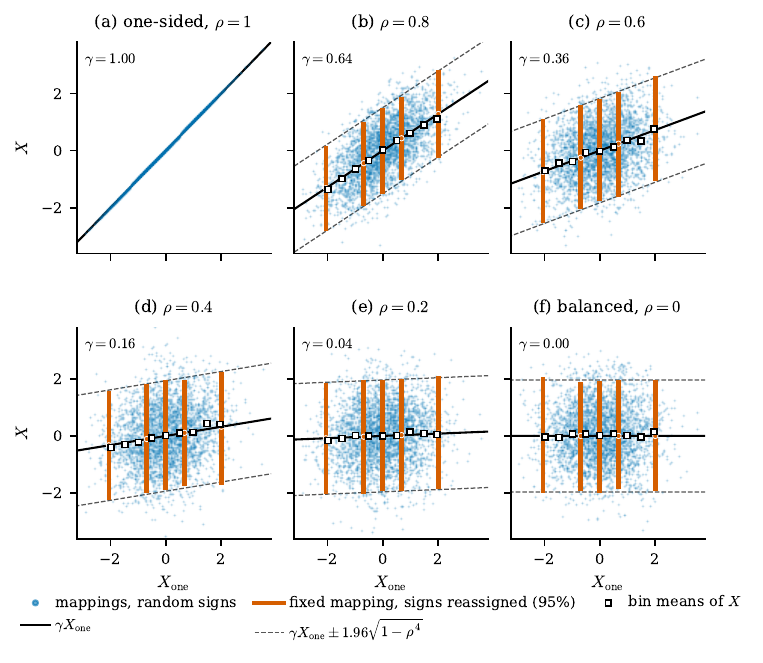}
\caption{Decomposition of the graph-conditional sign CLT ($k=3$, $M=16384$, 3000 Plain hypergraphs). Blue points: one random sign assignment per graph; squares: binned means;
orange vertical bars: the 95\% intervals of 5 fixed graphs with signs reassigned on each; dashed lines: $\gamma X_{\rm one}\pm1.96\sqrt{1-\rho^4}$.}
\label{fig:decomp}
\end{figure}

The technical difficulty of the theorem is that fixing the sign composition makes the signs exchangeable but not independent, while the overlap graph
is random, sparse, and irregular, so the existing quadratic form CLTs do not apply directly (Section~\ref{sec:related}).
The proof is in Appendix~\ref{sec:signproof}.

\begin{corollary}[Invariance under balanced selection]\label{cor:balanced}
Under the conditions of Theorem~\ref{thm:sign} let $\rho\to0$. Let the event $\mathcal S$ depend only on the unsigned mapping
and be invariant under element-label permutation. If $\liminf_M\Pr(\mathcal S)>0$, then
\begin{equation*}
\sup_x\left|\Pr(X\le x\mid \mathcal S)-\Phi(x)\right|\longrightarrow0.
\end{equation*}
\end{corollary}
Here $\Phi$ and $\phi$ are the standard normal distribution function and density. Under balanced signs $\gamma\to0$,
the mapping-dependent component $\gamma X_{\rm one}$ of the measurement vanishes, and the remaining $V$ is asymptotically independent of the graph,
so any selection that looks only at the unsigned graph cannot affect it; Figure~\ref{fig:decomp}(f) shows exactly this case. The requirement of a positive lower bound on the probability comes from the precision of asymptotic independence: the dependence of $V$ on $\mathcal G$ vanishes at rate $o(1)$,
and conditioning on an event of vanishing probability can amplify this residual. The sparse-end counterexample of Proposition~\ref{prop:axis} is exactly this situation.
The corollary holds for any positive load limit; it needs no critical-peeling analysis and imposes no rate on $\rho\to0$.
The proof is in Appendix~\ref{sec:signproof}.

Theorem~\ref{thm:sign} hands the role of the sign composition entirely to the factor $\gamma$. The next section then only needs to study the relation between the all-same-sign
$X_{\rm one}$ and the success or failure of peeling, from which the conclusion for general signs follows.

\section{The failure-conditioned law in the critical window}\label{sec:window}
This section studies the regime in which the load is close to the peeling threshold and success and failure have comparable
probability. Section~\ref{sec:sign} has already handed the role of the sign composition to the factor $\gamma$; what remains is to
answer how the failure event selects the mapping in the all-same-sign case. Section~\ref{sec:score} introduces a precritical score
that predicts the success or failure of peeling. Section~\ref{sec:thm} states the main theorem: measurement and score are
asymptotically joint Gaussian, failure is asymptotically equivalent to the score being below zero, and the failure-conditioned law is
therefore the restriction of a bivariate normal to a half-plane. Section~\ref{sec:mean} reads off the conditional mean, variance, and
quantile from this truncated distribution, and Section~\ref{sec:flip} explains why the sign of the correlation coefficient changes
between $k=4$ and $k=5$.

\subsection{The peeling bottleneck and the precritical score}\label{sec:score}
Fix an integer $k\ge3$ and let $x_c$ be the unique positive root of $e^{x}-1=(k-1)x$.
The critical load is\cite{molloy,dn,jl07}
\begin{equation*}
\alpha_c=\frac{x_c}{k(1-e^{-x_c})^{k-1}}.
\end{equation*}
Closed forms for the remaining critical constants are collected in Appendix~\ref{app:constants}. Consider a deterministic sequence
\begin{equation}
\frac dM=\alpha_c+\frac r{\sqrt M}+o(M^{-1/2}),\qquad
r\in\mathbb R,\qquad\rho\longrightarrow\rho_\infty\in[-1,1].\label{eq:window}
\end{equation}
Here $r$ is fixed, and the limiting failure probability lies in $(0,1)$.

When the load is near the critical value, the peeling process passes through a bottleneck. Each step of peeling removes an edge
incident to a degree-one vertex and changes the degrees of the neighboring vertices, possibly creating new degree-one vertices or
destroying existing ones. When the load is slightly below $\alpha_c$, the density of degree-one vertices first drops close to zero and
then recovers; when the load is slightly above $\alpha_c$, the fluid degree-one density dips below zero at its minimum, so the actual number of degree-one vertices reaches zero first and peeling halts. This minimum is
the bottleneck, and whether it can be crossed determines success or failure.

The precritical score $R$ reads the degree histogram before the bottleneck is reached and extrapolates how many degree-one vertices
remain at the bottleneck. Concretely, $R$ is a linear functional of the degree-one density of the corresponding profile at the
bottleneck, and its value depends only on the degree changes along the segment of the peeling path before the bottleneck is reached. A
negative remainder means that peeling will halt at the bottleneck. The window coordinate $r$ shifts the center of this remainder: for
$r>0$ (load above the threshold) the center is negative and the failure probability exceeds $1/2$; for $r<0$ the center is positive and the failure probability
is below $1/2$. The randomness of the mapping gives the remainder Gaussian fluctuations. The full construction of $R$ is in
Appendix~\ref{sec:mp-traj}; the main text only needs that it depends only on the unsigned mapping, is invariant under element-label
permutation, and its sign asymptotically determines the success or failure of peeling.

The correlation between measurement and score comes from the fluctuations of the cell degrees in the initial mapping. A mapping with
more overlap on the one hand changes the count energy (Section~\ref{sec:sign}), and on the other makes some cells have higher degree
and others lower degree, changing the degree profile and hence the remainder of peeling at the bottleneck. Both statistics are
functions of the initial degrees, and their correlation comes from this common source.

\subsection{Main theorem}\label{sec:thm}
\begin{theorem}[critical window joint limit]\label{thm:main}
Under the Plain model and the scaling~\eqref{eq:window}, there exists a precritical score $R$ that does not depend on the signs and is
invariant under element-label permutation. Let $\beta_k<0$, $\sigma_k>0$, and $a_c$ be the constants from
Appendix~\ref{app:constants}, and set
\begin{equation}
\eta_k=\frac{e^{-x_c}a_c(x_c-2)}{\sqrt{2\sigma_k^2}}.\label{eq:eta}
\end{equation}
\begin{enumerate}
\item[(i)] All-same-sign case:
\begin{align}
\left(X_{\rm one},\frac{R-\beta_k r}{\sigma_k}\right)
&\Rightarrow N\!\left(0,\begin{pmatrix}1&\eta_k\\\eta_k&1\end{pmatrix}\right),\label{eq:joint}\\
\Pr\bigl(F\triangle\{R\le0\}\bigr)&\longrightarrow0.\label{eq:transfer}
\end{align}
\item[(ii)] General sign composition: $(X,(R-\beta_kr)/\sigma_k)$ is asymptotically standard bivariate normal with correlation
$c=\eta_k\rho_\infty^2$.
\item[(iii)] Let $z_r=-\beta_k r/\sigma_k$. Then
\begin{gather*}
\Pr(F)\longrightarrow\Phi(z_r),\\
\sup_{x\in\mathbb R}\left|\Pr(X\le x\mid
F)-\frac{\Phi_2(x,z_r;c)}{\Phi(z_r)}\right|\longrightarrow0,
\end{gather*}
where $\Phi_2(\cdot,\cdot;c)$ is the standard bivariate normal distribution function with correlation $c$.
\end{enumerate}
\end{theorem}

The core of the theorem is (i), whose proof sketch is in Appendix~\ref{sec:proof-sketch} and full proof in
Appendix~\ref{sec:main-proof}. The correlation between measurement and score comes entirely from the initial mapping; the
fluctuations that the peeling process itself contributes to the score are asymptotically independent of measurement and only dilute the
correlation, so $|\eta_k|<1$. (ii) follows directly from (i) and Theorem~\ref{thm:sign}: $X=\gamma X_{\rm one}+V$, with $V$
asymptotically independent of $(X_{\rm one},R)$ and $\gamma\to\rho_\infty^2$, so the correlation between $X$ and the score is
$\eta_k\rho_\infty^2$. (iii) requires both \eqref{eq:joint} and \eqref{eq:transfer}. The former fixes the joint fluctuations of
measurement and score; the latter guarantees that the event that the score is below zero indeed describes the true peeling failure.
Together they give $\Pr(X\le x,F)\to\Phi_2(x,z_r;c)$, and dividing by the nonzero limiting failure probability yields the conditional
distribution. Since the limiting distribution is continuous, pointwise convergence can be lifted to uniform convergence of the
distribution functions.

Geometrically, the failure-conditioned law is obtained by restricting the limiting bivariate normal to the half-plane
$\{(x,y):y\le z_r\}$ and then looking at the first coordinate, where $y$ is the standardized score. For $c<0$, samples with low score
tend to have high measurement, and the selection shifts measurement to the right overall; for $c>0$ measurement shifts left overall;
for $c=0$ the selection has no effect. The smaller $z_r$ is, the farther the retained half-plane is from the center, the lower the
failure probability, and the more pronounced the shift. Figure~\ref{fig:joint} illustrates this truncation: (a) shows the limiting
density cut by the half-plane, and (b) gives a finite-size instance using the same batch of samples as Figure~\ref{fig:demo} in the
introduction. measurement and score are negatively correlated, the failing samples concentrate on the side where the score is below
zero, and their measurement is therefore high; this is the origin of the rightward shift of the failing samples in
Figure~\ref{fig:demo}(a).
\begin{figure}[!htbp]
\centering\includegraphics[width=\linewidth]{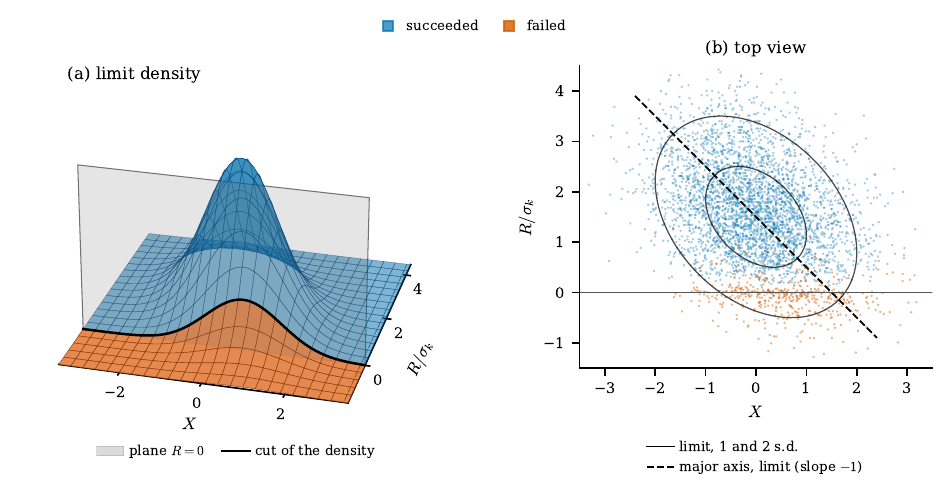}
\caption{Joint distribution of measurement and precritical score in the all-same-sign case ($k=3$, $M=16384$,
failure rate $0.125$). (a) Density of the limiting joint Gaussian; the gray plane is $R=0$, the orange part in front of the plane
corresponds to failure, and the black line is the cut;
(b) Top view in the same coordinates; each point is one Plain hypergraph, the ellipses are the one- and two-standard-deviation
contours of the limiting distribution, the dashed line is its major axis (slope $-1$), and the major axis of the sample covariance has
slope $-1.03$.}
\label{fig:joint}
\end{figure}

\subsection{Conditional mean, variance, and quantile}\label{sec:mean}
The protocol uses three numerical features of the truncated distribution. The conditional mean and the conditional variance
describe the direction and the dispersion of the shift of the estimate, and the lower quantile directly determines the threshold of
the second-round capacity. This subsection derives their expressions one by one from the limiting distribution of
Theorem~\ref{thm:main}.

\paragraph{The first two moments.}
Theorem~\ref{thm:main} gives uniform convergence of the distribution functions. Lifting this to convergence of moments requires
uniform integrability. Since $X\Rightarrow N(0,1)$ and $\E X^2=1$ holds exactly for every $M$, $X^2$ is uniformly integrable;
together with the positive limiting failure probability $\Phi(z_r)$, this gives the uniform integrability of $X^2$ under failure (Appendix~\ref{app:normalization}), and the first two conditional
moments converge to the corresponding values of the limiting distribution.

The limiting distribution is the distribution of the first coordinate of a bivariate normal truncated to the half-plane
$\{y\le z_r\}$; its moments can be expressed through the truncated moments of a standard normal. Let $A(z)=\phi(z)/\Phi(z)$ denote
the inverse Mills ratio. The mean of a standard normal variable truncated to $\{y\le z\}$ is $-A(z)$, and its variance is
$1-A(z)[z+A(z)]$. In the limiting bivariate normal, the first coordinate decomposes into $c$ times the standardized score plus a
normal term of variance $1-c^2$ that is independent of the score. The truncation acts only on the score, leaving the independent
term unaffected. Hence
\begin{equation*}
\E[X\mid F]\longrightarrow-cA(z_r),\qquad
\Var(X\mid F)\longrightarrow1-c^2A(z_r)\bigl[z_r+A(z_r)\bigr].
\end{equation*}
The shift of the mean equals the correlation coefficient $c$ times the average downward shift $A(z_r)$ of the score after
truncation: the part of the measurement proportional to the score moves together with the score, rightward when $c<0$ and leftward
when $c>0$. The conditional variance is less than one, because truncation compresses the fluctuations along the score direction
while the part independent of the score is preserved unchanged. Under balanced signs $c=0$, both return to their unconditional
values, consistent with Corollary~\ref{cor:balanced}. The smaller $z_r$ is (that is, the lower the load), the lower the failure
probability, the larger $A(z_r)$, and the more pronounced the shift.

\paragraph{Quantile.}
Let $q_{\rm sel}(u;z,c)$ denote the quantile at level $u\in(0,1)$ of the distribution function $\Phi_2(x,z;c)/\Phi(z)$. This
distribution has density
\begin{equation*}
f_{\rm sel}(x;z,c)=\frac{\phi(x)}{\Phi(z)}
\Phi\!\left(\frac{z-cx}{\sqrt{1-c^2}}\right)
\end{equation*}
which is positive everywhere, and the distribution function is strictly increasing, so the conditional quantile at every fixed level
converges to $q_{\rm sel}$.

Which side the quantile moves toward is likewise determined by the sign of $c$. Since
$\partial_c\Phi_2(x,z;c)=\phi_2(x,z;c)>0$, at every fixed $x$ the conditional distribution function increases strictly in $c$; at
$c=0$ it is the standard normal distribution function. Lifting the distribution function overall means that the quantile at every
level shifts leftward. Thus when $c>0$ all limiting quantiles lie below the corresponding standard normal values, and when $c<0$
they all lie above.

Thus far the failure-conditioned law inside the critical window is determined by only two numbers, $z_r$ and
$c=\eta_k\rho_\infty^2$. The sign of $\eta_k$ determines the direction in which the mean and the quantiles shift, while the variance decreases for either sign, by an amount that depends on $c^2$.

\subsection{Sign change of $\eta_k$ and protocol risk}\label{sec:flip}
By \eqref{eq:eta}, all factors of $\eta_k$ other than $x_c-2$ are positive, so the sign of $\eta_k$
equals the sign of $x_c-2$. The function $e^x-1-(k-1)x$ is negative on $(0,x_c)$ and positive on
$(x_c,\infty)$, so it suffices to look at its sign at $x=2$. $e^2-1\approx6.39$, while $(k-1)\cdot2$
equals $6$ for $k=4$ and $8$ for $k=5$. The former gives a positive value, hence $x_c<2$; the latter
gives a negative value, hence $x_c>2$. Since $x_c$ increases with $k$, we have $\eta_k<0$ for
$k=3,4$, while every fixed $k\ge5$ has $\eta_k>0$.
Table~\ref{tab:constants} lists the numerical values.

\begin{table}[!htbp]
\centering
\caption{Analytic constants and the direction of the conditional mean under one-sided failure.}
\begin{tabular}{rrrl}
\toprule
$k$ & $\sigma_k^2$ & $\eta_k$ & direction of the mean ($\rho\ne0$)\\
\midrule
3 & 0.047334402 & $-0.352023542$ & positive\\
4 & 0.022604525 & $-0.041540768$ & positive\\
5 & 0.014632421 & $+0.126662297$ & negative\\
6 & 0.010718241 & $+0.219499658$ & negative\\
\bottomrule
\end{tabular}
\label{tab:constants}
\end{table}

This sign can be translated back to the original scale. Under a one-sided sign, $\rho=1$, and
\eqref{eq:projection} gives $\E[\widehat d\mid F]-d=d(d-1)c_F$, hence $\E[X\mid F]=d(d-1)c_F/s$.
Using $s^2=2d(d-1)/(M-1)$ and $d/M\to\alpha_c$, we obtain
\begin{equation*}
M^{3/2}c_F\longrightarrow
-\frac{\sqrt2\eta_k}{\alpha_c}A(z_r).
\end{equation*}
The sign change is a property of $c_F$ itself inside the critical window, independent of the sign
composition; the sign composition only determines the proportion in which this shift is transmitted
to the estimate. Returning to the interpretation of $c_F$ in Section~\ref{sec:sign}: for $k=3,4$ the
mapping overlap selected by failure is larger than average and the estimate is biased upward; for
$k\ge5$ the overlap is smaller and the estimate is biased downward. For $k=3,4,5,6$, $x_c$ is about
$1.26$, $1.90$, $2.34$, $2.66$. The value of $x_c$ for $k=4$ is close to $2$, so $|\eta_4|$ is only
$0.04$; the selection effect is weakest at $k=4$, and markedly stronger at $k=3$ and for $k\ge5$.
Figure~\ref{fig:selection} gives the simulation comparison for $k=3$--$6$. The mean has another sign
change: for fixed $k\ge5$, as the load decreases gradually the conditional mean changes sign once
more, see Appendix~\ref{app:wm}.

The sign change directly determines on which side the protocol risk falls. Under the rule that takes
$\widehat d/q_u$ as the capacity of the second round, $\widehat d<q_u d$ is equivalent to the
capacity set from the estimate being smaller than the true number of differences, and the lower-tail
probability under failure is the probability of this event. Under a one-sided sign, if $c>0$
($k\ge5$) the true quantile after failure is lower than the quantile of the unconditional normal;
using the latter as the lower-tail threshold makes the lower-tail probability higher than the nominal
value and the capacity too small; if $c<0$ ($k=3,4$) the bias falls on the conservative side. Whether
the second round ultimately succeeds also depends on the capacity conversion factor and the failure
probability of the second-round decoding itself; the coverage experiment of Section~\ref{sec:e2e}
checks the non-coverage probability of the upper bound after failure, i.e.\ the probability that $U<d$.

The analysis above gives the full conditional distribution when $z_r$ and $c=\eta_k\rho_\infty^2$ are
known. In practice there is a further difficulty: both $z_r$ and $c$ are computed from the true $d$,
and $d$ is precisely the quantity the protocol is meant to estimate. The centering term of the energy
uses only the total sign $b=k^{-1}\sum_vC_v=d\rho$, which the protocol can compute directly.
Section~\ref{sec:upper} uses the observable $(\widehat d,b)$ to perform candidate testing and
inversion, removing the dependence on the true $d$.

\section{Upper confidence bound after failure}\label{sec:upper}
The selection law of Section~\ref{sec:window} gives the limiting distribution of the measurement under
failure, but its parameters ($z_r$ and $c$) both depend on the unknown $d$. After failure the protocol can
observe only three quantities: the failure event $F$, the estimate $\widehat d$, and the total sign
$b = d_+ - d_-$. The strategy of this section is to substitute each candidate value of $d$ into the
conditional law, perform a hypothesis test, and invert it, obtaining an upper confidence bound for $d$.
This construction requires only a null distribution for each candidate, so it can accommodate conditional
laws from different regimes: the selection law inside the critical window, and a mixture of the two
failure mechanisms in the low-load regime.

\subsection{Candidate tests and construction of the bound}\label{sec:upper-construct}
The Plain mapping is exchangeable in the element labels, so given $(d, b)$ the joint law of
$(\widehat d, F)$ depends only on these two numbers. This means that for each candidate value $n$, the
statement that the true difference count is $n$ and the total sign is $b$ constitutes a complete simple
hypothesis. A candidate $n$ also determines the window coordinate and the correlation coefficient: the
former is given by the distance of $n/M$ from $\alpha_c$, and the latter equals $\eta_k b^2/n^2$, i.e.\ the
value of $\eta_k\rho^2$ at $\rho = b/n$. Thus $b$ together with a candidate $n$ suffices to determine the
entire conditional law, and the candidates can be tested one by one.

The candidate set is
$\mathcal N(b) = \{n : n \ge \max(|b|, 2),\ n \equiv b \pmod 2\}$. For
$n \in \mathcal N(b)$ let
\begin{gather*}
s(n)^2 = \frac{2n(n-1)}{M-1},\qquad
z(n) = \frac{ka_ce^{-x_c}}{\sigma_k}\Bigl(\frac nM -
\alpha_c\Bigr)\sqrt M,\\
t(n,b) = n + s(n)\,q_{\rm sel}\bigl(u;\,z(n),\,\eta_k b^2/n^2\bigr),
\end{gather*}
where $q_{\rm sel}$ is the selection law quantile defined in Section~\ref{sec:mean}.
$t(n,b)$ is the acceptance threshold for candidate $n$: if the true difference count is indeed $n$, then
under failure the probability that $\widehat d$ lies below $t(n,b)$ is asymptotically $u$. If the observed
$\widehat d$ lies below the threshold, the candidate is rejected.

The acceptance set is $\mathcal C = \{n \in \mathcal N(b) : \widehat d \ge t(n,b)\}$, and the upper
confidence bound is
\begin{equation*}
U = \sup\,\mathcal C;
\end{equation*}
if $\mathcal C$ is empty, set $U = +\infty$ and the protocol uses full
transmission or doubling. Two properties of this construction are worth noting. First, the coverage
probability depends only on whether the true candidate $d$ is wrongly rejected: if $d \in \mathcal C$ then
$U \ge d$, and if $\mathcal C$ is empty then $U = +\infty$, so in both cases the bound covers the true
value. Hence neither monotonicity of the threshold in $n$ nor multiple-testing correction is needed; the
thresholds of the other candidates affect only how wide the bound is, not the coverage probability.
Second, for $M$ sufficiently large the acceptance set lies in a finite interval (Appendix~\ref{app:upper}),
so $U$ can be computed with finitely many comparisons.

\subsection{Coverage inside the critical window}\label{sec:upper-cover}
\begin{theorem}[Asymptotically conservative coverage]\label{thm:upper}
Fix $u \in (0,1)$.
\begin{enumerate}[label=(\roman*)]
\item Under the conditions of Theorem~\ref{thm:main}, $\Pr(\widehat d < t(d,b) \mid F) \to u$,
so $\liminf_M \Pr(d \le U \mid F) \ge 1-u$.
\item For every $R < \infty$, the convergence in (i) is uniform in $|r| \le R$ and over all admissible $b$.
\item If $d/M$ is bounded and $r \to +\infty$, then
$\liminf_M \Pr(d \le U \mid F) \ge 1-u$.
\end{enumerate}
\end{theorem}
Under a completely one-sided sign, $d = |b|$ is the smallest candidate and the coverage is always one, so
the theorem is stated uniformly as ``$\ge 1-u$''. The proof is in Appendix~\ref{app:upper}.

Item (iii) handles the side where the load is above the critical window. There failure becomes increasingly
certain, and the selection effect weakens accordingly: the center of the score moves to a deeply negative
value, truncation barely changes the marginal distribution, and the conditional law tends to the
unconditional normal. At the same time, as $z(n) \to +\infty$ the selection law quantile tends to the
normal quantile, and the two coincide automatically. When the load is below the critical window, the source
of failure changes, and a new conditional law is required.

\subsection{Two failure mechanisms outside the window}\label{sec:loadaxis}
Below the critical value, the fluid trajectory of peeling leaves a positive gap at the bottleneck, so that forming a
macroscopic core requires a large fluctuation, whose probability decays rapidly as the load decreases. At the same time
there is another failure mechanism almost independent of the load position: two elements choose exactly the same
$k$-subset. The two edges block each other, each endpoint has degree at least two, and peeling can never remove them,
so they constitute a core consisting of exactly two edges. Its probability is about $\lambda = \binom{d}{2}/\binom{M}{k}$,
which is small in the sparse low-load region, but the probability of a macroscopic failure decays even faster and is
eventually overtaken by it. At a fixed subcritical load, failure is almost always caused by two repeated edges
(Proposition~\ref{prop:lowload} of Appendix~\ref{app:lowload}).

The effect of a repeated edge on the measurement can be computed directly. If the two repeated edges have the same
sign, their counts on the $k$ shared cells add up, the energy exceeds that of two non-overlapping edges by about $2k$,
and the estimate is shifted by about $+2$; if they have opposite signs, the counts cancel and the estimate is shifted by
about $-2$. Conditioned on the presence of a repeated edge, the probability that the two edges have the same sign is
$(1+\gamma)/2$, where $\gamma$ is the sign factor of Section~\ref{sec:sign}. On average the estimate is shifted by
$2\gamma$. This shift is constant on the original scale, while the standard deviation $s$ grows with the system size, so
the standardized measurement at the repeated-edge layer tends again to $N(0,1)$.

Between the critical window and a fixed subcritical load lies the moderate deviation region ($r \to -\infty$ and
$|r| \le C\sqrt{\log M}$), where the two mechanisms coexist. Proposition~\ref{prop:mixture}
(Appendix~\ref{app:mixture}) proves that failures caused by a macroscopic core still obey the selection law, and the law
of the measurement conditioned on failure is a weighted mixture of the macroscopic layer and the repeated-edge layer,
with the two weights proportional to $\Phi(z_r)$ and $\lambda$, respectively. As the load moves down from the critical
window, $\Phi(z_r)$ decays like a normal tail while $\lambda$ stays almost unchanged, so the weight passes smoothly from
the macroscopic layer to the repeated-edge layer.

Replacing the null distribution in the candidate test by this mixture law gives an upper bound covering the entire load
axis. For a candidate $n$ write $z_n = z(n)$, $c_n = \eta_k b^2/n^2$,
$\lambda_n = \binom{n}{2}/\binom{M}{k}$, and $w_n = \Phi(z_n)/(\Phi(z_n) + \lambda_n)$, and set
\begin{equation*}
H_n(x) = w_n\frac{\Phi_2(x,z_n;c_n)}{\Phi(z_n)} + (1-w_n)\Phi(x),\qquad
T_n = n + s(n)\,H_n^{-1}(u).
\end{equation*}
Here $H_n$ is the mixture CDF of the macroscopic layer (weight $w_n$, selection law) and the repeated-edge layer
(weight $1-w_n$, standard normal). Replacing $t(n,b)$ of Section~\ref{sec:upper-construct} by $T_n$ and leaving the
rest of the construction unchanged gives an upper bound, still denoted by $U$.

\begin{proposition}[Load-axis coverage]\label{prop:axis}
Fix $k \ge 3$ and $u \in (0,1)$. If $d/\sqrt{M} \to \infty$, then, whatever the limit of $d/M$ is, whichever side of
the critical point the load lies on, and whether or not the sign composition converges, we have
\begin{equation*}
\sup_x\bigl|\Pr(X \le x \mid F) - H_d(x)\bigr| \to 0,\qquad
\liminf_M \Pr(d \le U \mid F) \ge 1-u.
\end{equation*}
The condition $d/\sqrt{M} \to \infty$ cannot be dropped. Take balanced signs and an even $d \to \infty$ with
$d = o(\sqrt{M})$; then for $u \in (0,1/2)$ we have $\Pr(U < d \mid F) \to 1/2$. For fixed $d = 4$ with two edges of
each sign, this probability tends to $2/3$.
\end{proposition}
The counterexample comes from the $\pm 2$ shift of repeated edges. The order of magnitude of $s$ is $d/\sqrt{M}$, so
when $d = o(\sqrt{M})$ the shift $2/s$ does not vanish on the standardized scale; opposite-sign repeated edges push the
estimate too low, and the true value is then rejected. The proof and the construction of the counterexample are given in
Appendix~\ref{app:axis}.

\subsection{Implementation rules at finite size}\label{sec:rules}
Theorem~\ref{thm:upper} and Proposition~\ref{prop:axis} are asymptotic statements. At finite
size two systematic biases require correction, and both admit analytic expressions without
fitted parameters.

The first is the threshold center shift. The critical threshold shift given by
Dembo--Montanari\cite{dm}, rescaled to the scale of this paper, replaces $z(n)$ with
\begin{equation*}
z_{\rm sh} = z(n) + C_k M^{-1/6}.
\end{equation*}
The constant $C_k$ is given in closed form by Airy integrals and the critical constants
(Appendix~\ref{app:ck}); for $k = 3, 4, 5, 6$, $C_k \approx 1.930$, $2.015$, $2.016$, $2.005$.
The empirical basis is given in Section~\ref{sec:threshold}.

The second is the chi-square shape correction. At finite size the unconditional measurement
is closer to the standardized chi-square $(M-1)$ distribution (Section~\ref{sec:chisq}), and
the quantile requires the additional term
\begin{equation*}
\kappa(u) = \frac{\chi^2_{M-1,u} - (M-1)}{\sqrt{2(M-1)}} - \Phi^{-1}(u).
\end{equation*}
Both corrections vanish as $M \to \infty$ and do not affect the asymptotic coverage.

Both corrections act on the acceptance threshold of Section~\ref{sec:upper-construct}:
replacing $z(n)$ by $z_{\rm sh}$ and adding $\kappa(u)$ to the quantile gives the corrected
selection threshold
\begin{equation}
n + s(n)\bigl[q_{\rm sel}\bigl(u;\,z_{\rm sh},\,\eta_k b^2/n^2\bigr) + \kappa(u)\bigr].
\label{eq:corrected}
\end{equation}
It describes macro-core failure. On the low-load side below the critical window, failure comes
mainly from repeated edges, and using \eqref{eq:corrected} alone gives undercoverage
(Section~\ref{sec:num-loadaxis}). The two implementation rules below handle this in different
ways; the comparison baseline is the unconditional threshold $n + s(n)\Phi^{-1}(u)$, which
applies no selection correction at all.

The envelope rule takes, for each candidate, the lower of \eqref{eq:corrected} and the
unconditional threshold. The acceptance set is the union of the acceptance sets of the two
thresholds, so a candidate is rejected only when both thresholds reject it, and the probability
that the true value is wrongly rejected is at most the probability under either threshold alone.
For $k = 3, 4$ the quantile of the selection law lies above the Gaussian quantile, so the
envelope rule takes the unconditional threshold. For $k \ge 5$ the selection correction lowers
the threshold; when the signs are far enough from balance it outweighs $\kappa(u)$ and the
envelope rule applies it, while near balance the unconditional threshold is lower.

The mixture rule replaces the selection law in \eqref{eq:corrected} with the mixture law $H_n$
of Section~\ref{sec:loadaxis}, again taking the coordinate $z_{\rm sh}$ and adding $\kappa(u)$;
the macro-layer weight is $\Phi(z_{\rm sh})/(\Phi(z_{\rm sh}) + 1 - e^{-\lambda_n})$, where
$1 - e^{-\lambda_n}$ is the Poisson approximation to the repeated-edge probability. It describes
the two failure mechanisms explicitly. Its asymptotic counterpart covers the entire load axis when
$d/\sqrt M\to\infty$ (Proposition~\ref{prop:axis}), and the finite-size rule is checked numerically. The envelope rule is
simpler in form, while the mixture rule gives a narrower upper bound on the low-load side at
$k=3$. Coverage tests of the two rules are given in Section~\ref{sec:e2e}.

\section{Numerical Evaluation}\label{sec:numerics}
The numerical evaluation has two parts. Sections~\ref{sec:constants}--\ref{sec:num-loadaxis}
verify the theory term by term under the known true load and sign composition, while
Section~\ref{sec:e2e} constructs an upper bound from only the quantities observable after
failure and checks its coverage. Table~\ref{tab:exp-overview} lists the question addressed in
each section and the corresponding theoretical result.

\begin{table}[!htbp]
\centering
\caption{Overview of the experiments.}
\label{tab:exp-overview}
\small
\begin{tabular}{lp{.5\linewidth}ll}
\toprule
Section & Question tested & Corresponding theory & Result\\
\midrule
\ref{sec:constants} & the conditional variance, the direction of the conditional mean, and the $\rho^2$ dependence of the selection effect & Theorem~\ref{thm:main} & Figure~\ref{fig:selection}\\
\ref{sec:threshold} & the shift of the threshold center at finite size & the shift constant $C_k$ & Figure~\ref{fig:threshold}\\
\ref{sec:chisq} & the shape of the unconditional measurement under balanced signs & the chi-square correction & Figure~\ref{fig:chisq}\\
\ref{sec:calib} & the accuracy of the lower-tail probability after combining the two corrections & Eq.~\eqref{eq:corrected} & Figure~\ref{fig:calib}\\
\ref{sec:num-loadaxis} & the separation of the macro layer and the repeated-edge layer at low load & Proposition~\ref{prop:lowload}, \ref{prop:mixture} & Figure~\ref{fig:twomech}\\
\ref{sec:e2e} & the coverage and width of the upper bound $U$ when $d$ is unknown & the envelope rule and the mixture rule & Table~\ref{tab:e2e}\\
\bottomrule
\end{tabular}
\end{table}

\subsection{Experimental Design and Comparison Targets}\label{sec:exp-design}
\emph{Simulation setup.} The experiments use an independent Plain mapping and ideal peeling.
The analytic constants are computed by a separate script and are not fitted to the simulation
data. The experiments below the critical load cover $k\in\{3,5\}$, $M\in\{256,1024,4096,16384,65536\}$, and
four target coordinates $\{-2,-2.5,-3,-3.5\}$, for 40 configurations in total. Because $d$
must be an integer, the actual window location is determined by the chosen integer $d$.

\emph{Three coordinates.} The first is the theorem coordinate:
\begin{equation*}
z_0=\frac{(d/M-\alpha_c)\sqrt M}{w_k},\qquad w_k=\sigma_k/|\beta_k|.
\end{equation*}
The second is the shifted coordinate $z_{\rm sh}=z_0+C_kM^{-1/6}$, with the numerical values
of the shift constant $C_k$ in Section~\ref{sec:rules} and the closed form in
Appendix~\ref{app:ck}. The result of DM concerns the configuration model;
Section~\ref{sec:threshold} verifies this shift on the Plain model. A sequence obtained by
fixing $z_{\rm sh}$ corresponds to a single limiting $r$, and therefore satisfies the
conditions of the theorem. The third is the empirical coordinate $z_{\rm emp}=\Phi^{-1}(\widehat p_F)$,
recovered from the measured failure rate and used only for diagnosis.

\emph{Classification of failures.} A run fails when the 2-core is nonempty after peeling.
According to the number of edges in the remaining core, failures are divided into macro-core
failures and small-core failures; a small core containing a pair of identical edges is
assigned to the repeated-edge layer. A macro core is identified by the core edge count
exceeding $0.1d$ (or $0.1M$ for direct sampling).

\emph{Data split.} For the calibration experiments, samples from each configuration are split by trial
index into two disjoint halves: the empirical threshold and the empirical failure rate come
from the training half, and the actual lower-tail probability is computed on the failing
samples of the test half. The calibration experiments provide 624 measurements across sign compositions
and independent runs, covering 115 distinct $(k,M,d)$ configurations.

\emph{Thresholds compared.} Sections~\ref{sec:constants}--\ref{sec:num-loadaxis} use the true
load and sign composition and compare four thresholds: the unconditional normal threshold,
the unconditional chi-square threshold, the threshold given by the selection law, and the
corrected selection threshold of Eq.~\eqref{eq:corrected} (with $n=d$ and
$b^2/n^2=\rho^2$). Both corrections come from analytic formulas and contain no fitted
coefficients; the basis for the chi-square term is given in Section~\ref{sec:chisq}.
Section~\ref{sec:e2e} compares the envelope rule and the mixture rule of
Section~\ref{sec:rules}, with the unconditional threshold as the baseline and $d$ given by the
inversion of Section~\ref{sec:upper}.

\subsection{Verification of the analytic constants}\label{sec:constants}
Each constant in Theorem~\ref{thm:main} can be verified separately. This section gives the
conditional variance and the direction of the mean; the window width, dynamic variance, and
selection strength are verified in Appendix~\ref{app:exp-constants}, and cross-validation
across independent realizations appears in Appendix~\ref{app:exp-indep}.

\emph{Conditional variance.} At $k=6$, $M=16384$, and failure probability $0.10$, the
measured value is $0.9604$, and the selection law gives $0.9600$. At $k=3$, on the
low-failure-probability grid points $M=65536$ and $262144$, the measured values are $0.8989$
and $0.8987$, and the theoretical value is $0.8971$.

\emph{Sign change.} At every failure-probability level and every $M$, the one-sided
conditional mean is negative for $k=5,6$ and positive for $k=3$. At $M=262144$, the largest
deviation between measurement and prediction is $1.4$ standard errors.

Figure~\ref{fig:selection} summarizes 444 measurements across sign compositions and independent runs
at 85 distinct $(k,M,d)$ configurations with $M\ge4096$ ($k=3$--$6$, $M=4096$--$262144$); the theoretical curves use only the analytic
constants $\eta_k$. The horizontal axis is the measured failure probability, equivalently
comparing on the empirical coordinate $z_{\rm emp}=\Phi^{-1}(\widehat p_F)$. Thus this figure
verifies the conditional distribution itself; the finite-size shift of the threshold center is
handled separately in Section~\ref{sec:threshold}. When all signs are equal, the conditional
means of all four $k$ lie on the curve $-\eta_kA(z_{\rm emp})$, and the direction reverses
between $k=4$ and $k=5$. After dividing the conditional mean by $-A(z_{\rm emp})$, all sign
compositions lie on a straight line through the origin with slope $\eta_k$, and the balanced
sign composition returns to the origin; this panel verifies both the value of $\eta_k$ and its
decay in $\rho^2$. For the all-equal-sign rows, the median relative error at $k=3$ decreases
from $4.7\%$ at $M=4096$ to $0.8\%$ at $M=262144$, and at $k=5$ it does not exceed $0.5\%$ at
any size.
\begin{figure}[!htbp]
\centering\includegraphics[width=\linewidth]{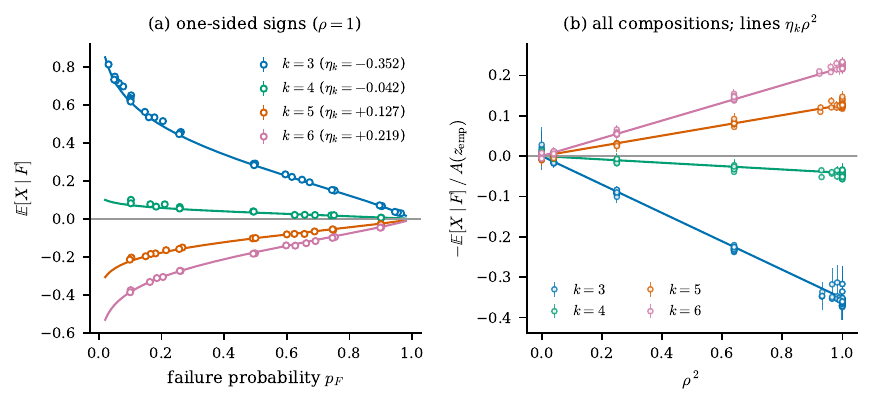}
\caption{Failure conditional mean versus the analytic prediction ($M\ge4096$). (a) all signs
equal, curve $-\eta_kA(\Phi^{-1}(p_F))$; (b) all sign compositions, line $\eta_k\rho^2$.
}
\label{fig:selection}
\end{figure}

\subsection{Threshold shift at finite size}\label{sec:threshold}
At finite $M$, the failure probability at $r=0$ is not equal to $1/2$. We use 112 central-window simulation runs across 111 distinct $(k,M,d)$ configurations ($|z_0|\le2$, $k=3$--$6$, $M=64$--$262144$) to test this shift, converting
in three steps. First we approximate the probability of duplicate edges by
$p_{\rm dup}=1-e^{-\lambda}$, where $\lambda=\binom d2/\binom Mk$. Then, assuming the
macro-core event and the duplicate-edge event are approximately independent, we convert to the
macro failure probability $p_{\rm core}=(\widehat p_F-p_{\rm dup})/(1-p_{\rm dup})$. Finally we
set $z_{\rm emp}=\Phi^{-1}(p_{\rm core})$ (Figure~\ref{fig:threshold}).

For $M\ge1024$, the measured shift $z_{\rm emp}-z_0$ is 94\%--101\% of the analytic prediction
$C_kM^{-1/6}$; including all data, the ratio is 88\%--103\%. Using only $z_0$, even at
$M=262144$ the mean absolute error of the failure probability remains 0.06--0.10; after
switching to $z_{\rm sh}$, the error does not exceed 0.008 for $M\ge4096$. For example, at
$k=3$, $r=0$, $M=65536$, the measured failure rate is 0.617 while $\Phi(z_{\rm sh})=0.619$.
After the shift, the residual for $M\ge4096$ is at most 0.06 on the $z$ scale, and the largest
residual appears at the window edge, with a direction and magnitude consistent with the known
finite-size bias of the window width.

These data support the following practice: in the tested Plain configurations, the shift
constant of the configuration model approximates the threshold center at finite size. This
conclusion is a numerical verification; whether the shift order transfers to the Plain model
still requires a separate proof. Together with Appendix~\ref{sec:convergence}, this section
separates two finite-size effects: the shift of the threshold center is captured by the
translation, while the bias of the conditional distribution itself at a fixed failure
probability is unaffected by the translation and is examined separately.
\begin{figure}[!htbp]
\centering\includegraphics[width=\linewidth]{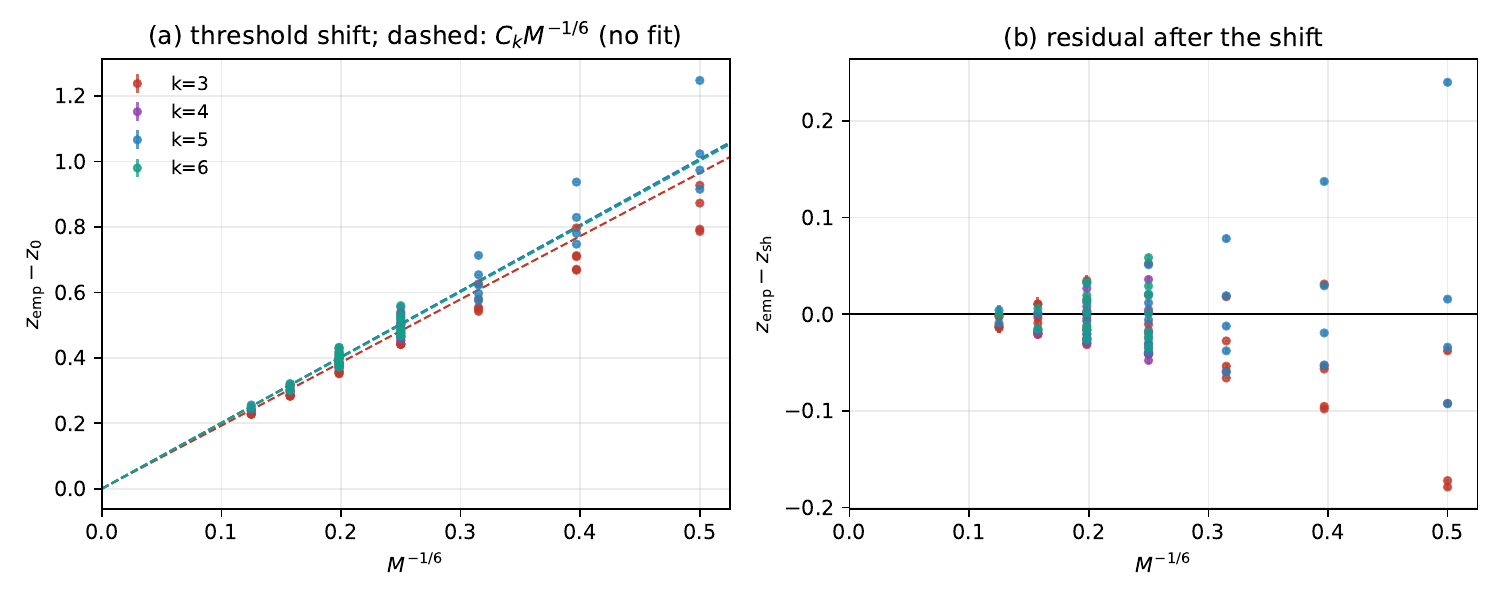}
\caption{Location of the threshold center. (a) measured shift and the analytic prediction
$C_kM^{-1/6}$ (dashed); (b) residual after the shift.}
\label{fig:threshold}
\end{figure}

\subsection{Chi-square shape of the unconditional measurement}\label{sec:chisq}
Under balanced signs, failure selection does not operate, so the shape of the unconditional
measurement at finite size can be tested separately. We compare its empirical quantile with the
quantile $[\chi^2_{M-1,u}-(M-1)]/\sqrt{2(M-1)}$ of the standardized chi-square $(M-1)$,
without fitting any parameters (Figure~\ref{fig:chisq}). For $k=3$ and $u=0.01$, at
$M=64,256,1024,4096,65536$ the measured quantiles are $-2.058$, $-2.205$, $-2.255$, $-2.293$,
$-2.319$; the chi-square gives $-2.064$, $-2.196$, $-2.261$, $-2.294$, $-2.318$; the normal
quantile is $-2.326$. The cases $k=4,5,6$ and $u=0.05$ likewise agree to about $0.01$.

Thus even at critical load, where the count of a single cell is far from normal, the shape at
finite $M$ is consistent with the chi-square inheritance obtained under deep overload in
\cite{prior}. The standardized chi-square has skewness $\sqrt{8/(M-1)}$, which vanishes as $M$
grows, consistent with the normal limit of Theorem~\ref{thm:main}. This yields the shape
correction $\kappa(u)$ of Section~\ref{sec:rules}, which contains no fitted parameters. The
chi-square result of \cite{prior} fixes $M$ and increases the load; here $M$ and $d$ grow
together. The agreement provides an empirical bridge between the two limiting orders.
\begin{figure}[!htbp]
\centering\includegraphics[width=.8\linewidth]{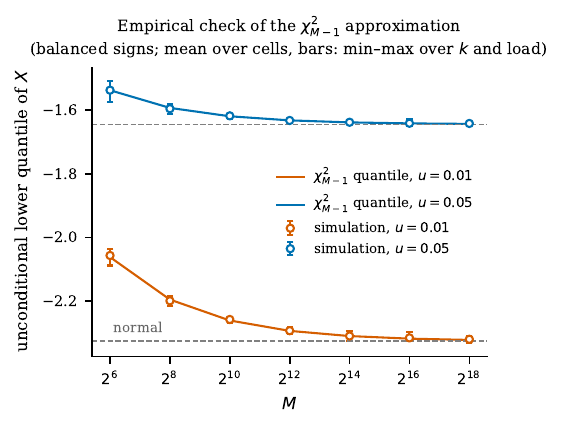}
\caption{Lower quantile of the unconditional measurement under balanced signs, and the
standardized chi-square $(M-1)$ quantile.}
\label{fig:chisq}
\end{figure}

\subsection{Lower-tail calibration in the critical window}\label{sec:calib}
Figure~\ref{fig:calib} compares the nominal lower-tail level with the actual lower-tail
probability in the failed test samples. The measured configurations have one-sided signs, lie
in the central window, and have failure probability at least about $0.03$. For these
configurations, when $M\ge256$ the additive corrections of selection and chi-square reduce the
grouped RMS relative error to about $0.04$--$0.08$; the unconditional threshold still carries
a substantial selection error.

The configurations included in each group differ across sizes, so the figure reports both the
dispersion and the number of configurations. The main theorem gives the direction and limiting
magnitude of the selection effect; this experiment gives the precision of the two analytic
corrections at finite size.
\begin{figure}[!htbp]
\centering\includegraphics[width=\linewidth]{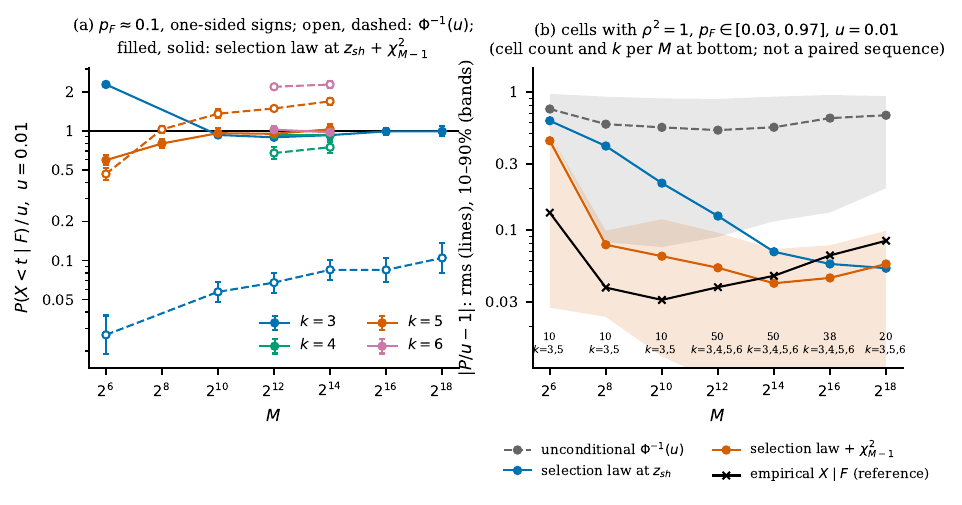}
\caption{Lower-tail calibration in the central window. (a) ratio of the actual lower-tail
probability to its nominal value when the failure probability is about $0.1$ and all signs
positive; (b) root mean square relative error at the $\rho^2=1$ grid points.}
\label{fig:calib}
\end{figure}

\subsection{Two Failure Mechanisms Along the Load Axis}\label{sec:num-loadaxis}
This set of experiments covers 32 configurations and about 1.69 million failure samples. We split failures into two
layers according to the actual core: a macroscopic layer (the core contains more than $0.1d$ edges) and a
repeated-edge layer (the core contains a pair of identical edges); the classification criterion is stated in
Section~\ref{sec:exp-design}.

The repeated-edge layer is further grouped by the sign product $\sigma$; its conditional mean equals the
unconditional mean plus $\pm2/s_M$. The largest $|z|$ among the 54 test units is $3.13$. On the macroscopic layer,
for each of the 96 test units we test whether the rejection rate of the corrected selection threshold equals $u$,
using a $99\%$ interval; no interval excludes $u$, and for $M\ge4096$, rej$/u$ lies between $0.83$ and $1.25$.
The difference between the macroscopic-layer weight computed from the analytic formula and the measured value is at
most $0.032$ for $M\ge4096$, and at most $0.009$ for $M=16384$ and $65536$. These results support the mechanistic
picture described by Proposition~\ref{prop:lowload} and Proposition~\ref{prop:mixture}.

Figure~\ref{fig:twomech} shows the layering results for $k=3$ with one-sided signs. As $z_{\rm sh}$ decreases, the
weight of the macroscopic layer passes from 1 to 0, and the transition location deepens as $M$ grows; the theoretical
curve $\Phi(z_{\rm sh})/(\Phi(z_{\rm sh})+p_{\rm dup})$ contains no adjustable parameters. The
conditional mean of all failures falls to $2/s_M$ of the repeated-edge layer on the low-load side and returns to the
mean given by the macroscopic-layer selection law on the window side.
The ratio of the measured shift of the repeated-edge layer to $2\sigma/s_M$ lies between $0.90$ and $1.05$ for every
$M$ and both values of $\sigma$.
In the region where the macroscopic layer dominates, the conditional mean exceeds the theoretical value by
15\%--60\% at $M=1024$, and drops to 3\%--8\% for $M\ge16384$; this has the same origin as the positive residual of
the macroscopic layer in Appendix~\ref{sec:convergence}.

Using only the corrected selection threshold also leads to undercoverage: in the low-load region with $k=3$, the
noncoverage under small-core failure can reach 99 times the nominal value. The two implementation rules of
Section~\ref{sec:rules} are designed precisely for this.
\begin{figure}[!htbp]
\centering\includegraphics[width=\linewidth]{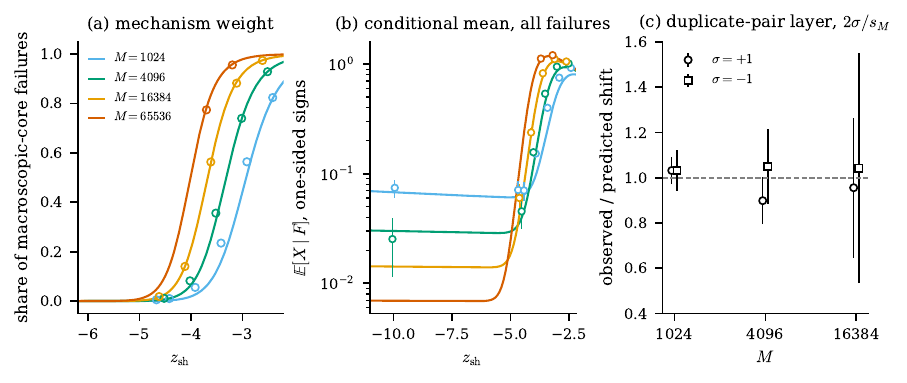}
\caption{$k=3$ low-load two-mechanism layering. (a) Fraction of macroscopic-core failures; (b) conditional
mean of all failures (all-equal sign); (c) ratio of the measured shift of the repeated-edge layer to
$2\sigma/s_M$. All curves are analytic predictions.}
\label{fig:twomech}
\end{figure}

\subsection{Coverage of the upper bound for unknown $d$}\label{sec:e2e}
The preceding sections test the conditional law under a known true load and sign composition. This section tests what the protocol can actually use: using only the post-failure observable
$\widehat d$ and the total sign $b$, construct the upper bound $U$ following Section~\ref{sec:upper}, and check whether it covers the true $d$.
What is tested here is the non-coverage probability of the post-failure upper bound, namely the probability that $U<d$, and it does not involve a second decoding round.

The post-failure coverage test comprises 288 test units, with criteria fixed before running. For each unit we estimate the non-coverage probability $\Pr(U<d\mid F)$ and reject the claim when the lower end of its $99.9\%$ interval (Wilson for direct sampling, the delta method for splitting) exceeds $u$ for the envelope rule or $1.3u$ for the mixture rule.
No unit is rejected for either rule. The largest estimated ratios to the nominal level are 1.106 and 1.151, and the lower ends of their intervals lie below $u$. Table~\ref{tab:e2e} summarizes the two rules and the
unconditional threshold by region. We also tested two further candidate rules: the asymptotic form of Proposition~\ref{prop:axis} without a
finite-size correction, whose maximum reaches $1.2$ times the nominal value, and the rule that switches the threshold according to
$\widehat d$, whose maximum reaches $3.24$ times. Their definitions and results appear in Appendix~\ref{app:exp-e2e-other}.

The crossover region in Table~\ref{tab:e2e} covers $-8.17\le z_0\le-5.49$ at $k=5$;
its 72 test units use splitting estimates of the two failure mechanisms.

\begin{table}[!htbp]
\centering
\caption{Post-failure coverage test. miss$/u$ is the ratio of the non-coverage rate to the nominal level, listed as median/maximum; the width is the median of $(U-d)/\sqrt M$.}
\label{tab:e2e}
\small
\setlength{\tabcolsep}{4pt}
\begin{tabular}{llrccccc}
\toprule
 & & & \multicolumn{3}{c}{miss$/u$ (median/max)} & \multicolumn{2}{c}{width}\\
\cmidrule(lr){4-6}\cmidrule(lr){7-8}
$k$ & region & units & Envelope & Mixture & Uncond. & Envelope & Mixture\\
\midrule
3 & window ($z_0\in\{0,-2\}$) & 48 & 0.63/1.04 & 0.98/1.08 & 0.63/1.04 & 2.73 & 2.71\\
3 & tail ($z_0\in\{-3.5,-4.5\}$) & 40 & 0.74/1.03 & 0.97/1.05 & 0.74/1.03 & 2.73 & 2.69\\
3 & low load ($z_0\in\{-6,-10\}$) & 32 & 0.94/1.10 & 0.97/1.15 & 0.94/1.10 & 2.21 & 1.83\\
5 & window & 48 & 0.98/1.07 & 0.99/1.07 & 1.08/1.87 & 1.86 & 1.86\\
5 & tail & 32 & 0.91/1.11 & 0.92/1.13 & 1.36/2.65 & 1.77 & 1.77\\
5 & crossover & 72 & 0.83/1.00 & 0.92/1.04 & 1.01/3.14 & 1.91 & 1.91\\
6 & window ($z_0\in\{0,-1.5\}$) & 16 & 1.00/1.04 & 1.00/1.04 & 1.15/2.15 & 1.65 & 1.63\\
\bottomrule
\end{tabular}
\end{table}

Three points can be read from the table. First, the maximum ratio of the envelope rule and the mixture rule is within $1.15$ in all regions, and the median does not exceed $1$. Second, for $k=3$ the
unconditional threshold and the envelope rule give exactly the same result, because for $k=3,4$ the envelope rule takes precisely the
unconditional threshold; once $k\ge5$, the unconditional threshold reaches maximum ratios of $1.87$, $2.65$, and $3.14$ in the
window, tail, and crossover regions respectively, and among the 288 test units there are 90 whose $99.9\%$ lower bound exceeds
$u$. This is direct evidence that the selection correction is necessary at realistic scales. Third, for $k\ge5$ the widths of the two rules are almost identical; for $k=3$ the mixture rule is narrower,
with a median width of $1.83\sqrt M$ on the low-load side against $2.21\sqrt M$ for the envelope rule, at the cost of a non-coverage rate closer to the nominal value. The envelope rule is conservative in the $k=3$
window (median ratio $0.63$), which is the cost of giving up the selection correction for $k=3,4$. The proportion of empty acceptance sets $\mathcal C$ is at most
$0.009$, occurring in the near-one-sided composition of the $k=3$ window, and does not exceed $0.002$ in the remaining regions. The test covers only $M\ge4096$ and the above configurations,
and does not cover the far low-load region for $k\ge5$ or the sparse end where $d/\sqrt M$ is small.

\section{Related work}\label{sec:related}
\paragraph{Random hypergraph core and finite-size scaling.} The threshold and the thinning solution of the random hypergraph 2-core
go back to Molloy\cite{molloy}, Darling--Norris\cite{dn}, and Janson--Łuczak\cite{jl07}. Dembo--Montanari
proved the failure probability of the configuration model inside the critical window together with a correction of order $n^{-1/6}$
\cite{dm}; their Remark 2.6 also gives the limiting law of the core size conditional on the core being present. Amraoui
et al.\ proposed the finite-length scaling and the threshold shift in iterative decoding\cite{amru}, and Nozaki--Kasai--Sakaniwa gave
an analytic solution of the covariance evolution\cite{nks}. These works study the failure itself or the core itself.
This paper studies a different object: the joint law between an external measurement, determined already before peeling,
and the failure score. The constant $\eta_k$
is precisely the cross-correlation between the measurement and the score; like the window width $\sigma_k$, it is determined by the full degree profile.
The numerical value of $\sigma_k$ agrees with the width parameter of AMRU, related by
$\sigma_k=ka_ce^{-x_c}\alpha_{\rm AMRU}$; the NKS result reproduces the dynamic part of $\sigma_k$ in the Poisson case.

\paragraph{Conditioning and rare cores.} Scott--Sorkin characterized the conditional structure of specific rare cores
below the threshold\cite{ss}. Janson--Łuczak proved the normality of the random graph $k$-core by conditioning on the defect\cite{jl}.
This paper carries over the same conditioning strategy when moving from the configuration model to the Plain model; unlike those works,
we also need to prove that the joint Gaussian law of the measurement
and the score remains unchanged after conditioning.

\paragraph{Two failure mechanisms.} The iterative decoding literature writes the block error rate as the sum of a waterfall
term and an error floor term (AMRU 2006, Eq.~(17)(18)\cite{amru06}). Pittel--Sorkin
also gave bounds that juxtapose the two mechanisms for the satisfiability probability of $k$-XORSAT\cite{ps}. Goodrich--Mitzenmacher
observed that, at a fixed subcritical load, the failure bottleneck of IBLT is two edges mapping to the same set of cells, giving a
failure rate of $O(t^{-k+2})$\cite{gm}. These results all stay at the level of probabilities.
Proposition~\ref{prop:lowload} goes further to give the relative constant and the conditional law of the measurement; Proposition~\ref{prop:mixture}
combines the two mechanisms into an explicit mixture of conditional CDFs and handles the crossover region where macroscopic failure has
already become rare.

\paragraph{Quadratic forms and selection distributions.} For the CLT of quadratic forms with random signs, there are results of de
Jong\cite{dejong}, Bloznelis--Götze\cite{bg}, and
Bresler--Nagaraj\cite{bn}, which respectively require independent inputs, a nondegenerate linear part, or a deterministic regular
underlying graph. For the general framework of conditional CLTs, see Dedecker--Merlevède\cite{dmer}. Theorem~\ref{thm:sign} treats
the case where the overlap graph is random, sparse, and nonregular, while the sign composition is fixed. The distribution
$\Phi_2(x,z;c)/\Phi(z)$ is Azzalini's extended skew-normal distribution\cite{azz85,aa06}. This paper's contribution is to prove that it
appears as the limit of the measurement conditional on failure, and to give the closed form of the parameter $\eta_k\rho^2$.

\paragraph{Failure analysis of IBLT.} The IBLT literature gives the failure probability and list guarantees at finite
length\cite{gm,yw,mizrahi}, but does not discuss the count measurement read after failure. Reference~\cite{prior} proposed the
self-sizing protocol and proved two things: the exact sign projection of the conditional mean, and the chi-square inheritance at fixed
$M$. This paper advances these conclusions to the critical window, where the failure probability is nondegenerate, and to the whole load
axis.

\section{Discussion and conclusion}
The count measurement, read only after decoding failure, carries a selection effect. Under sign imbalance this effect persists after standardization.
Its strength is governed by the closed-form constant $\eta_k\rho^2$, and its direction reverses between $k=4$ and $k=5$; improving the unconditional
reference distribution does not remove it. For the actual protocol, when $k\ge5$, the one-sided imbalance places the bias on the side of insufficient capacity;
under asymptotically balanced signs the measurement is unaffected by selection. The observable sign totals suffice to construct an upper confidence bound after failure,
and both implementation rules, the envelope rule and the mixture rule, pass the post-failure coverage test.

Along the load axis, the failure mechanism passes gradually from the macroscopic core of the critical window to edge repetition at a fixed subcritical load, and the failure-conditioned distribution is an explicit two-layer mixture.
When $k\ge5$, the mechanism that dominates the failure probability separates from the mechanism that dominates the conditional mean: the layer where the mean changes sign lies at
$t^2\approx(2k-3)\log M$, deeper than the equal-probability layer. In the tested finite-size configurations, the Dembo--Montanari shift
approximates the threshold center, and the chi-square $(M-1)$ law approximates the shape of the unconditional
measurement; neither involves a fitted parameter.

This paper's theory covers fixed $k\ge3$, the ideal Plain mapping, and deterministic signs independent of the mapping. The failure-conditioned measurement
distribution is now characterized along the whole load axis (Proposition~\ref{prop:axis}, requiring $d/\sqrt M\to\infty$). The following questions are left for future work:
the error rate at finite size; in the moving subcritical region with $|r_M|/\sqrt{\log M}\to\infty$,
the relative probability and the fine conditional law of the macroscopic layer itself, a region beyond the scope of Proposition~\ref{prop:mda}, where the failure-conditioned law
is dominated by the repeated-edge layer; the complete classification of the sparse end $d=O(\sqrt M)$; $k$ growing with the size; and non-ideal checksums and other mapping families.

\appendix
\section{Notation}\label{app:notation}
This appendix collects the notation used throughout the paper, grouped in the order of first
appearance. Unless an explicit comparison across different scales is required, the dependence
of each quantity on $M$ is omitted. The finite-size sign imbalance is denoted by $\rho$ and its
limit by $\rho_\infty$; for the remaining parameters ($r$ and the like), the absence of a
subscript denotes the finite quantity, and the subscript $\infty$ is added when the limit is
meant. The subscripts $c$ and $k$ denote constants that depend only on $k$. The symbols
$\Rightarrow$ and $\xrightarrow p$ denote convergence in distribution and convergence in
probability, respectively. Unless stated otherwise, $k\ge3$ is fixed and $M\to\infty$.

\newcommand{\notationtable}[2]{%
\par\smallskip\noindent\textbf{#1}\par\nopagebreak\smallskip
{\small\noindent\begin{tabular}{@{}p{.22\linewidth}p{.74\linewidth}@{}}
\toprule
Notation & Meaning\\
\midrule
#2
\bottomrule
\end{tabular}}\par\medskip}

\notationtable{Model and peeling}{
$M$ & number of cells, i.e. the number of vertices of the mapping hypergraph\\
$d$ & number of input elements, i.e. the estimation target\\
$k$ & number of cells to which each element is mapped\\
$S_i$ & the $k$-subset chosen by element $i$\\
$F$ & the event that ideal peeling fails, equivalently that the residual 2-core is nonempty\\
$\alpha$ & limit of the load $d/M$, Theorem~\ref{thm:sign}\\
$r$ & window coordinate, $d/M=\alpha_c+r/\sqrt M$; within the critical window $r$ is fixed\\
$n$ & candidate count, used in the test in place of the unknown $d$\\
}

\notationtable{Signs and measurement}{
$s_i$ & sign of element $i$, taking values in $\pm1$\\
$d_+$, $d_-$ & numbers of positive and negative elements\\
$b$ & total sign $d_+-d_-=d\rho$, observable by the protocol\\
$\rho$, $\rho_\infty$ & the sign imbalance $b/d$ and its limit\\
$C_v$ & count of cell $v$, $\sum_{i:v\in S_i}s_i$\\
$T$ & the count energy before peeling, Eq.~\eqref{eq:T}\\
$\widehat d$ & estimate of $d$ obtained from the energy\\
$s$ & standard deviation $\sqrt{2d(d-1)/(M-1)}$; $s(n)$ is the same expression evaluated at a candidate $n$\\
$X$ & standardized measurement $(\widehat d-d)/s$\\
$X_{\rm one}$ & the standardized measurement when all signs are equal\\
$\gamma$ & sign factor $(d\rho^2-1)/(d-1)$, with $\gamma\to\rho_\infty^2$\\
$V$ & sign residual $X-\gamma X_{\rm one}$\\
$c_F$ & increment in the expected overlap of two elements under failure selection relative to the unconditional case, Eq.~\eqref{eq:projection}\\
$q_u$ & the $u$ lower quantile of $\widehat d/d$ under failure\\
}

\notationtable{Critical constants}{
$x_c$ & the unique positive root of the equation $e^{x}-1=(k-1)x$\\
$q_c$, $a_c$ & $1-e^{-x_c}$ and $q_c^{k-1}$\\
$\mu_c$, $\alpha_c$ & critical mean degree $x_c/a_c$ and critical load $\mu_c/k$\\
$\beta_k$ & slope of the score with respect to the window coordinate, $-ka_ce^{-x_c}$\\
$\sigma_k^2$ & limit variance of the score, equal to the initial part $\sigma_0^2$ plus the dynamic part $\tau_k^2$\\
$\tau_k^2$ & variance of the dynamic part of the score, Eq.~\eqref{eq:tau}\\
$\eta_k$ & limit correlation coefficient of the measurement and the score when all signs are equal, Eq.~\eqref{eq:eta}\\
$w_k$ & window width $\sigma_k/|\beta_k|$\\
$C_k$ & coefficient of the finite-size threshold shift, Section~\ref{sec:rules}\\
}

\notationtable{Joint limit and conditional law}{
$R$ & precritical score, depending only on the unsigned mapping, Eq.~\eqref{eq:score}\\
$z_r$ & standardized window coordinate $-\beta_kr/\sigma_k$\\
$c$ & correlation coefficient of the joint limit, $\eta_k\rho_\infty^2$\\
$\Phi$, $\phi$ & standard normal distribution function and density\\
$\Phi_2(\cdot,\cdot;c)$ & standard bivariate normal distribution function with correlation $c$\\
$A(z)$ & $\phi(z)/\Phi(z)$\\
$f_{\rm sel}$, $q_{\rm sel}$ & density and quantile of the selection law\\
$\mathcal G$ & unsigned Plain hypergraph with element labels\\
$\mathcal S$ & event that depends only on the unsigned mapping and is invariant under permutation of the element labels\\
}

\notationtable{Upper confidence bound and implementation rule}{
$u$ & nominal lower-tail level\\
$\mathcal N(b)$ & candidate set $\{n\ge\max(|b|,2):n\equiv b\pmod 2\}$\\
$z(n)$ & standardized window coordinate corresponding to a candidate $n$\\
$t(n,b)$ & acceptance threshold for a candidate $n$\\
$\mathcal C$ & acceptance set; $U=\sup\mathcal C$ is the upper confidence bound\\
$\kappa(u)$ & chi-square shape correction\\
}

\notationtable{Load axis}{
$P_2$ & event that the actual core is exactly two repeated edges\\
$\lambda$, $\lambda_n$ & expected number of repeated edges $\binom d2/\binom Mk$, and its value at a candidate $n$\\
$\delta_k$, $b_k$ & exponents of the error terms in Proposition~\ref{prop:lowload} and Proposition~\ref{prop:mixture}\\
$\mathcal A$ & macroscopic failure event in which the number of edges of the actual 2-core exceeds $\varepsilon_0M$\\
$H_n$, $T_n$, $w_n$ & mixture distribution, acceptance threshold, and macroscopic-layer weight of Proposition~\ref{prop:axis}\\
}

\notationtable{Experimental coordinates}{
$z_0$ & theorem coordinate $(d/M-\alpha_c)\sqrt M/w_k$\\
$z_{\rm sh}$ & shifted coordinate $z_0+C_kM^{-1/6}$\\
$z_{\rm emp}$ & empirical coordinate recovered from the measured failure rate\\
$\widehat p_F$ & measured failure rate\\
$\sigma$ & product of the signs of the two edges of a repeated edge\\
}

\notationtable{Notation in the proofs}{
$D_v$ & initial degree of cell $v$\\
$\mu$ & finite-size mean degree $kd/M$\\
$G(f)$ & centered empirical functional of the initial degree\\
$e$, $h_c$ & influence functions of the initial parts of the measurement and the score\\
$\Gamma(f,g)$ & limit covariance at a fixed total number of endpoints, Eq.~\eqref{eq:bridge}\\
$W$ & dynamic innovation of the score\\
$\Lambda_k$ & limit mean of the number of bad edges in the configuration model\\
$H_{ij}$, $m_{ij}$ & overlap $|S_i\cap S_j|$ of the mappings of two elements\\
$\Omega$, $\Omega_2$, $\ell_i$, $\Omega_\ell$ & total overlap, sum of squares, weighted degree, and sum of squared weighted degrees of the overlap graph\\
$\theta$ & number of steps until the number of degree-one vertices first becomes zero\\
$\vartheta$ & error scale $\sqrt M(\log M)^6$ in the uniform concentration before stopping\\
}

\section{Proof of the main theorem}\label{sec:main-proof}
This appendix proves Theorem~\ref{thm:main}, Theorem~\ref{thm:sign}, and Corollary~\ref{cor:balanced}.
Section~\ref{sec:proof-sketch} first gives the proof sketch; the subsequent subsections give the complete
chain of proof, and the expansion of the key estimates is collected in Appendix~\ref{app:technical}.

The proof proceeds in five steps. The first step analyzes the initial degree distribution and obtains the
joint Gaussian limit of the energy and the influence function (Section~\ref{sec:mp-init}). The second step
tracks the full degree trajectory along the peeling process, constructs the precritical score, and decomposes
it into an initial-state part and the dynamic innovation (Sections~\ref{sec:mp-traj} and \ref{sec:mp-stop}).
The third step proves that the sign of the score determines the actual success or failure
(Section~\ref{sec:mp-global}). The fourth step transfers the result from the configuration model to the
Plain model (Section~\ref{sec:mp-plain}). The fifth step handles general sign composition
(Section~\ref{sec:signproof}). The two subsections after the proof sketch first give the critical constants
and then fill in the details of a few short arguments in the main text.

\subsection{Proof sketch}\label{sec:proof-sketch}
The proof of Theorem~\ref{thm:main} first establishes the joint limit of $(X_{\rm one},R_M)$ in the
all-equal-sign case, and then handles general sign composition via Theorem~\ref{thm:sign}. The
all-equal-sign case proceeds in four steps; the complete argument appears in the subsequent subsections, and
the key estimates appear in Appendix~\ref{app:technical}.

\paragraph{Initial state.} We first work in the configuration model: the $k$ endpoints of each edge are
drawn independently and uniformly, with repetition allowed. The measurement depends only on the initial
degrees, so it can be written as a functional $G_M(e_M)$ of the degrees, where $e_M(j)=(j-\mu_M)^2$. The
initial-state part of the score is another functional $G_M(h_c)$, with influence function
$h_c(j)=j(1-a_c)^{j-1}$ depending on the complete profile of the initial degrees. The total number of
endpoints is fixed, so the joint Gaussian limit of the two must subtract the correction term in
\eqref{eq:bridge}.

\paragraph{Peeling trajectory and the score.} Along the peeling process, the degree histogram follows the
fluid equation; at the critical parameters the degree-one density has a non-degenerate quadratic minimum,
which is the bottleneck of the peeling. The score $R_M$ reads the histogram at the entry time before the
bottleneck and extrapolates to the bottleneck with backward weights (\eqref{eq:score}). A Doob decomposition
of the backward functional writes $R_M$ as three terms: the deterministic term $\beta_kr$ from the window
location, the initial-state part $G_M(h_c)/\mu_c$, and the dynamic innovation $W$ (\eqref{eq:decomp}). The
increments of the innovation are $O(M^{-1/2})$, and its quadratic variation tends to the constant
$\tau_k^2$; stable convergence (\eqref{eq:stable}) shows that it is independent of the initial-state limit.
Hence the correlation between the measurement and the score comes entirely from the initial state, equal to
$\Gamma(e,h_c)/(\sqrt2\mu_c^2\sigma_k)$. Substituting $\Gamma(e,h_c)=\mu_ce^{-x_c}x_c(x_c-2)$ gives
\eqref{eq:eta}, and the sign of $\eta_k$ is exactly the sign of $x_c-2$. The influence function is
determined by the complete degree profile, so the correlation depends on the full initial degree
distribution.

\paragraph{The sign of the score determines success.} Within the short window near the bottleneck the local
noise freezes, and the number of degree-one vertices is approximately $x_c\sqrt M R_M$ plus the fluid
difference (\eqref{eq:freeze}). When $R_M$ is clearly negative, the number of degree-one vertices drops to
zero before the central time, so peeling fails; when $R_M$ is clearly positive, the process crosses the
bottleneck, after which the positive gap of the fluid and a first-moment bound on the small core rule out
failure. The Gaussian limit has no atoms, which gives \eqref{eq:transfer}.

\paragraph{Transfer to the Plain model.} The Plain model is the configuration model conditioned on the
event of having no bad edges, whose probability tends to $e^{-\Lambda_k}>0$. Planting finitely many bad
edges does not change the limit of $(X_{\rm one},R_M)$; expanding with mixed factorial moments and
inclusion-exclusion keeps the joint Gaussian law unchanged after conditioning.

\paragraph{General sign composition.} For statistics invariant under permutation of the element labels, one
can first generate the unsigned graph $\mathcal G$ and then distribute a given sign composition uniformly
over it. Theorem~\ref{thm:sign} decomposes the measurement into $\gamma X_{\rm one}$ and a term $V$
asymptotically independent of the graph. The correlation of the former with the score is $\eta_k$, and
multiplying by $\gamma\to\rho^2$ gives the correlation coefficient $\eta_k\rho^2$. The proof of
Theorem~\ref{thm:sign} itself is in Section~\ref{sec:signproof}: first approximate with iid signs, then flip
a whole batch from the surplus sign class to correct the total sign count exactly, and finally obtain the
normal limit from the self-averaging of the overlap graph and cumulant estimates on the dependency graph.

\subsection{Critical constants}\label{app:constants}
Fix $k\ge3$, and let $x_c$ be the unique positive root of the equation $e^{x}-1=(k-1)x$\cite{molloy,dn,jl07}. Let
\begin{equation*}
q_c=1-e^{-x_c},\quad a_c=q_c^{k-1},\quad
\mu_c=x_c/a_c,\quad\alpha_c=\mu_c/k.
\end{equation*}
$\mu_c$ is the critical average degree and $\alpha_c$ the critical load. In Theorem~\ref{thm:main} the slope
of the score against the window coordinate and the limiting variance of the score are
\begin{equation*}
\beta_k=-ka_ce^{-x_c},\qquad
\sigma_k^2=e^{-x_c}q_c\left(1-\frac{k a_c}{k-1}\right).
\end{equation*}
$\sigma_k^2$ is composed of the initial-state part $\sigma_0^2$ and the dynamic part $\tau_k^2$; see
\eqref{eq:tau} and the simplification that follows it. For $k=3,4,5,6$ the values of $\tau_k^2$ are
$0.031017361$, $0.016877142$, $0.011169746$, $0.008107775$; the numerical values of $\sigma_k^2$ and
$\eta_k$ appear in Table~\ref{tab:constants}.

\paragraph{Finite-size shift constants.}\label{app:ck}
The shift constants of Section~\ref{sec:rules} are obtained by converting the critical threshold shift of
Dembo--Montanari\cite{dm} to the scale of this paper. Write $w_k=\sigma_k/|\beta_k|$,
$\widetilde G=(k-2)/(k-1)$, $\widetilde F=[(k-1)x_c-(k-2)]/(kq_c^k)$. Then
\begin{equation*}
C_k=\frac{1}{w_k}\cdot\frac{0.996193\,\widetilde G^{2/3}\widetilde F^{-1/3}\alpha_c^{4/3}}{x_c^2e^{-x_c}},
\end{equation*}
where the constant $0.996193$ is the numerical value of the Airy integral in DM's eq. (2.16).

\subsection{Exact normalization and local uniformity}\label{app:normalization}
The normalization in \eqref{eq:measurement} of Section~\ref{sec:model} is exact, for the following reason.
Let $H_{ij}=|S_i\cap S_j|$. Then
\begin{equation*}
T=kd(1-k/M)+2\sum_{i<j}s_i s_j(H_{ij}-k^2/M).
\end{equation*}
The hypergeometric distribution gives
\begin{equation*}
\E H_{ij}=k^2/M,\qquad
\Var H_{ij}=\frac{k^2(M-k)^2}{M^2(M-1)}.
\end{equation*}
Different unordered pairs of elements have uncorrelated overlaps. Even when two pairs share an element, once
the mapping of that element is given, the conditional mean of each pair is still the constant $k^2/M$.
Summing these quantities gives the normalization in \eqref{eq:measurement}.

\paragraph{Local uniformity.} The window assumption \eqref{eq:window} allows an $o(M^{-1/2})$ slack, and the
way $\rho_M\to\rho$ is unrestricted, so any sequence with $r_M\to r$ satisfies the conditions of the theorem.
A subsequence argument gives: the joint limit in Theorem~\ref{thm:main} and the limit of the conditional CDF
hold uniformly over $(r,\rho)$ on any compact subset of $\mathbb R\times[-1,1]$. This removes the need to
prove a separate uniform CLT.

\paragraph{Convergence of conditional moments.} Convergence in distribution alone does not imply convergence
of moments: a small part of the probability mass can escape to infinity, vanish in the limit, yet still bias
the mean. Section~\ref{sec:mean} therefore needs $X^2$ to be uniformly integrable under failure. This has two
layers. The first is unconditional. Since $X\Rightarrow N(0,1)$ and $\E X^2=1$ holds exactly for every $M$,
the second moment equals the second moment of the limit, so $X^2$ is uniformly integrable. The second is
conditional. For any $K>0$,
\begin{equation*}
\E\bigl[X^2;\,X^2>K\bigm|F\bigr]\le\frac{\E\bigl[X^2;\,X^2>K\bigr]}{\Pr(F)}.
\end{equation*}
The numerator tends to zero uniformly in $M$ as $K\to\infty$, and the denominator tends to $\Phi(z_r)>0$, so
$X^2$ is also uniformly integrable under failure. The critical window requires the limiting failure
probability to be positive, and it matters here: if the denominator tended to zero, conditioning could blow up
the tail without control. Combined with local uniformity, the limits of the first two conditional moments
hold uniformly over $(r,\rho)$ on compact sets.

\subsection{Initial joint Gaussian limit}\label{sec:mp-init}
First relax the model: the $k$ endpoints of each edge are drawn independently and uniformly, with repetition
allowed. This is the configuration model. Write $D_v$ for the initial degree of vertex $v$ and
$\mu_M=kd/M$. For $e_M(j)=(j-\mu_M)^2$ and $h_c(j)=j(1-a_c)^{j-1}$ (with $h_c(0)=0$), define
\begin{equation*}
G_M(f)=\frac1{\sqrt M}\left[\sum_v f(D_v)-M\E f(\Pois(\mu_M))\right].
\end{equation*}
Conditioned on the total number of endpoints being fixed, the occupancy CLT gives the joint Gaussian limit of
$(G_M(e_M),G_M(h_c))$ with covariance
\begin{equation}
\Gamma(f,g)=\Cov(f(D),g(D))-
\frac{\Cov(f(D),D)\Cov(g(D),D)}{\mu_c},\quad D\sim\Pois(\mu_c).\label{eq:bridge}
\end{equation}
The second term is the correction coming from the constraint that the total number of endpoints is fixed.

The concrete approach is de-Poissonization. On the same independent endpoint sequence, replace the fixed
length $L=kd$ by an independent random length $N\sim\Pois(L)$. Write $S_f(l)=\sum_v f(D_v(l))$. Then
\begin{equation*}
S_f(N)-S_f(L)=a_f(N-L)+o_p(\sqrt M),\qquad
a_f=\frac{\Cov(f(D),D)}{\mu_c}.
\end{equation*}
For $e_M$, each newly added endpoint changes the sum on average by $2(l/M-\mu_M)+1$; for $h_c$, both the
first and second discrete differences are bounded. Over the range $|l-L|\le\sqrt M\log M$ the accumulated
drift error is $O_p(\log^2M)$ and the added martingale fluctuation is $o_p(\sqrt M)$. The ordinary Poisson
vector CLT plus this correction gives \eqref{eq:bridge}. The one-sided measurement satisfies
\begin{equation*}
X_{\rm one}=\frac{G_M(e_M)}{\sqrt2\mu_c}+o_p(1).
\end{equation*}
The measurement in the configuration model differs from the center used here, but this difference vanishes on
the $\sqrt M$ scale.

\subsection{Degree trajectory and the precritical score}\label{sec:mp-traj}
Let $t$ be the number of deleted edges divided by $M$, and write
\begin{equation*}
m_M(t)=\mu_M-kt,\quad a_M(t)=\left(\frac{m_M(t)}{\mu_M}\right)^{(k-1)/k},\quad
q_M(t)=a_M(t)^{1/(k-1)}.
\end{equation*}
Write $U_j$ for the number of vertices of current degree $j$ and $Q=U_1$ for the number of degree-one
vertices. While peeling has not stopped, the coordinates of degree $j\ge2$ satisfy the following fluid
equation; starting from the Poisson initial state, its solution is
\begin{equation*}
\dot u_j=\frac{k-1}{m_M(t)}[(j+1)u_{j+1}-ju_j],\qquad
u_j(t)=e^{-\mu_Ma_M(t)}\frac{[\mu_Ma_M(t)]^j}{j!},\quad j\ge2.
\end{equation*}
The corresponding degree-one density is
\begin{equation*}
\ell_M(t)=m_M(t)-\mu_Ma_M(t)[1-e^{-\mu_Ma_M(t)}].
\end{equation*}
At the critical parameters, the interior zero of $\ell_M$ is a non-degenerate quadratic minimum. This is the
bottleneck of the peeling.

Now define the backward weights of the heavy stubs at the final point:
\begin{equation*}
f_j(b)=jb[1-(1-b)^{j-1}],\quad f_0=f_1=0.
\end{equation*}
When the thinning parameter $a_c$ of the final point is fixed, this weight satisfies the backward equation.
The projection of the initial fluctuation onto it is $-\sum_j f_j(a_c)G_j=a_c\sum_jh_c(j)G_j$, where we used
the total constraint $\sum_jjG_j=0$. We adopt the convention that failure corresponds to a low score; under
this convention the direction of the initial-state influence is the \emph{positive} $h_c/\mu_c$.

Take
\begin{equation*}
t_{*,M}=\frac{\mu_M}{k}(1-q_c^k),\quad
\varepsilon_M=M^{-1/8},\quad t_-=t_{*,M}-\varepsilon_M,\quad
n_-=\lfloor Mt_-\rfloor.
\end{equation*}
If the process reaches $n_-$, let $b_-=a_c/a_M(t_-)$ and define
\begin{equation}
R_M=\frac{M m_M(t_{*,M})-\sum_{j\ge2}f_j(b_-)U_j(n_-)}{x_c\sqrt M}.\label{eq:score}
\end{equation}
If the process does not reach $n_-$, set $R_M=0$ for now; we will show later that this case has probability
$o(1)$. $R_M$ is an extrapolated score, computed directly from the degree histogram at the entry time $n_-$,
and it can be negative.

\subsection{Concentration before the bottleneck and dynamic innovation}\label{sec:mp-stop}
\begin{lemma}[Uniform concentration before stopping]
Let $\theta_M$ be the first step at which the number of degree-one vertices is zero. For fixed $\delta>0$,
take $\vartheta_M=\sqrt M(\log M)^6$. Then on $n\le\theta_M\wedge(d-\delta M)$, the uniform errors of
$Q_n-M\ell_M(n/M)$ and of $U_j(n)-M u_j(n/M)$ ($2\le j\le\lceil\log M\rceil$) are all $O_p(\vartheta_M)$.
\end{lemma}
The argument uses the backward binomial weights and $f_j$, and stops the process at $\theta_M$. The maximum
initial degree is at most $L_M=\lceil\log M\rceil$ with high probability. As long as the number of remaining
stubs is $\Theta(M)$, the backward equation cancels the main drift, and the per-step remainder is controlled
by $C_\delta L_M^4/M$, with the projected martingale difference uniformly bounded. Combining three tools
gives the conclusion: the Azuma concentration inequality, the bounded-difference concentration of the initial
state, and a union bound over $O(ML_M)$ target times and degrees. The fluid gap at the entry is
$M\varepsilon_M^2=M^{3/4}$, much larger than $\vartheta_M$, so the process reaches $n_-$ with probability
$1-o(1)$. The detailed estimates are in Appendix~\ref{app:technical}.

A Doob decomposition of the backward functional gives
\begin{equation}
R_M=\beta_k r+\frac{G_M(h_c)}{\mu_c}+W+o_p(1).\label{eq:decomp}
\end{equation}
On the right are, in order, the deterministic term from the window location, the initial-state part, and the
dynamic innovation $W$. The normalized increments of $W$ are $O(M^{-1/2})$, and its predictable quadratic
variation tends to the deterministic constant $\tau_k^2$. Let $\Psi_M$ be measurable with respect to the
initial-state information and uniformly bounded. The stable decomposition we need is
\begin{equation}
\E[\Psi_Me^{iuW}]-e^{-u^2\tau_k^2/2}\E[\Psi_M]\longrightarrow0.\label{eq:stable}
\end{equation}
The proof applies a one-step Taylor expansion to the compensated exponential, with the sum of the remainders
$O(M^{-1/2})$; together with the uniform boundedness of the quadratic variation and its convergence to the
deterministic value, this gives \eqref{eq:stable}. Taking $\Psi_M$ to be the characteristic function of the
initial-state vector yields simultaneously the joint Gaussian limit and the independence between the
innovation and the initial-state limit. The detailed estimates are in Appendix~\ref{app:technical}.

Next we compute the dynamic variance. On the critical fluid trajectory let $b=a_c/a(t)$ and
$g_j(b)=f_j(b)-f_{j-1}(b)$, and set $g_1=0$. A random partner hits a vertex of degree $j\ge2$ with
probability $p_j(t)=j u_j(t)/m(t)$; the remaining probability mass goes into $p_1$. Each step deletes $k-1$
fixed partners, so
\begin{equation}
\tau_k^2=\frac{k-1}{x_c^2}\int_0^{t_c}
\left[\sum_jp_j(t)g_j(b)^2-\left(\sum_jp_j(t)g_j(b)\right)^2\right]dt.\label{eq:tau}
\end{equation}
Let $s_c=1-a_c$. The initial-state contribution $\sigma_0^2$ and the total variance simplify to
\begin{align*}
\sigma_0^2&=(1/\mu_c+s_c^2)e^{\mu_c(s_c^2-1)}
-e^{-2x_c}[1+(1-x_c)^2/\mu_c],\\
\sigma_k^2&=\sigma_0^2+\tau_k^2=e^{-x_c}q_c\left(1-\frac{k a_c}{k-1}\right).
\end{align*}
Using $\Gamma(e,e)=2\mu_c^2$ and $\Gamma(e,h_c)=\mu_ce^{-x_c}x_c(x_c-2)$, \eqref{eq:decomp} yields
\eqref{eq:eta}.

\subsection{The sign of the score determines the peeling outcome}\label{sec:mp-global}
Consider the window $[t_{*,M}-\varepsilon_M,t_{*,M}+\varepsilon_M]$ around the bottleneck. For all times
inside the window at which the process has not yet stopped (including the stopping time itself), the local
noise freezing gives
\begin{equation}
Q_n=x_c\sqrt M R_M+M[\ell_M(n/M)-\ell_M(t_{*,M})]+o_p(\sqrt M),\label{eq:freeze}
\end{equation}
with the error uniform over these times. The reason is as follows. Inside the window the added martingale
fluctuation is $O_p(\sqrt{M\varepsilon_M})$, and the accumulated error of the centered drift is
$O_p(\varepsilon_M\vartheta_M)$ times a logarithmic factor, both $o_p(\sqrt M)$. The difference between the
backward weights at the entry and the actual heavy-stub weights satisfies a bound of the same order. The
detailed estimates are in Appendix~\ref{app:technical}.

From this one can deduce success or failure from the sign of the score. Fix $\zeta>0$. If $R_M\le-\zeta$ and
the process has still not stopped by the central time, \eqref{eq:freeze} gives $Q<0$, a contradiction, so
peeling fails. If $R_M\ge\zeta$, the lower bound on the local fluid difference is $-O(1)$, so a first stop
inside the window likewise contradicts \eqref{eq:freeze}. Outside the window, until only few edges remain,
the positive gap of the fluid together with the uniform concentration before stopping rule out a first stop.

Finally one must rule out a small core. Let $d\le AM$. If $e$ edges form a subgraph of minimum degree at
least two, their endpoints occupy at most $v=\lfloor ke/2\rfloor$ vertices, hence
\begin{equation*}
\Pr(\exists\text{such }e\text{ edges})
\le \binom{d}{e}\binom{M}{v}(v/M)^{ke}
\le[C_{k,A}(e/M)^{k/2-1}]^e.
\end{equation*}
For $k\ge3$ one can choose $\delta_0>0$ small enough that the sum over $1\le e\le\delta_0M$ is $o(1)$. If
the process crosses the bottleneck yet still fails in the end, the residual part would violate this bound.
Thus the outcome is determined on both sides of the safe interval $\pm\zeta$; letting $\zeta\downarrow0$ and
using that the Gaussian limit has no atoms gives \eqref{eq:transfer}.

\subsection{Transfer from the configuration model to the Plain model}\label{sec:mp-plain}
In the configuration model, call an edge bad if some endpoint is repeated within that edge, and let
$N^{\rm bad}_M$ be the number of bad edges. The edges are independent, and each edge is bad with probability
$\binom{k}{2}/M+O(M^{-2})$, so
\begin{equation*}
N^{\rm bad}_M\Rightarrow\Pois(\Lambda_k),\qquad
\Lambda_k=\alpha_c\binom{k}{2}.
\end{equation*}
Conditioning on $N^{\rm bad}_M=0$ gives exactly the Plain model. The probability of this event tends to
$e^{-\Lambda_k}>0$, so error events of probability $o(1)$ transfer directly. But for the joint Gaussian law
to be preserved after conditioning, it must also be asymptotically independent of $N^{\rm bad}_M$.

Let $Y_M=(X_{\rm one},R_M)$. Consider planting a fixed number $l$ of special edges. Planting changes the
initial quadratic energy by $O_p(l\log M+l^2)$ and changes the difference-bounded initial influence function
by $O_l(1)$. In ordinary edge-deletion steps the partner pool changes by only finitely many stubs, and the
per-step corrections to the projected drift and quadratic variation are $O_l(M^{-1})$; there are at most $l$
special deletion steps. Hence the Gaussian limit of $Y_M$ at the precritical time is unchanged. The detailed
estimates are in Appendix~\ref{app:technical}.

For any bounded continuous function $f$, exchangeability of the element labels gives
\begin{equation*}
\E[f(Y_M)(N^{\rm bad}_M)_l]\longrightarrow
\Lambda_k^l\E f(Y),
\end{equation*}
where $(b)_l$ is the falling factorial. Expanding the indicator of having no bad edges with
inclusion-exclusion, and bounding the tail of the series by $\|f\|_\infty(dp_M)^l/l!$, gives
\begin{equation*}
\E[f(Y_M)\mathbf1_{\{N^{\rm bad}_M=0\}}]
\longrightarrow e^{-\Lambda_k}\E f(Y).
\end{equation*}
Dividing both sides by the probability of having no bad edges transfers the joint Gaussian law to the Plain
model. Combined with the preceding symmetric-difference error, the Plain result in the one-sided case
follows. Methodologically this step is close to existing conditioning strategies for defects\cite{jl}.

\subsection{Proof of Theorem~\ref{thm:sign} and Corollary~\ref{cor:balanced}}\label{sec:signproof}
For a joint measurement invariant under permutation of the element labels, one can first generate the
unsigned graph $\mathcal G$ and then distribute the signs uniformly according to the given composition.
Define the weighted overlap graph
\begin{equation*}
m_{ij}=|S_i\cap S_j|,\quad
\Omega=\sum_{i<j}m_{ij},\quad \Omega_2=\sum_{i<j}m_{ij}^2,\quad
\ell_i=\sum_{j\ne i}m_{ij},\quad\bar\ell=2\Omega/d.
\end{equation*}
Here $\Omega$ is the total overlap weight. For general $d/M\to\alpha>0$, with $\mu=k\alpha$,
\begin{equation}
\Omega_2/M\xrightarrow p\mu^2/2,\quad
M^{-1}\sum_i(\ell_i-\bar\ell)^2\xrightarrow p\mu^2,\quad
\max_i \ell_i=O_p(\log M).\label{eq:selfavg}
\end{equation}
The first relation comes from the concentration of the overlap mean and $\E(\Omega_2-\Omega)=O(1)$. The
second expands into a sum over overlap edges and adjacent edge pairs, then counts by the pattern of shared
elements to obtain a variance of $O(M)$. The third follows from $\ell_i\le k\max_vD_v$.

Write the sign quadratic form as $Y_s=\sum_{i<j}m_{ij}s_i s_j$. The construction has two steps: first
generate independent signs $t_i$ with $\E t_i=\rho_M$, then flip a uniformly chosen subset of the surplus
sign class so that the total sign count equals the given value exactly. Let $\xi_i=t_i-\rho_M$. When
$|\rho|<1$, the number of signs to flip is $O_p(\sqrt d)$. Projecting to correct the total count gives
\begin{equation*}
Y_s-\rho_M^2\Omega=
\rho_M\sum_i(\ell_i-\bar\ell)\xi_i+
\sum_{i<j}m_{ij}\xi_i\xi_j+o_p(\sqrt M).
\end{equation*}
The dependency graph of the main terms on the right has maximum degree $O(\log M)$, so the $r$-th connected
cumulant divided by $M^{r/2}$ tends to zero ($r\ge3$). The linear and quadratic terms are uncorrelated, and
the total variance is
\begin{equation*}
\rho_M^2(1-\rho_M^2)\sum_i(\ell_i-\bar\ell)^2+(1-\rho_M^2)^2\Omega_2.
\end{equation*}
By \eqref{eq:selfavg}, divided by $M$ this tends to $\mu^2(1-\rho^4)/2$. The detailed estimates are in
Appendix~\ref{app:technical}.

Let $\gamma=(d\rho_M^2-1)/(d-1)$. Combining the exact finite-sample conditional mean with the CLT above gives
\begin{equation}
X=\gamma X_{\rm one}+V,\qquad
\E[e^{iuV}\mid\mathcal G]\xrightarrow p
\exp[-u^2(1-\rho^4)/2].\label{eq:quenched}
\end{equation}
The conditional characteristic function is bounded, so the convergence also holds in $L^1$. This proves that
$V$ is asymptotically independent of any graph-measurable vector that already has a joint limit. The endpoint
case $|\rho|=1$ is handled by $\E V^2=1-\gamma^2\to0$, and there is no need to establish a uniform CLT on
the rare sign class. This completes the proof of Theorem~\ref{thm:sign}. Combining it with the one-sided
joint limit gives the correlation coefficient $\eta_k\rho^2$, completing the chain of proof of the main
theorem.

When $\rho_M\to0$, on the event $\mathcal S$ whose probability is bounded away from zero, $\gamma X_{\rm
one}$ tends to zero; the $L^1$ error of \eqref{eq:quenched} divided by $\Pr(\mathcal S)$ still tends to
zero. The event is invariant under permutation, so the signs can be randomized legitimately. This proves
Corollary~\ref{cor:balanced}.

\subsubsection{Technical estimates}\label{app:technical}
This subsection gives the key estimates in the main chain of proof, with notation continuing from the main
text and the preceding subsections. All constants may depend on the fixed $k$ and on the distance $\delta$ to
the exhaustion of the endpoints, but not on the target time, the truncation degree, or $M$.

\subsubsection{Error of the initial de-Poissonization}
Let $L=kd$ and $w=\lceil\sqrt M\log M\rceil$. On the deterministic interval $[L-w,L+w]$, construct a forward
martingale in the order in which endpoints are added, and use Doob's maximal inequality to control the
martingale difference between any two times. The event $|N-L|\le w$ has probability tending to one, so this
method covers both $N<L$ and $N\ge L$ simultaneously, without a conditional-independence assumption for
deleting endpoints in reverse.

For $e_M$, the accumulated difference between the average increment and $1$ is at most
$O(w^2/M)=O(\log^2 M)$. For $h_c$, write the average increment as $A_h(l)$. Bounded differences give
$A_h(L)-\E A_h(L)=O_p(M^{-1/2})$; the difference between the Binomial and Poisson expectations on a single
bin is $O(M^{-1})$; each added endpoint changes $A_h$ by at most $C/M$. This again gives an accumulated
drift error of $O_p(\log^2M)$. The second moment of the maximum of the martingale corresponding to $h_c$ is
$O(w)$. For $e_M$, stop the martingale when the maximum degree exceeds $\lceil\log M\rceil$; this moment is
$O(w\log^2M)=o(M)$. The degree increases monotonically as endpoints are added, so a union bound on the
binomial tail at the right endpoint alone rules out the bad event of truncation failure.

For the triangular array, first take the drift coefficient as $a_{f,M}=\E[f(P_M+1)-f(P_M)]$ with
$P_M\sim\Pois(\mu_M)$. For $e_M$ this coefficient is exactly one; for $h_c$ it differs from the critical
value $a_f$ by $O(M^{-1/2})$, which after multiplying by $N-L=O_p(\sqrt M)$ contributes only $O_p(1)$.
Applying the ordinary CLT to the Poisson triple $(e_M,h_c,P_M)$ and subtracting $a_f$ times the total-count
coordinate gives all the cross-covariances of \eqref{eq:bridge}. One can also cite directly the existing
conditional CLT, namely Janson's Theorem 2.2\cite{janson2001}. Its conditions all hold here: the Poisson
parameter lies in a positive compact interval, the lattice characteristic function satisfies a uniform
quadratic lower bound, and the required fixed-order moments are bounded. The theorem does not require the
quadratic energy to have a positive exponential moment.

\subsubsection{One-step drift of the backward functional and uniform concentration}
Fix the target step $n$, and let the current step be $s<n$. Each step deletes exactly the $k$ stubs of one
edge, so the number of remaining stubs $N_s=k(d-s)$ is deterministic, and the cancellation in the backward
equation is exact. On the interval under consideration, $N_s\ge k\delta M$ and the maximum degree is at most
$L_M=\lceil\log M\rceil$. For
\[
\phi_l(b)=\Pr\{\operatorname{Bin}(l,b)=j\},\quad j\ge2,
\qquad\text{or}\qquad \phi_l(b)=f_l(b),
\]
we have $\phi_0=\phi_1=0$ and $|\phi_l-\phi_{l-1}|\le2$. Let $b_s=a_M(n/M)/a_M(s/M)$ and
$V_s=\sum_v\phi_{D_v(s)}(b_s)$. Fixing $b_s$, write $\Delta_l=\phi_{l-1}-\phi_l$ and
$A=\sum_vD_v\Delta_{D_v}$; then $|A|\le2N_s$. The forced-deleted leaves contribute nothing to either $V$ or
$A$. Each deleted partner changes $A$ by at most $CL_M$, and the denominator decreases by only one. Exposing
the fixed $k-1$ partners in order gives
\[
\E[\Delta V\mid\mathcal F_s,\ b_s\text{ fixed}]
=\frac{k-1}{N_s}A+O(L_M/M).
\]
This comparison already covers the case where the same vertex is hit several times, and does not need the
hits to be independent.

On an interval at fixed distance from the exhaustion of the endpoints, the first two time derivatives of $b$
are bounded. From the differences of Bernstein polynomials we have $|\partial_b\phi_l|\le CL_M^2$ and
$|\partial_b^2\phi_l|\le CL_M^3$. First expand the time change in the post-deletion state, then switch back
to the pre-deletion state: the second-order remainders over all vertices sum to $O(L_M^3/M)$; at most $k$
vertices change, and the cross terms this produces are $O(L_M^2/M)$. The backward equation cancels the
first-order main terms in both the sampling and the time directions, hence
\[
|\E[V_{s+1}-V_s\mid\mathcal F_s]|\le CL_M^3/M
\le CL_M^4/M.
\]
After $O(M)$ steps the accumulated error is $o(\sqrt M)$. After the process stops, the state and the weights
freeze together. The conditional random difference of each step comes only from at most $k$ deletions, so the
centered martingale difference is at most $4k$ in absolute value.

The weight sum at the starting point satisfies a uniform bounded-difference condition for each independent
endpoint, and its expectation differs from the Poisson thinning value by $O(1)$. The Azuma inequality plus a
union bound over the $O(ML_M)$ targets $(n,j)$ gives the uniform $O_p(\vartheta_M)$ bound in
Section~\ref{sec:mp-stop}. The tail of the Poisson distribution at $j>L_M$, weighted by any fixed polynomial
and multiplied by $M$, still tends to zero. The statement above includes the moment of the first stop.

\subsubsection{Convergence of the dynamic bracket and its separation from the initial state}
Deleting one partner contributes $g_j(b)$ to the projection. When the number of remaining stubs is
$\Theta(M)$, the variance correction from sampling without replacement and from repeated hits is
$O(L_M/M)$; switching the weights of the next moment back to the current weights produces another error of
$O(L_M^2/M)$. So the true one-step conditional variance is
\[
(k-1)\left[\sum_{j\ge2}p_jg_j^2-
 \left(\sum_{j\ge2}p_jg_j\right)^2\right]+O(L_M^2/M).
\]
Replacing the state by the fluid value via uniform concentration gives an error of
$O_p(\vartheta_M L_M^2/M)=o_p(1)$. Truncating the Poisson tail and writing the sum as a Riemann sum gives
\eqref{eq:tau}; the missing final segment of the integral has length $\varepsilon_M\to0$, and the integrand
is bounded, so this part can be ignored.

Let $Z_s$ be the normalized stopped martingale and $C_s=\langle Z\rangle_s$. On the whole sample space
$|\Delta Z_s|\le4k/(x_c\sqrt M)$, and
\[
0\le C_s\le\frac{16k^2d}{x_c^2M},\qquad C_{n_-}\xrightarrow p\tau_k^2.
\]
Applying a one-step conditional Taylor expansion to the compensated exponential $\exp(iuZ_s+u^2C_s/2)$, the
second-order terms cancel and the sum of the third-order remainders is $O(M^{-1/2})$. The modulus of the
exponential is controlled by the deterministic upper bound above. For any $\Psi_M$ measurable with respect to
the initial state and satisfying $|\Psi_M|\le1$, the same telescoping sum gives
\[
\left|\E[\Psi_Me^{iuZ_{n_-}}]
-e^{-u^2\tau_k^2/2}\E \Psi_M\right|
\le C_u\left(M^{-1/2}+\E|C_{n_-}-\tau_k^2|\right)=o(1).
\]
This proves \eqref{eq:stable} uniformly over all initial-state test functions, and does not need sigma-fields
nested in $M$. The independence between the Gaussian limit of the initial state and the dynamic innovation
follows.

\subsubsection{Stopping near the bottleneck and the global outcome}
Write $\mathcal E_n=Q_n-M\ell_M(n/M)$ and $\Delta U_2=U_2-Mu_2$. Sequential partner deletion gives the true
drift
\[
\E[Q_{n+1}-Q_n\mid\mathcal F_n]
=-1+(k-1)\frac{2U_2-Q_n}{Mm_M(n/M)}+O(M^{-1}).
\]
Each deletion changes $2U_2-Q$ by at most three, so this remainder does not require truncating the maximum
degree. After subtracting the fluid drift, the remaining main term is
\[
\frac{k-1}{Mm_M(n/M)}(2\Delta U_2-\mathcal E_n).
\]
Over a window of length $O(M\varepsilon_M)$ it accumulates to $O_p(\varepsilon_M\vartheta_M)=o_p(\sqrt M)$.
Multiplying the increments by $\mathbf1_{\{n<\theta_M\}}$ produces a stopped martingale on the original
probability space, with maximum fluctuation $O_p(\sqrt{M\varepsilon_M})$ inside the window. The last step can
reach $\theta_M$, so this control also covers the stopping moment.

At the entry write $\Delta U_j=U_j(n_-)-Mu_{M,j}(n_-/M)$. The Poisson thinning identity together with
handling the rounding gives
\[
x_c\sqrt M R_M-M\ell_M(t_{*,M})
=-\sum_{j\ge2}f_j(b_-)\Delta U_j+O(1).
\]
Since $|j-f_j(b_-)|\le C\varepsilon_Mj^2$, the above differs from the true entry fluctuation
$-\sum j\Delta U_j$ by $O_p(\varepsilon_M\vartheta_ML_M^3)+o(1)=o_p(\sqrt M)$. This proves
\eqref{eq:freeze}.

The global positivity of the fluid degree-one density comes from the fact that the function
$-\log(1-q)/q^{k-1}$ attains its minimum $\mu_c$ uniquely at $q_c$. Hence the critical fluid has only this
one interior zero before exhaustion. On a compact interval at fixed distance from the exhaustion of the
endpoints, the gap outside the window is at least $cM\varepsilon_M^2-O(\sqrt M)$, so the uniform bound before
stopping rules out a first stop in that interval and outside the window. Inside the window the lowest value
of the fluid difference is $-O(1)$, because the derivative at $t_{*,M}$ is $O(M^{-1/2})$ while the local
second derivative has a positive lower bound.

The rounding in the small-core count can also be handled with uniform constants. Let $v=\lfloor ke/2\rfloor$;
the vertex factor in the binomial upper bound is $e^v(v/M)^{ke-v}$. Its logarithmic derivative
$ke/v-\log(v/M)>0$, so one may replace $v$ by the real number $ke/2$ as an upper bound. Take
$C_{k,A}=e^{1+k/2}A(k/2)^{k/2}$ and choose $\delta_0$ so that $C_{k,A}\delta_0^{k/2-1}<1/2$. Summing the
part with $e\le\log M$ as a geometric series and the part with larger $e\le\delta_0M$ as $2^{-e}$, both
parts are $o(1)$. A path with positive score that reaches the moment with fewer than $\delta_0M$ edges
remaining would violate this bound if it still failed; a path with negative score has its degree-one count
drop to zero before the central time. Finally letting the safe interval $\zeta\downarrow0$ completes the
transfer to the real event.

\subsubsection{Limit after planting finitely many bad edges and the inclusion-exclusion expansion}
At time zero, expose all special edges and the initial degrees of the ordinary edges. Given this information,
the ordinary labeled stubs can be distributed uniformly into the slots of the ordinary edges; as each
ordinary edge is exposed, the assignment of the remaining slots is still uniform. When choosing the leaf with
the smallest vertex label, whether it belongs to a special edge is already predictable. The ordinary partner
pool differs from the full pool by at most $kl$ stubs, so the total variation of the per-step sampling kernel
is $O_l(M^{-1})$, and the corrections to the drift and bracket of the bounded projected increment are of the
same order. There are at most $l$ special steps, whose total contribution is $O_l(1)$.

Thus the initial limit, the uniform concentration, the arrival at the precritical time, and the
compensated-exponential argument all preserve the same limit after planting. Planting preserves the limiting
distribution, not the score of each individual graph; special edges may change the final outcome of an
individual graph.

Let $I_e$ be the indicator that edge $e$ is bad. After fixing $l$ edges with $I_e=1$, the remaining edges are
still independent, so the planting lemma applies. Invariance of the edge labels gives the mixed factorial
moment formula. Since $\sup_M dp_M<\infty$, the absolute value of the $l$-th term of the inclusion-exclusion
series is at most $\|f\|_\infty (dp_M)^l/l!$, and the tail of the series tends to zero uniformly in $M$. This
completes the transfer of the Gaussian law. The symmetric-difference error of the true failure is handled
separately by dividing by the probability of having no bad edges, which tends to a positive number; the two
transfers are done in two steps.

\subsubsection{From iid signs to a fixed total sign count}
Fix a deterministic graph satisfying \eqref{eq:selfavg}, and first handle the case $|\rho|<1$. Then
$\Omega=O(d)$ and $\Omega_\ell=\sum_i\ell_i^2=O(d)$. Take iid signs $t_i$ and let
\[
S=\sum_i t_i,\quad q=|S-d\rho_M|/2,\quad
\epsilon=\operatorname{sign}(S-d\rho_M),\quad J=\{i:t_i=\epsilon\}.
\]
Write $N=|J|$. Draw a uniformly chosen $q$-subset $U$ of $J$ and flip its signs. This mechanism is
equivariant under label permutations, and the output composition is fixed, so the output is exactly the
uniform distribution on that composition. Write $h_i=\sum_jm_{ij}t_j$ and
$I_U=\sum_{i<j,\ i,j\in U}m_{ij}$; expanding edge by edge gives
\[
Y_s-Y_t=-2\epsilon\sum_{i\in U}h_i+4I_U.\label{eq:batchflip}
\]
The edges inside $U$ have both endpoints flipped, so their actual contribution is unchanged; the first term
counts them twice, so the correction $4I_U$ cannot be omitted.

Consider the two cases of $\epsilon$ separately with $T_\epsilon=\sum_i\mathbf1_{\{t_i=\epsilon\}}h_i$, whose
expectation is $p_\epsilon\rho_M\sum_i \ell_i$ with $p_\epsilon=(1+\epsilon\rho_M)/2$. Redrawing $t_z$ changes
$T_\epsilon$ by at most $4\ell_z$, so the Efron--Stein inequality gives $\Var T_\epsilon\le
C\Omega_\ell=O(d)$. Both groups have size $p_\epsilon d+O_p(\sqrt d)$, so for both $\epsilon$ simultaneously
$\bar h_J=\rho_M\bar\ell+O_p(d^{-1/2})$. This also covers the random $\epsilon$ selected by $S$.

Given $t$, sampling without replacement gives the exact variance, together with the expectation bound on the
internal edges:
\begin{align*}
\Var\left(\sum_{i\in U}h_i\given t\right)
&=\frac{q(N-q)}{N(N-1)}\sum_{i\in J}(h_i-\bar h_J)^2
\le\frac{q\Omega_\ell}{N-1},\\
0\le\E[I_U\mid t]&\le\frac{q(q-1)}{N(N-1)}\Omega.
\end{align*}
On the event $q\le K\sqrt d$, $N\ge cd$, apply the conditional Chebyshev and Markov inequalities and then let
$K$ grow; the sampling error is $O_p(d^{1/4})$ and $I_U=O_p(1)$. For a degenerate group the variance is
simply set to zero. Combining $2\epsilon q=S-d\rho_M$ with \eqref{eq:batchflip} gives
\[
Y_s=Y_t-\rho_M\bar\ell(S-d\rho_M)+o_p(\sqrt d),
\]
which is the fixed-total-count projection of Section~\ref{sec:signproof}.

The counting behind the graph self-averaging is as follows. Two distinct overlap pairs sharing only one
element have a constant conditional mean by the symmetry of vertex permutations, so their covariance is zero.
Expanding $\sum_i \ell_i^2$ gives star-triple terms; if two star-triples share two or three elements, their
joint support graph is connected. Exposing independent $k$-subsets in order along a spanning tree, the
product expectations are $O(M^{-3})$ and $O(M^{-2})$ respectively, with $O(d^4)$ and $O(d^3)$ corresponding
terms, so the total variance is $O(M)$. One does not need to assume that all overlaps are mutually
independent.

After projection, the local terms of the linear and quadratic parts number $O(d)$ in total, and the size of
each term as well as the maximum degree of the dependency graph is $O(1+\Delta)$, where
$\Delta=\max_i \ell_i$. The connected cumulant of fixed order $r\ge3$ satisfies
\[
|\kappa_r|\le C_r d(1+\Delta)^{2r-1}=o(M^{r/2}).
\]
The variance converges as computed in Section~\ref{sec:signproof}, and the method of moments gives the normal
limit. For a random graph, take a subsequence of an arbitrary subsequence on which the relevant statistics
converge almost surely and $\Delta\le\log^2d$ holds eventually; arguing graph by graph on this subsequence
yields convergence in probability of the conditional characteristic function. The exact centering coefficient
$\gamma=(d\rho_M^2-1)/(d-1)$ differs from $\rho_M^2$ by a quantity whose product with $\Omega$ is $O(1)$.
The conditional characteristic function is bounded, so the convergence holds in $L^1$, and one may multiply by
any bounded graph-measurable test function to establish joint asymptotic independence. The endpoint
$|\rho|=1$ uses the exact total-variance identity of Section~\ref{sec:signproof}, and the above construction
is not performed on the rare sign class.

\subsubsection{Non-degeneracy of the constants}
The integrand of the dynamic variance is strictly positive on a positive-length interval before the critical
point, so $\tau_k^2>0$. For the initial-state part, the Cauchy--Schwarz inequality gives
$\eta_k^2\le \sigma_0^2/(\sigma_0^2+\tau_k^2)<1$, which guarantees that the density of the selection law is
strictly positive. Since $\beta_k<0$ and failure corresponds to a low score, the threshold is indeed
$z_r=-\beta_k r/\sigma_k$.

All remainders in the finite-dimensional joins are handled by Slutsky's theorem. Conditioning on having no
bad edges preserves $o_p(1)$ because the probability of that event has a positive lower bound. When the score
threshold is replaced by the true failure, the error of the joint CDF is bounded for all $x$ simultaneously
by the probability of the symmetric difference of the events. Finally, dividing by $\Phi(z_r)>0$ gives the
conditional distribution; the limit is continuous, so the CDF convergence is uniform. The convergence of the
conditional mean additionally uses the uniform integrability from $\E X^2=1$ to truncate; the convergence of
the quantile additionally uses the strict positivity of the limiting density. The argument above is for a
fixed-level quantile.

\section{Proof of the upper confidence bound}\label{app:upper}
This appendix proves Theorem~\ref{thm:upper} and states and proves the lemma it uses,
Lemma~\ref{lem:loadaxis}. We first sketch the ideas behind (ii) and (iii), and then give, in
order, the finiteness of the search range, the proofs of (i) and (ii), Lemma~\ref{lem:loadaxis},
and (iii).

(ii) follows from the compactness of $[-R,R]\times[-1,1]$ and a subsequence argument; on this range
$\Pr(F)$ has a uniform positive lower bound. (iii) proceeds in two steps. First, an edge-addition
coupling and the monotonicity of the 2-core give $\Pr(F)\to1$. Then the exact identity
$\mathrm{KS}(\text{failure},\text{all})=(1-p_F)\,\mathrm{KS}(\text{failure},\text{success})$
from \cite{prior}, together with the lemma below, shows that the conditional law tends to the
standard normal.

\begin{lemma}[Unconditional CLT on the load axis]\label{lem:loadaxis}
Fix $k\ge3$, use the Plain mapping, and assume $0<\liminf d_M/M\le\limsup d_M/M<\infty$, with the
signs a deterministic sequence independent of the mapping. Then $X\Rightarrow N(0,1)$.
\end{lemma}
The proof reuses three existing tools: the occupancy CLT from the de-Poissonization in
Appendix~\ref{sec:main-proof} (which only requires $\mu_M$ to lie in a positive compact interval);
the mixed factorial moments in the conditioning on having no bad edges; and Theorem~\ref{thm:sign}.
The complete proof appears below in the paragraph on Lemma~\ref{lem:loadaxis}. For
$D\sim\Pois(\mu)$, the identity $\Gamma(e,e)=2\mu^2$ holds for every $\mu$. The endpoint case
$|\rho|=1$ is handled by the exact second moment $\E V^2=1-\gamma^2\to0$.

\paragraph{Search range.} Let
$L=\min\{\Phi^{-1}(u),\min_{0\le c\le\max(\eta_k,0)}q_{\rm sel}(u;0,c)\}$. The selection law can be
written as $X=cY+\sqrt{1-c^2}Z$, where $Y$ is a standard normal truncated at $Y\le z$, and $Z$ is
independent of it. When $c\le0$, $X$ is stochastically larger than $N(0,1)$; when $c>0$, $X$ is
stochastically increasing in $z$. Hence when $z_M(n)\ge0$ we have $q_{\rm sel}\ge L$; together with
$s(n)\le n\sqrt{2/(M-1)}$, this gives $t_M(n,b)\ge n(1-|L|\sqrt{2/(M-1)})$. When $M>1+2L^2$, the
acceptance set lies in a finite interval. For $u=0.01$, $L=-2.416$ for $k=5$ and $L=-2.467$ for
$k=6$; for $k=3,4$, $L=\Phi^{-1}(0.01)$.

\paragraph{Theorem~\ref{thm:upper}(i).} Write
$G_{z,c}(x)=\Phi_2(x,z;c)/\Phi(z)$ and $q_M=q_{\rm sel}(u;z_M(d),\eta_k\rho_M^2)$. Since
$z_M(d)\to z_r$ and $\eta_k\rho_M^2\to\eta_k\rho^2$, and $G$ is jointly continuous and strictly
increasing in $x$, we have $q_M\to q_*$. The inequality $\widehat d<t_M(d,b)$ is equivalent to
$X<q_M$; by the uniform convergence in Theorem~\ref{thm:main} and the continuity of $G$,
$\Pr(X<q_M\mid F)\to u$. Since $\{d\le U\}\supseteq\{\widehat d\ge t_M(d,b)\}$, taking the limit
gives (i). For (ii), suppose the convergence is not uniform, and take a subsequence whose deviation
stays at least $\epsilon$; by compactness take a further subsequence with $r_M\to r_*$ and
$\rho_M\to\rho_*$; completing it to a full sequence satisfying the window assumption contradicts (i).

\paragraph{Lemma~\ref{lem:loadaxis}.} By a subsequence argument it suffices to treat the case
$d/M\to\alpha\in(0,\infty)$ and $\rho_M\to\rho$; write $\mu=k\alpha$. The de-Poissonization of
$e_M(j)=(j-\mu_M)^2$ in Appendix~\ref{sec:main-proof} uses only two facts: $\mu_M$ lies in a
positive compact interval, and the Poisson moments are uniformly bounded. Hence in the configuration
model $G_M(e_M)\Rightarrow N(0,\Gamma(e,e))$. For $D\sim\Pois(\mu)$, $\Var(e(D))=\mu+2\mu^2$ and
$\Cov(e(D),D)=\mu$, so $\Gamma(e,e)=2\mu^2$, and thus $X_{\rm one}\Rightarrow N(0,1)$. The transfer
to the Plain model reuses the mixed factorial moments: $l$ special edges change the energy by
$O_p(l\log M+l^2)$, $d\,p_M\to\alpha\binom k2$ is bounded, and the tail of the inclusion-exclusion
series is uniformly controlled. The three-term self-averaging of the overlap graph holds when
$d/M\to\alpha$, so Theorem~\ref{thm:sign} gives: when $|\rho|<1$,
$\mathcal L(V\mid\mathcal G)\Rightarrow N(0,1-\rho^4)$. Therefore
\begin{equation*}
\E e^{iuX}=\E\bigl[e^{iu\gamma X_{\rm one}}\E(e^{iuV}\mid\mathcal G)\bigr]\to e^{-u^2\rho^4/2}e^{-u^2(1-\rho^4)/2}=e^{-u^2/2}.
\end{equation*}
When $|\rho|=1$, $\E V^2=1-\gamma^2\to0$.

\paragraph{Theorem~\ref{thm:upper}(iii).} Adding independent edges one by one couples the mappings
for different $d$ so that they increase with $d$. The property that the 2-core is nonempty is
monotone under edge addition, so $\Pr(F)$ is nondecreasing in $d$. For any $R$, $d_M$ eventually
exceeds the load corresponding to the window coordinate $R$, hence $\liminf\Pr(F)\ge\Phi(z_R)$;
letting $R\to\infty$ gives $\Pr(F)\to1$. This identity gives
$\mathrm{KS}(\text{failure},\text{all})\le1-p_F\to0$, which together with Lemma~\ref{lem:loadaxis}
gives $X\mid F\Rightarrow N(0,1)$. When $z_M(d)\to+\infty$, $q_{\rm sel}\to\Phi^{-1}(u)$, and the
remaining argument is the same as the second half of (i).

\section{Proof of the load-axis results}\label{app:loadaxis}
\subsection{Macro moderate deviation probabilities}\label{app:mda}
\begin{proposition}[Moderate deviation probability of a macro core]\label{prop:mda}
Fix $k\ge3$, $d/M=\alpha_c+r_M/\sqrt M$, with $r_M\to-\infty$ and
$|r_M|\le C\sqrt{\log M}$. Let $\mathcal A$ be the event that the actual 2-core has more than
$\varepsilon_0M$ edges, with $0<\varepsilon_0<\alpha_cq_c^k$. Then in the Plain model
$\Pr(\mathcal A)=\Phi(z_r)(1+o(1))$, where $z_r=-\beta_kr_M/\sigma_k$.
\end{proposition}
Proposition~\ref{prop:mda} gives the relative probability of a macro failure between the critical
window and low load. Equation (2.18) of Dembo--Montanari gives an expansion for the configuration
model for each fixed $r$, with additive error of order $n^{-\eta}$ ($\eta<5/26$)\cite{dm}. The
additive error converts to relative precision only when $\Phi(-t)\gg n^{-\eta}$.
Proposition~\ref{prop:mda} gives a relative asymptotic over the entire range
$|r_M|=O(\sqrt{\log M})$ and transfers the result to the Plain model.

The next four subsections prove Proposition~\ref{prop:mda}. Write $L=\lceil\log M\rceil$,
$\vartheta_M=\sqrt M(\log M)^6$, $\varepsilon_M=M^{-1/8}$, and $t=-z_r$. ``Superpolynomially
small'' means $o(M^{-b})$ for every fixed $b>0$.

\subsubsection{The event on which the concentration estimates hold}
The estimates used for the uniform concentration before stopping in Appendix~\ref{sec:main-proof}
hold uniformly on a fixed compact neighborhood of $\mu_M=kd/M$. Concretely: the tail probability
of the maximum degree is superpolynomially small; for a fixed target $(n,j)$ the difference of the
backward initial-state functional has a constant upper bound, the dynamic martingale difference
also has a constant upper bound, and with threshold $\sqrt M(\log M)^2$ the McDiarmid and Azuma
inequalities give a tail of $\exp(-c\log^4M)$; under the maximum-degree truncation the accumulated
drift is $O(L^4)$. Taking a union over $O(ML)$ targets yields a good event $G_M$ whose complement
has superpolynomially small probability. On $G_M$, up to the stopping time (including the stopping
instant itself), the error of every histogram is at most $CB_M$. Before $n_-$ the fluid degree-one
density is at least $cM\varepsilon_M^2-C\sqrt M|r_M|$, which is far larger than $\vartheta_M$, so
on $G_M$ the process does not stop before $n_-$.

\subsubsection{Concatenating the initial and dynamic stages into a martingale}
The configuration model can be generated as follows: first draw $kd$ iid uniform endpoints, then
assign the labeled stubs to the slots of the edges independently and uniformly. The initial-state
filtration records the endpoint sequence; the dynamic filtration adds to this the part of the
pairing revealed so far by the peeling process. Given the endpoint sequence, the unrevealed pairing
is still uniform, so the initial Doob increments and the dynamic increments can be concatenated
into a single martingale difference sequence. For $h_c(j)=j(1-a_c)^{j-1}$, the unnormalized Doob
increment of the $(l+1)$-th endpoint is
\begin{gather*}
\Delta_l=\nu_{n-l-1}\bigl(D_V^{(l)}\bigr)-\frac1M\sum_v\nu_{n-l-1}\bigl(D_v^{(l)}\bigr),\\
\nu_m(j)=\E h_c\bigl(j+1+\mathrm{Bin}(m,1/M)\bigr)-\E
h_c\bigl(j+\mathrm{Bin}(m,1/M)\bigr),
\end{gather*}
This increment is bounded. Changing one endpoint changes the initial-state bracket $V_{\rm in}$ by
at most $C/M$, so by McDiarmid's inequality $V_{\rm in}$ concentrates within a band of width
$(\log M)^2/\sqrt M$. A direct computation of $\Var\left(\sum_vh_c(D_v)\right)$ (retaining the
$1/M$ terms in the two-vertex covariance) gives
$\E V_{\rm in}=\sigma_0^2/\sigma_k^2+O((1+|r_M|)/\sqrt M)$. On $G_M$ the dynamic bracket differs
from $\tau_k^2/\sigma_k^2$ by $O(M^{-1/8}+(\log M)^8/\sqrt M+|r_M|/\sqrt M)$. Hence the
concatenated martingale $S_M$ has increments of size at most $C/\sqrt M$; its bracket has a
deterministic upper bound and, except for a superpolynomially small exceptional probability, lies
within $1\pm v_M$, where $v_M=C\{M^{-1/8}+(\log M)^8/\sqrt M\}$.

\subsubsection{Exponential tilting and the relative tail probability}
Let $m_j=\E[e^{-t\xi_j}\mid\mathcal F_{j-1}]$ and define the measure $Q$ with density
$\prod_je^{-t\xi_j}/m_j$. On the good set,
$\prod_jm_j=e^{t^2/2}(1+O(t^2v_M+t^3/\sqrt M))$. The bound $Z_t^2\le Z_{2t}e^{Ct^2}$ gives the
second-moment bound $\E_P(dQ/dP)^2\le e^{Ct^2}$, so the bad set remains negligible under $Q$ and
the total variation between the exponential tilting $P_t$ and $Q$ tends to zero. Under $Q$ the
conditional mean of the increment is $-tv_j+O(t^2v_j/\sqrt M)$; recentering then makes the bracket
close to $1$. For a centered martingale with increment bound $\beta'$ and bracket exceptional
probability $q'$, the telescoping expansion of the compensating exponential and Esseen smoothing
give
\begin{equation*}
\sup_x|P(T\le x)-\Phi(x)|\le C(\beta'+v'+q')^{1/8}+Cq'.
\end{equation*}
At the magnitudes of the present model, the Kolmogorov error $\epsilon_M$ of the tilted
distribution satisfies $t\epsilon_M\to0$. The untilting formula is
\begin{equation*}
\Pr(S_M\le-t)=\E e^{-tS_M}\,e^{-t^2}\,\E_t\bigl[e^{t(S_M+t)}\mathbf1_{\{S_M+t\le0\}}\bigr]
\end{equation*}
Integrating by parts gives $\Pr(S_M\le-t)=\Phi(-t)(1+o(1))$. The same argument gives
$\Pr(|S_M+t|\le\delta)=o(\Phi(-t))$ for bandwidths $\delta$ satisfying $t\delta\to0$.

\subsubsection{Transfer to the actual macro failure and the Plain model}
On $G_M$, $R_M/\sigma_k=S_M+b_M+O((\log M)^4/\sqrt M)$; the exact mean expansion gives
$b_M=-g_kr_M+O((1+r_M^2)/\sqrt M)$, where $g_k=ka_ce^{-x_c}/\sigma_k$. Replacing the $O_p$ bound
from the local freezing by a deterministic bandwidth $\delta_M\to0$ (satisfying
$b_M\delta_M\to0$) at the cost of a superpolynomially small exceptional probability gives
\begin{equation*}
\mathcal A\,\triangle\,\{S_M\le-b_M\}\subseteq\mathcal B_M\cup\{|S_M+b_M|\le C\delta_M\},
\end{equation*}
where $\mathcal B_M$ has superpolynomially small probability. When the score lies below the
negative bandwidth the process stops while the residual is still of linear size, so it lies in
$\mathcal A$; when the score lies above the positive bandwidth the process does not stop in the
window or on either side of it, until the residual drops below $\varepsilon_0M/2$, so it does not
lie in $\mathcal A$. This step does not require excluding small cores. Combining the relative tail
estimate with the bandwidth estimate gives the result in the configuration model.

Then transfer to the Plain model. When $l$ special edges are fixed, the change of
$\sum_vh_c(D_v)$ induced by the special edges has a deterministic upper bound and the per-step
error of the sampling kernel is $O_l(1/M)$, so all the rates above remain unchanged and
$\Pr(\mathcal A\mid I_1=\cdots=I_l=1)/p_M\to1$. Applying finite-order upper and lower Bonferroni
truncations to the indicator of having no bad edges, first letting $M\to\infty$ and then letting
the truncation order tend to infinity, gives $\Pr(\mathcal A,N^{\rm bad}_M=0)/p_M\to
e^{-\Lambda_k}$; dividing by $\Pr(N^{\rm bad}_M=0)\to e^{-\Lambda_k}$ gives the result in the
Plain model. The key point is that every fixed planted term already has precision relative to
$p_M$.
\subsection{Fixed subcritical load}\label{app:lowload}
\begin{proposition}[Fixed subcritical load]\label{prop:lowload}
Fix $k\ge3$ and $d/M\to\alpha\in(0,\alpha_c)$. Write $\lambda=\binom d2/\binom Mk$, let $P_2$ be
the event that the actual core is exactly two repeated edges, and let
$\gamma=(d\rho_M^2-1)/(d-1)$ be the sign factor of Section~\ref{sec:sign}. There exists
$\delta_k\ge1$ such that
\begin{equation*}
\Pr(F)=\lambda\{1+O(M^{-\delta_k})\},\qquad \Pr(P_2\mid F)\to1,\qquad X\mid F\Rightarrow N(0,1),
\end{equation*}
\begin{equation*}
\E[X\mid F]=\frac{2\gamma}{s_M}+O(|\gamma|M^{-\delta_k}),\qquad c_F\binom d2\longrightarrow1.
\end{equation*}
\end{proposition}
Proposition~\ref{prop:lowload} shows that at a fixed subcritical load the failure is almost always
caused by two repeated edges; in this regime the standardized energy tends again to normality, and
the conditional mean on the original scale retains the leading term $2\gamma/s_M$, whose origin is
explained in Section~\ref{sec:loadaxis}.

This subsection proves Proposition~\ref{prop:lowload}. Write $K=\binom Mk$, $J=\binom d2$,
$N=\sum_{i<j}\mathbf1_{\{S_i=S_j\}}$, $D=\{N\ge1\}$, and $\lambda=J/K$. Let $a=(k-2)/2$,
$b_k=\lceil3(k-2)/2\rceil$, $\delta_k=\min\{k-2,\,b_k-(k-2)\}\ge1$,
$W_{ij}=|S_i\cap S_j|-k^2/M$, and
$Y=\widehat d-d=\frac{2}{k(1-k/M)}\sum_{i<j}s_is_jW_{ij}$.

\subsubsection{Uniform counting of small cores}
A 2-core subgraph with $e$ edges uses at most $v_e=\lfloor ke/2\rfloor$ vertices. For a prescribed
set $I$ of $e$ edge labels and a $v_e$-element vertex set $V$, let
$A_{I,V}=\{S_i\subseteq V,\ i\in I\}$; every such core lies in some $A_{I,V}$. The first moment
is
\begin{equation*}
B_e=\binom de\binom M{v_e}\Bigl(\binom{v_e}k\Big/\binom Mk\Bigr)^e\le\bigl[C(e/M)^a\bigr]^e,
\end{equation*}
This bound holds uniformly for $d\le AM$ and $e\le\varepsilon M$ ($\varepsilon<1/k$). For $e=3$,
retaining the rounding gives $B_3=O(M^{-b_k})$. For $e\ge4$ the ratio of adjacent terms is at
most $Ce^a((e+1)/M)^a$. For fixed $r\ge0$, choose $\varepsilon$ small enough that the adjacent
ratio of $e^rB_e$ is at most $1/2$ on $4\le e<\varepsilon M$; then
\begin{equation}
\sum_{e=3}^{\lfloor\varepsilon M\rfloor}e^rB_e=O(M^{-b_k}).\label{eq:smallcore}
\end{equation}
This sum also covers the intermediate range with $e\to\infty$ and $e=o(M)$.

\subsubsection{Superpolynomially small probability of a linear-size core}
For every fixed $\varepsilon>0$ and $B>0$, $\Pr(|\mathrm{core}|\ge\varepsilon M)=o(M^{-B})$. The
proof uses message passing of fixed depth. Let $Z_T$ be the number of edges whose incoming
messages are all $1$ at round $T$; then $|\mathrm{core}|\le Z_T$ on every graph. Local convergence
to a Poisson branching tree gives $\E Z_T/M\to\alpha q_T^k$, where $q_0=1$ and
$q_{t+1}=1-\exp(-k\alpha q_t^{k-1})$. When $\alpha<\alpha_c$ this recursion has no positive fixed
point, so one can choose $T$ with $\alpha q_T^k<\varepsilon/4$. On the set $H$ of graphs with
maximum degree at most $\lceil\log M\rceil$, changing $h$ edge coordinates changes $Z_T$ by at
most $L_Mh$, where $L_M$ is of polylogarithmic order. Applying a McShane extension of $Z_T|_H$
truncated to $[0,d]$, and then McDiarmid's inequality, gives
$\Pr(Z_T\ge\varepsilon M)\le\Pr(H^c)+\exp\{-cM/L_M^2\}=o(M^{-B})$.

\subsubsection{Failure probability and the conditional law}
A single edge cannot form a core; two edges form a core if and only if they are identical. Hence
on $F\setminus P_2$ the actual core has at least three edges, and by the previous two subsections
$\Pr(F\setminus P_2)=O(M^{-b_k})$. The probability that any two distinct edge pairs are
simultaneously repeated is $K^{-2}$, so $\E N=\lambda$ and $\E[N(N-1)]\le\lambda^2$; by
Bonferroni's inequality $\Pr(F)=\lambda+O(\lambda^2+M^{-b_k})$.

Let $Q$ be the distribution obtained by planting a repeated pair at random, whose density with
respect to the original distribution is exactly $N/\lambda$; then
$\|P(\cdot\mid D)-Q\|_{\rm TV}=O(\lambda)$. Since $D\subseteq F$,
$\|P(\cdot\mid F)-P(\cdot\mid D)\|_{\rm TV}=\Pr(F\setminus D)/\Pr(F)$. Combining the two bounds
gives $\|P(\cdot\mid F)-Q\|_{\rm TV}=O(M^{-\delta_k})$. This is a total variation bound over the
entire labeled mapping graph.

\subsubsection{CLT under the unconditional law and under the repeated-edge conditioning}
Let $v_i=\mathbf1_{S_i}-(k/M)\mathbf1$ and reveal the edges in order of their labels. The
martingale difference is $\Delta_j=\frac{2}{k(1-k/M)}s_jv_j^\top\sum_{i<j}s_iv_i$, and
$Y=\sum_j\Delta_j$. The sum of conditional variances is
\begin{gather*}
\mathcal
V_d=s_M^2+\frac{8}{k(1-k/M)(M-1)}\sum_{i<\ell}(d-\ell)s_is_\ell
W_{i\ell},\\
\Var\mathcal
V_d=\frac{64}{(M-1)^3}\sum_{\ell=2}^d(\ell-1)(d-\ell)^2=O(M),
\end{gather*}
so $\mathcal V_d/s_M^2$ converges to $1$ in $L^2$. The bound $\E\Delta_j^4=O(1)$ gives the
conditional Lindeberg condition, so the martingale CLT yields $X\Rightarrow N(0,1)$, uniformly in
any deterministic signs. Fix a planted pair $(a,b)$, write their common mapping as $T$ and their
sign product as $\sigma$; then
\begin{equation*}
Y=Y_R+2\sigma+\frac{2}{k(1-k/M)}(s_a+s_b)\sum_{j\in R}s_j\bigl(|T\cap S_j|-k^2/M\bigr),
\end{equation*}
The variance of the last term is $O(1)$. Applying the unconditional CLT to the remaining edges and
then transferring back to $P(\cdot\mid F)$ via the total variation bound completes the proof.

\subsubsection{Expansion of the conditional mean}
The total variation bound does not by itself control the mean, so
$\E[|Y|\,|\mathbf1_F-N|]$ must be estimated directly. Given $A_{I,V}$ (with $|I|=e$), splitting
$Y$ into an internal, a cross, and an external part gives $\E[|Y|\mid A_{I,V}]\le C(\sqrt
M+e^2)$. By \eqref{eq:smallcore} (with $r=0,2$) and the superpolynomial bound for linear-size
cores, $\E[|Y|;F,N=0]=O(\sqrt M M^{-b_k})$. When two prescribed repeated pairs occur
simultaneously, $\E|Y|\le C\sqrt M$, so $\E[|Y|(N-1);N\ge1]\le C\sqrt M\lambda^2$. When a
prescribed edge pair $(i,j)$ is repeated, $\E Y=2s_is_j$, so $\E[YN]=2\gamma\lambda$, and hence
\begin{equation*}
\E[Y;F]=2\gamma\lambda+O\bigl(\sqrt M(M^{-b_k}+\lambda^2)\bigr).
\end{equation*}
Since $F$ is invariant under permutations of the edge labels, for each $M$ exactly
$\E[X\mid F]=\gamma\E[X_+\mid F]$, where $X_+$ is the measurement with all-positive signs.
Applying the above identity in the all-positive case and multiplying by $\gamma$ gives the
remainder with factor $|\gamma|$ in the mean expansion. Matching the coefficients with
\cite{prior} in the all-positive case gives
$c_F\binom d2=\tfrac12\E[Y_+\mid F]=1+O(M^{1/2-\delta_k})$.
\subsection{The macro conditional law and the two-mechanism mixture}\label{app:mixture}
\subsubsection{Setup and the results used}

This subsection fixes $k\ge3$, uses the Plain model, and places the load in the moderate
deviation regime:
\begin{equation}
\frac dM=\alpha_c+\frac{r_M}{\sqrt M},\qquad r_M\to-\infty,\qquad |r_M|\le C\sqrt{\log M}.\label{F:range}
\end{equation}
Write $t_M=-g_kr_M\to\infty$, where $g_k=ka_ce^{-x_c}/\sigma_k$;
$p_M=\Phi(-t_M)$; and $\lambda_M=\binom d2/\binom Mk$. The macro event is
$\mathcal A=\{\text{number of edges of the actual 2-core}>\varepsilon_0M\}$. The sign factor is
$\gamma=(d\rho_M^2-1)/(d-1)$; the all-positive measurement is denoted $X_+$, and the measurement
under an arbitrary deterministic sign composition is denoted $X_s$. The left-tail Mills ratio is
written as $A(-t)=\phi(t)/\Phi(-t)$ following the main text. Let
\begin{equation*}
G_{M,s}(x)=\frac{\Phi_2(x,-t_M;\eta_k\rho_M^2)}{\Phi(-t_M)}.
\end{equation*}

\begin{proposition}[Macro conditional law and two-mechanism mixture]\label{prop:mixture}
Within the range \eqref{F:range}, take $\varepsilon_0$ small enough (Lemma~\ref{F:lem:cutoff}),
let the signs be an arbitrary deterministic composition, and do not require $\rho_M$ to converge.
Then
\begin{gather*}
\sup_x\bigl|\Pr(X\le x\mid \mathcal A)-G_{M,s}(x)\bigr|\to0,\qquad
\Var(X\mid \mathcal A)=1-\eta_k^2\rho_M^4+o(1),\\
\sup_x\left|\Pr(X\le x\mid
F)-\frac{p_MG_{M,s}(x)+\lambda_M\Phi(x)}{p_M+\lambda_M}\right|\to0.
\end{gather*}
\end{proposition}
The first two displays say that in the moderate deviation regime a macro failure still obeys the
selection law of the critical window, with only its center shifted by $t_M$. The last display says
that the failure-conditioned distribution is a two-layer weighted mixture whose weights $p_M$ and
$\lambda_M$ are both explicit.

This subsection and Appendix~\ref{app:wm} prove four groups of statements, which correspond to
the main text as follows.
\begin{enumerate}[label=(\arabic*)]
\item The macro conditional CDF with all-positive signs:
$\sup_x|\Pr(X_+\le x\mid \mathcal A)-\Phi_2(x,-t_M;\eta_k)/\Phi(-t_M)|\to0$ (Proposition~\ref{F:prop:cdf}).
\item The conditional moments with all-positive signs and the conditional mean under an arbitrary
composition (Proposition~\ref{F:prop:moments}).
\item The macro conditional CDF and conditional variance under an arbitrary composition, namely
the first two displays of Proposition~\ref{prop:mixture} (Proposition~\ref{F:prop:sign}).
\item The control of the three weighted error terms, namely the weighted mean formula of
Proposition~\ref{prop:wm}, proved in Appendix~\ref{app:wm}.
\end{enumerate}

We use directly the scalar inputs already established in Appendix~\ref{app:mda} (the proof of
Proposition~\ref{prop:mda}). For ease of reference they are collected here.

\begin{enumerate}[label=(C\arabic*)]
\item There exists a good event $G_M$ whose complement has superpolynomially small probability. On
$G_M$, up to the stopping time, the error of every histogram is at most $C\vartheta_M$, with
$\vartheta_M=\sqrt M(\log M)^6$.
\item In the configuration model, the initial Doob increments and the dynamic increments
concatenate into a score martingale $S$; its increments are at most $C/\sqrt M$, its bracket has a
deterministic upper bound, and except for a superpolynomially small exceptional probability it
lies within $1\pm v_M$, with $v_M=C\{M^{-1/8}+(\log M)^8/\sqrt M\}$.
\item There is a bandwidth sandwich between the macro event and the score event:
$\mathcal A\,\triangle\,\{S\le-b_M\}\subseteq\mathcal B_M\cup\{|S+b_M|\le C\delta_M\}$, where
$\mathcal B_M$ is superpolynomially small, the deterministic bandwidth $\delta_M\to0$ satisfies
$b_M\delta_M\to0$, and $b_M=t_M+O((1+r_M^2)/\sqrt M)$.
\item The relative asymptotic of the macro probability: $\Pr(\mathcal A)=p_M(1+o(1))$; the same
holds with $l$ fixed special edges, and transfers to the Plain model via a relative Bonferroni
bound.
\end{enumerate}

``Superpolynomially small'' means $o(M^{-b})$ for every fixed $b>0$. By \eqref{F:range}, $p_M$ is
no smaller than a fixed negative power divided by a polylogarithm, so a superpolynomially small
quantity divided by $p_M$ is still superpolynomially small. This is used repeatedly below.

\subsubsection{Choosing the lower bound on the macro core edge count}

Proposition~\ref{prop:mda} holds for any $0<\varepsilon_0<\alpha_cq_c^k$. The two-mechanism
reduction also needs the small-core counting to cover all core sizes below the macro layer, so
here a small enough cutoff is fixed.

\begin{lemma}[Choice of the cutoff]\label{F:lem:cutoff}
Take $\varepsilon_*>0$ such that the small-core sum \eqref{eq:smallcore} holds for
$r\in\{0,2\}$ on $3\le e\le\varepsilon_*M$, and $\varepsilon_*<\alpha_cq_c^k/2$. Fix
$0<\varepsilon_0\le\varepsilon_*$. Let $D$ be the event that some repeated edge pair exists, and
$R=F\setminus(\mathcal A\cup D)$. Then under \eqref{F:range}
\begin{equation*}
\Pr(R)=O(M^{-b_k})=o(\lambda_M),\qquad \Pr(\mathcal A\cap D)=o(\lambda_M).
\end{equation*}
\end{lemma}

\begin{proof}
On $R$ the actual core has between $3$ and $\varepsilon_0M$ edges. The first-moment sum of
Appendix~\ref{app:lowload} requires only $d\le AM$, independent of the distance to the threshold,
so $\Pr(R)=O(M^{-b_k})$; and since $b_k-(k-2)\ge1$, this is $o(\lambda_M)$.

For the intersection, first consider a fixed window. Let
$d/M=\alpha_c+r/\sqrt M+o(M^{-1/2})$ with $r$ fixed, plant a pair of identical uniform
$k$-subsets, and let the remaining edges be independent. The finite-planting lemma of
Appendix~\ref{sec:mp-plain} allows the special edges to be correlated, so the precritical score
still has the same Gaussian limit; the two repeated edges are themselves valid $k$-subsets, the
probability of no bad edge among the remaining ordinary edges has a positive limit, and the
inclusion-exclusion transfer proceeds as usual. The planted repeated edge forces the final
failure, so here only the macro event is transferred: when the score lies below a fixed negative
gap the process stops with $(\alpha_cq_c^k+o(1))M$ edges remaining, which lies in $\mathcal A$;
when the score lies above a fixed positive gap the pre-stopping concentration and local freezing
rule out early stopping, until the residual drops below $\varepsilon_0M/2$, which does not lie in
$\mathcal A$. Removing the safety gap gives $Q_r(\mathcal A)\to\Phi(g_kr)$, where $Q_r$ is the
planted distribution.

Return to \eqref{F:range}. For any fixed $r_0<0$, eventually
$d\le\lfloor\alpha_cM+r_0\sqrt M\rfloor$. Coupling by adding the same sequence of ordinary edges,
the macro event is monotone under edge addition, so
$\limsup_MQ_M(\mathcal A)\le\Phi(g_kr_0)$. Letting $r_0\to-\infty$ gives $Q_M(\mathcal A)\to0$.
The density of the random-pair planted distribution with respect to the original distribution is
$N/\lambda_M$, where $N$ is the number of repeated edge pairs, so
\[
\Pr(\mathcal A\cap D)\le\E[N\mathbf1_{\mathcal A}]=\lambda_MQ_M(\mathcal A)=o(\lambda_M).\qedhere
\]
\end{proof}

The probability asymptotic of Proposition~\ref{prop:mda} holds for any
$\varepsilon_0<\alpha_cq_c^k$, so the macro probability itself is independent of the choice of
cutoff. Proposition~\ref{prop:mixture}, Proposition~\ref{prop:wm}, and the results in this
subsection conditioned on $\mathcal A$ all refer to the cutoff selected in
Lemma~\ref{F:lem:cutoff}.

\subsubsection{Initial-state Doob decomposition of the measurement and the score with all-positive signs}

Work first in the configuration model. Let $n_e=kd$ be the total number of endpoints,
$\mu=n_e/M$, the endpoints $V_1,\ldots,V_{n_e}$ independent and uniform, and $D_v$ the occupancy
count. Define
\begin{equation*}
K=\sum_v(D_v)_2,\qquad Y_K=\frac{K-\E K}{\sqrt2\mu\sqrt M},\qquad H=\sum_vh_c(D_v),
\end{equation*}
where $(j)_2=j(j-1)$, $h_c(j)=j(1-a_c)^{j-1}$, and $\E K=n_e(n_e-1)/M$. $Y_K$ is the
standardized occupancy energy, which is converted back to $X_+$ in \eqref{F:affine} at the end.

After revealing the first $l$ endpoints, write the occupancy count as $D_v^{(l)}$ and let
$m=n_e-l-1$. For a function $f$ define
$\nu^f_m(j)=\E[\Delta f(j+\mathrm{Bin}(m,1/M))]$, where $\Delta f(j)=f(j+1)-f(j)$. When the next
endpoint $V$ falls, the Doob increment of $\sum_vf(D_v)$ is
\begin{equation*}
\Delta_f=\nu^f_m(D_V^{(l)})-\frac1M\sum_v\nu^f_m(D_v^{(l)}).
\end{equation*}
Taking $f(j)=(j)_2$ gives $\nu^f_m(j)=2j+2m/M$, so the Doob increment of the energy is exactly
$2(D_V^{(l)}-l/M)$. The total number of endpoints is fixed sample by sample; here the Poisson
ensemble does not replace the true initial state.

\begin{lemma}[Mean of the bracket]\label{F:lem:bracket}
Let the matrix bracket of the initial-state vector $(Y_K,(H-\E H)/(\mu_c\sigma_k\sqrt M))$ be
$(V_{EE},V_{EH},V_{HH})$. Then
\begin{equation}
\E V_{EE}=(1-1/n_e)(1-1/M),\qquad \E V_{EH}=\eta_k+O\bigl((1+|r_M|)/\sqrt M\bigr),\label{F:bracketmean}
\end{equation}
$\E V_{HH}$ is the same as the initial-state bracket of Appendix~\ref{app:mda}.
\end{lemma}

\begin{proof}
Expanding the collision indicators gives $\Var K=2n_e(n_e-1)(M-1)/M^2$. For any $h$, let
$B\sim\mathrm{Bin}(n_e-2,1/M)$; there is the finite identity
\begin{equation*}
\Cov(K,H)=\frac{n_e(n_e-1)(M-1)}{M^2}\E\Delta^2h(B).
\end{equation*}
Here is the proof. Write $K=\sum_{i\ne j}\mathbf1_{\{V_i=V_j\}}$. Conditioned on $V_1=V_2$, the
mean of $H$ is $\E h(B+2)+(M-1)\E h(B)$; the unconditional mean is
$M\E h(B)+2\E\Delta h(B)+M^{-1}\E\Delta^2h(B)$. Their difference is
$(1-M^{-1})\E\Delta^2h(B)$; multiplying by $n_e(n_e-1)/M$ gives the claim.

For $h=h_c$, the Poisson limit gives
$\E_{\Pois(\mu)}\Delta^2h_c(D)=a_ce^{-\mu a_c}(\mu a_c-2)$. The function $\Delta^2h_c$ is
bounded, the total variation error between Binomial and Poisson is $O(M^{-1})$, and the parameter
$\mu$ lies in a compact interval and the dependence is smooth. By the Doob isometry, $\E V_{EE}$
and $\E V_{EH}$ are the corresponding variance and covariance divided by the normalization
constant; substituting $\mu=\mu_c+O(|r_M|/\sqrt M)$ gives \eqref{F:bracketmean}. The limit of the
cross term is exactly $\eta_k$, which comes from the true occupancy model and does not presuppose
independence between the dynamic noise and the energy.
\end{proof}

The maximum occupancy count can be of order $\log M$, so the energy increments are unbounded. To
obtain concentration, take $L=\lceil\log M\rceil$ and define the truncated energy $e_L$ by
$e_L(0)=0$ and $\Delta e_L(j)=2\min(j,L)$. For $j\le L+1$ we have $e_L(j)=(j)_2$, and
$|\Delta e_L|\le2L$, $|\Delta^2e_L|\le2$. The true Doob martingale is reconstructed from $e_L$;
the large increments of the original martingale are not simply deleted.

\begin{lemma}[Superpolynomial concentration of the truncated bracket]\label{F:lem:conc}
For $(f,g)\in\{(e_L,e_L),(e_L,h_c),(h_c,h_c)\}$,
\begin{equation}
\Pr\Bigl(|V^{(L)}_{fg}-\E V^{(L)}_{fg}|>CL^2(\log M)^2/\sqrt M\Bigr)\le2e^{-c(\log M)^4},\label{F:conc}
\end{equation}
and the difference between $\E V^{(L)}_{fg}$ and the mean of the untruncated bracket, multiplied
by any fixed power of $M$, still tends to zero.
\end{lemma}

\begin{proof}
Changing one input endpoint changes, at each subsequent $l$, two occupancy coordinates by $1$
each. For $e_L$, $\nu^{e_L}_m$ is bounded by $O(L)$ and a single coordinate changes by $O(1)$;
for $h_c$ both are $O(1)$. So each
$q_l^{fg}=M^{-1}\sum_v\nu^f_m\nu^g_m-\bar\nu^f_m\bar\nu^g_m$ changes by at most $CL^2/M$; after
normalization and summing over $O(M)$ steps, the sensitivity of the whole bracket to each input
coordinate is at most $CL^2/M$. McDiarmid's inequality gives \eqref{F:conc}.

The probability that the maximum occupancy count exceeds $L$ is at most $M(eC/L)^L$, smaller than
any fixed negative power of $M$. On its complement $K$ equals the truncated energy, and both are
at most $O(M^2)$ in absolute value on the whole sample space. Hence the differences in the center,
the variance, and the covariance with $H$, multiplied by any fixed power of $M$, still tend to
zero. Here only the replacement of the mean is needed; the truncated bracket need not equal the
original bracket path by path.
\end{proof}

\subsubsection{The dynamic stage and the deterministic upper bound on the bracket}

After the initial state, the labeled stubs are assigned to the edges independently and uniformly,
and the dynamic filtration follows the construction of Appendix~\ref{app:mda}, without revealing
the full pairing in advance. The initial-state increment of the two-dimensional vector is
$(\Delta e_L/(\sqrt2\mu\sqrt M),\Delta h_c/(\mu_c\sigma_k\sqrt M))$, and the dynamic increment is
$(0,\text{score martingale increment})$. In the dynamic stage the first coordinate is zero, so it
contributes no new $EE$ or $EH$ bracket; this is an identity on the filtration and does not
require an independence assumption.

By Lemma~\ref{F:lem:bracket}, Lemma~\ref{F:lem:conc}, and input (C2), except for a
superpolynomially small exceptional set the total matrix bracket $\mathcal V$ satisfies
\begin{equation}
\|\mathcal V-\Sigma\|\le v_M,\qquad \Sigma=\begin{pmatrix}1&\eta_k\\\eta_k&1\end{pmatrix}.\label{F:Sigma}
\end{equation}
$\Sigma$ is positive definite: the initial-state energy has variance $1$, and the score has
strictly positive dynamic variance $\tau_k^2$, so $|\eta_k|<1$.

After truncation the increment bound is $\beta_M=CL/\sqrt M$, but on the whole sample space the
trace of the bracket is only known to be $O(L^2)$. To perform the exponential tilting, the entire
vector is stopped when the trace of the bracket first exceeds the fixed constant
$C_0>\operatorname{tr}\Sigma+1$. The bracket increments are positive semidefinite matrices, so the
trace is monotone; the stopping changes the terminal value only when the total trace exceeds
$C_0$, an event of superpolynomially small probability. After the stopping,
$\operatorname{tr}\mathcal V\le C_0+C\beta_M^2$ holds deterministically. Below, $(Y^*,S^*)$
denotes this stopped vector. On the good set it equals the original score and energy up to a
deterministic center term that is superpolynomially small. The subsequent exponential tilting acts
only on $S^*$.

\subsubsection{Error rate of the two-dimensional distribution function under exponential tilting}

This gives the explicit polynomial rate needed for the macro conditional law.

\begin{lemma}[Single-coordinate tilting lemma]\label{F:lem:tilt}
Let a two-dimensional martingale difference $\zeta_j=(u_j,v_j)$ satisfy $|\zeta_j|\le\beta$, let
the trace of the total matrix bracket $\mathcal V$ have a deterministic upper bound, and let
$\Pr(\|\mathcal V-\Sigma\|>v)\le q$, where $\beta,v$ are polynomially small in $M$ and $q$ is
superpolynomially small. Write $(U,S)=\sum_j\zeta_j$, let $1\le t\le C\sqrt{\log M}$, and define
the tilted measure $dP_t/dP=e^{-tS}/\E e^{-tS}$. Then the lower-rectangle CDF distance
$\epsilon_M$ between $(U+\eta_kt,S+t)$ under $P_t$ and $N(0,\Sigma)$ satisfies
\begin{equation}
\epsilon_M\le C\Bigl\{(\beta+v+t\beta+q_Q)^{1/64}+tv+t^2\beta+t^2v+t^3\beta+q_Q\Bigr\},\label{F:rate}
\end{equation}
where $q_Q$ is superpolynomially small. Substituting $\beta=O(\log M/\sqrt M)$ and
\eqref{F:Sigma}, the right-hand side is a fixed negative power times a polylogarithm, so
$t\epsilon_M\to0$.
\end{lemma}

\begin{proof}
\emph{Step 1 (product tilting).} Define the measure $Q$ with density
$\prod_je^{-tv_j}/\E(e^{-tv_j}\mid\mathcal F_{j-1})$. The deterministic trace bound gives
$\E(dQ/dP)^2\le e^{Ct^2}$, so the bad set remains superpolynomially small under $Q$; the total
variation distance between the true tilting $P_t$ and $Q$ is $O(t^2v+t^3\beta)$ plus a
superpolynomially small term.

\emph{Step 2 (drift and bracket).} Write
$\mathcal B_j=\E(\zeta_j\zeta_j^\top\mid\mathcal F_{j-1})$ and $e_2=(0,1)^\top$. A stepwise
Taylor expansion gives
\begin{equation}
\E_Q(\zeta_j\mid\mathcal F_{j-1})=-t\mathcal B_je_2+\varrho_j,\qquad|\varrho_j|\le Ct^2\beta\operatorname{tr}\mathcal B_j,\label{F:drift}
\end{equation}
and the operator-norm difference between the recentered conditional covariance and $\mathcal B_j$
is at most $Ct\beta\operatorname{tr}\mathcal B_j$. Let $T$ be the vector martingale centered
under $Q$; then on the good set
\[
\|\langle T\rangle-\Sigma\|\le v+Ct\beta,\qquad|(U+\eta_kt,S+t)-T|\le C(tv+t^2\beta).
\]
The drift $-\eta_kt$ of the first coordinate comes from the cross entry of $\mathcal B_j$ in
\eqref{F:drift}; knowing only that the scalar variance is close to $1$ does not give this drift,
so the $EH$ bracket must be controlled simultaneously.

\emph{Step 3 (quantitative two-dimensional rectangle CLT).} For a centered two-dimensional
martingale $T$ with increment bound $b$, deterministic trace upper bound, and matrix bracket
within $w$ of $\Sigma$ with probability at least $1-q'$, set $\delta=b+w+q'$. For each unit
direction $\theta$, standardize $\theta\cdot T$ by $(\theta^\top\Sigma\theta)^{1/2}$ and stop when
the directional bracket first crosses $1+Cw$; the stopping changes only samples of probability at
most $q'$. A one-step Taylor expansion of the compensating characteristic function gives,
direction by direction, for $r=|u|\le\delta^{-1/16}$,
\[
|\varphi_T(u)-e^{-u^\top\Sigma u/2}|\le C\{br^3+b^2r^4+(w+b^2+q')r^2+q'\}\le C\delta(1+r^4),
\]
with constants uniform in the direction. Here only the characteristic function of the original
$T$ is estimated direction by direction; no common stopping vector is constructed.

Add an independent $N(0,h^2I_2)$ to both sides, with $h=\delta^{1/32}$. Fourier inversion gives
the supremum of the smoothed density difference: the integral over the low-frequency ball
$|u|\le\delta^{-1/16}$ is at most $C\delta^{5/8}$, and the high-frequency part is controlled by
$e^{-h^2|u|^2/2}$, which is superpolynomially small relative to $\delta$. The second moments of
both sides are uniformly bounded; restricting the lower rectangle to $[-R,R]^2$ with
$R=\delta^{-1/32}$, the outside probability is at most $C/R^2$ and the integral of the density
error inside is at most $CR^2\delta^{5/8}$, so the smoothed rectangle error is at most
$C\delta^{1/16}$. Finally remove the smoothing: take the coordinate tolerance
$a=h\sqrt{C\log(1/\delta)}$, so the probability that the noise exceeds the tolerance is less than
$\delta^2$; the marginal density of the smoothing Gaussian is uniformly bounded, and the
probability of the boundary strip of the rectangle is $O(a)$. Sandwiching the original rectangle
between the $\pm a$ shifts of the smoothed rectangle gives
\begin{equation}
d_{\rm rect}(T,N(0,\Sigma))\le C\delta^{1/64}.\label{F:rect}
\end{equation}
This conclusion also covers unbounded endpoints, that is, the marginal distributions.

Applying \eqref{F:rect} to the $T$ of Step 2 (with $b=2\beta$, $w=v+Ct\beta$), controlling the
shift $C(tv+t^2\beta)$ by the anticoncentration of the Gaussian marginals, and adding the total
variation error of Step 1 gives \eqref{F:rate}.
\end{proof}

\begin{remark}
The rate in \eqref{F:rate} is weak; it is used only to guarantee $t\epsilon_M\to0$, not to
predict the error at finite $M$.
\end{remark}

\subsubsection{Macro conditional CDF with all-positive signs}

\begin{proposition}[All-positive macro conditional law]\label{F:prop:cdf}
Under \eqref{F:range}, with the cutoff of Lemma~\ref{F:lem:cutoff}, the Plain model satisfies
\begin{equation}
\sup_x\left|\Pr(X_+\le x\mid \mathcal A)-\frac{\Phi_2(x,-t_M;\eta_k)}{\Phi(-t_M)}\right|\to0.\label{F:cdf}
\end{equation}
Equivalently, $X_++\eta_kt_M$ tends to $N(0,1-\eta_k^2)$ conditioned on $\mathcal A$.
\end{proposition}

\begin{proof}
\emph{Inverse tilting.} The stopped $S^*$ satisfies the scalar Laplace approximation
$\log\E e^{-tS^*}=t^2/2+o(1)$; take $t=b_M$. For any $y$, let
$\mathcal K_M(y)=\E_t[e^{t(S^*+t)}\mathbf1_{\{S^*+t\le0,\ Y^*+\eta_kt\le y\}}]$. A
Lebesgue--Stieltjes integration by parts writes $\mathcal K_M$ as an integral of the tilted joint
CDF, so Lemma~\ref{F:lem:tilt} gives $\sup_y|\mathcal K_M(y)-\mathcal K_G(y)|\le2\epsilon_M$,
where $\mathcal K_G$ corresponds to the target Gaussian. Also
$\mathcal K_G(\infty)=e^{t^2/2}\Phi(-t)\sim(\sqrt{2\pi}t)^{-1}$. The exact inverse-tilting
formula is
\[
\Pr(Y^*\le x,\ S^*\le-t)=\E e^{-tS^*}\,e^{-t^2}\,\mathcal K_M(x+\eta_kt),
\]
and controlling the numerator and denominator simultaneously by the difference between
$\mathcal K_M$ and $\mathcal K_G$ gives a conditional CDF error of at most
$O(\epsilon_M/\mathcal K_G(\infty))=O(t\epsilon_M)=o(1)$:
\[
\sup_x\left|\Pr(Y^*\le x\mid S^*\le-t)-\frac{\Phi_2(x,-t;\eta_k)}{\Phi(-t)}\right|\to0.
\]

\emph{Transfer to the macro event.} By input (C3) and the truncation and stopping bad sets of
this appendix, $\Pr(\mathcal A\triangle\{S^*\le-t\})=o(\Phi(-t))$. This probability bound
controls the joint subprobability error for all $\{Y_K\le x\}$ simultaneously; no energy-weighted
bound is needed, so the conditional CDF transfers. $Y_K$ and $Y^*$ differ by only a
superpolynomially small amount on the good set; the density of the selection Gaussian is uniformly
bounded, and given the score the conditional density of the energy is at most
$[2\pi(1-\eta_k^2)]^{-1/2}$, and so remains after mixing. Replace $t$ by $t_M$: the two
Gaussian-tail conditioning events are nested, so the total variation distance of the conditional
distributions is at most the relative tail difference $O(t_M|t-t_M|)=o(1)$. This gives the result
for $Y_K$ in the configuration model.

\emph{Fixed planting.} Fix $l$ special edges and construct the initial two-dimensional martingale
from $k(d-l)$ ordinary endpoints, then append the special-edge filtration of
Appendix~\ref{app:mda}. Replacing $\mu$ by $k(d-l)/M$ changes only an $O_l(M^{-1})$ error in
\eqref{F:bracketmean}. When the maximum ordinary degree is at most $L$, the $kl$ special
incidences change $K$ by at most $C_l(L+1)$ and $H$ by at most $C_l$; after normalization this is
$O_l(L/\sqrt M)$, which, placed in the perturbation of the energy coordinate and the bandwidth of
the macro event, still tends to zero after multiplying by $t$. Write
$G_M(x)=\Phi_2(x,-t_M;\eta_k)/\Phi(-t_M)$; then for each fixed $l$,
\begin{equation*}
\sup_x\left|\Pr\nolimits_l(Y_K\le x,\ \mathcal A)/p_M-G_M(x)\right|\to0.
\end{equation*}
This already contains the relative precision of the probability denominator, which an ordinary
finite-planting CLT cannot replace.

\emph{Plain transfer.} Let $N^{\rm bad}_M$ be the number of bad edges, with
$dp_M^{\rm bad}\to\Lambda_k$. Specifying $l$ bad edges is exactly the admissible special-edge
model, so for each fixed $l$,
\[
\sup_x\left|\E[\mathbf1_{\{Y_K\le x,\ \mathcal A\}}(N^{\rm bad}_M)_l]/p_M-\Lambda_k^lG_M(x)\right|\to0.
\]
Apply finite-order Bonferroni upper and lower bounds to the nonnegative indicator. First fix the
truncation order and let $M\to\infty$; the finitely many error terms are uniform in $x$. The
difference between the two-part sum and $e^{-\Lambda_k}$, multiplied by $G_M(x)\le1$, and then
the truncation order is increased. This gives
$\sup_x|\Pr(Y_K\le x,\mathcal A,N^{\rm bad}_M=0)/p_M-e^{-\Lambda_k}G_M(x)|\to0$; dividing by
$\Pr(\mathcal A,N^{\rm bad}_M=0)/p_M\to e^{-\Lambda_k}$ gives the conditional CDF in the Plain
model.

\emph{Convert back to $X_+$.} In the Plain model exactly
\begin{equation}
X_+=\frac{\sqrt2\mu\sqrt M}{k(1-k/M)s_M}Y_K+\frac{(k-1)\mu}{k(1-k/M)s_M},\label{F:affine}
\end{equation}
that is, $X_+=(1+O(M^{-1}))Y_K+O(M^{-1/2})$. Since $Y_K+\eta_kt_M$ is tight conditioned on
$\mathcal A$, the transformation error is $o_P(1)$, and the density upper bound above gives
\eqref{F:cdf}. The energy with all-negative signs is the same as with all-positive signs.
\end{proof}

\subsubsection{Conditional moments and uniform integrability}

Convergence in distribution does not give moment convergence after the center is shifted. This
section controls the tail directly under the macro conditioning.

\begin{proposition}[Macro conditional moments]\label{F:prop:moments}
Under the assumptions of Proposition~\ref{F:prop:cdf}, every fixed positive-integer moment and
absolute moment of $W^+_M=X_++\eta_kt_M$ conditioned on $\mathcal A$ converges to the
corresponding moment of $N(0,1-\eta_k^2)$. In particular,
\begin{equation}
\E[X_+\mid \mathcal A]=-\eta_kA(-t_M)+o(1),\qquad\Var(X_+\mid \mathcal A)=1-\eta_k^2+o(1).\label{F:moments}
\end{equation}
For any deterministic sign composition independent of the mapping,
$\E[X_s\mid \mathcal A]=\gamma\bigl[-\eta_kA(-t_M)+o(1)\bigr]$, where the $o(1)$ is the same
all-positive remainder, hence uniform in the composition and not requiring $\rho_M$ to converge.
\end{proposition}

\begin{proof}
\emph{Quantitative CDF input.} Work first in the configuration model. The polynomial rectangle
error of Lemma~\ref{F:lem:tilt}, the $t$ factor of the inverse tilting, the relative bandwidth
error $O(t(\delta_M+\epsilon_M))$, and the superpolynomially small bad set divided by $p_M$
together give
\begin{equation}
d_K\bigl(\mathcal L(Y_K\mid \mathcal A),\ \mathcal G_t\bigr)\le e_M=O(M^{-c}(\log M)^b),\label{F:dK}
\end{equation}
where $\mathcal G_t$ is the distribution of the first coordinate of a standard bivariate Gaussian
with correlation $\eta_k$, conditioned on the second coordinate being $\le-t$, and $c>0$,
$b<\infty$ are fixed. The Bonferroni bound has not yet been used here; later, transferring to the
Plain model needs only qualitative CDF convergence.

\emph{Tilted tail bound on the whole sample space.} Continue with the product tilting $Q$ of
\eqref{F:drift}. Under $Q$, $T_1=\sum_j(u_j-\E_Q[u_j\mid\mathcal F_{j-1}])$ is a centered
martingale. By \eqref{F:drift}, on the whole sample space
\[
\Bigl|\sum_j\E_Q[u_j\mid\mathcal F_{j-1}]\Bigr|\le t\sum_j|(\mathcal B_j)_{12}|+Ct^2\beta\sum_j\operatorname{tr}\mathcal B_j\le C_1t,
\]
using only $|(\mathcal B_j)_{12}|\le\operatorname{tr}\mathcal B_j/2$, the deterministic trace
bound, and $t\beta\to0$; it does not require the cross bracket to be close to $\eta_k$ on the bad
set. The increments of $T_1$ are at most $2\beta$ and the sum of predicted variances is at most
$C_2$, so the Bernstein--Freedman inequality gives
$Q(|T_1|>x)\le2\exp\{-x^2/[2(C_2+2\beta x/3)]\}$. Let $W^*=Y^*+\eta_kt$; then
$|W^*-T_1|\le C_3t$, so for $x\ge2C_3t$
\begin{equation}
Q(|W^*|>x)\le2\exp\{-cx^2/(1+\beta x)\}.\label{F:Qtail}
\end{equation}

\emph{Estimate unbounded moments without total variation.} The density of the true tilting with
respect to $Q$ is exactly $\prod_jm_j(-t)/\E e^{-tS^*}$, where
$m_j(-t)=\E(e^{-tv_j}\mid\mathcal F_{j-1})$. The deterministic trace bound and $t\beta\to0$ give
$\prod_jm_j\le e^{Ct^2}$; combined with $\E e^{-tS^*}\sim e^{t^2/2}$, this gives the
deterministic bound $dP_t/dQ\le e^{C_4t^2}\le M^{D_0}$ on the whole sample space, with $D_0$
fixed. This crude bound is used together with the martingale tail bound, whose order can be raised
arbitrarily; multiplying the small total variation distance by a growing maximum of the energy
does not yield a vanishing error.

\emph{Truncated tail.} Fix a positive integer $q$. The actual energy satisfies
$|Y_K|\le CM^{3/2}$, and the stopped martingale has at most $O(M)$ steps of size
$O(\log M/\sqrt M)$ each, so $|W^*|\le CM^2$ is a deterministic crude bound. Take
$R_M=A_q\sqrt{\log M}$ with $A_q$ large enough that $R_M\ge2C_3t$, and \eqref{F:Qtail} gives
$Q(|W^*|>R_M)\le2M^{-H_q}$ with $H_q>D_0+2q+2$. Outside $\mathcal B_M$, the event
$\mathcal A$ implies $S^*\le-t+C\delta_M$; for any nonnegative $f(W^*)$, the inverse tilting
gives
\[
\E[f(W^*);\mathcal A\setminus\mathcal B_M]\le\E e^{-tS^*}e^{-t^2}e^{Ct\delta_M}\E_tf(W^*)\le Cp_M\,t\,M^{D_0}\E_Qf(W^*),
\]
where the last step uses $\E e^{-tS^*}e^{-t^2}/p_M=O(t)$. Taking
$f=|W^*|^q\mathbf1_{\{|W^*|>R_M\}}$ gives
$\E[|W^*|^q\mathbf1_{\{|W^*|>R_M\}};\mathcal A\setminus\mathcal B_M]/p_M\le
CtM^{D_0+2q-H_q}=o(1)$. The contribution of the bad set is at most
$CM^{2q}\Pr(\mathcal B_M)/p_M=o(1)$. Hence
\begin{equation}
\E[|W_K|^q\mathbf1_{\{|W_K|>R_M\}}\mid \mathcal A]\to0,\qquad W_K=Y_K+\eta_kt.\label{F:tailcut}
\end{equation}
Since $R_M\to\infty$, \eqref{F:tailcut} is not yet uniform integrability by itself.

\emph{Between a fixed threshold and $R_M$.} Let $Z_0,G$ be independent standard normals; the
shifted-center variable of the selection Gaussian can be written as
$W^G_t=\eta_k(Z_0+t)+\sqrt{1-\eta_k^2}\,G$, conditioned on $Z_0\le-t$. For $t\ge1$ and $x\ge0$,
\[
\Pr(-(Z_0+t)>x\mid Z_0\le-t)=\frac{\Phi(-(t+x))}{\Phi(-t)}\le e^{-tx-x^2/2},
\]
which follows by shifting the integration variable of the numerator by $x$. Hence all moments of
the overshoot tend to zero, and the moments of $W^G_t$ tend to the corresponding moments of
$N(0,1-\eta_k^2)$. For the $q$-th moment use $\mathrm{clip}(w,-R_M,R_M)^q$, and for the absolute
moment use $\min(|w|,R_M)^q$; their total variation is at most $2R_M^q$, and Stieltjes integration
by parts together with \eqref{F:dK} gives that the truncated expectations differ by at most
$2R_M^qe_M=o(1)$. Removing the truncation via \eqref{F:tailcut} and the uniform tail control of
the Gaussian gives the convergence of all moments in the configuration model. Taking the
conclusion for order $q+1$ gives the uniform integrability
\begin{equation}
\sup_{M\text{ sufficiently large}}\E[|W_K|^q\mathbf1_{\{|W_K|>R\}}\mid \mathcal A]\le C_q/R.\label{F:UI}
\end{equation}

\emph{Plain transfer.} The scalar relative Bonferroni already gives
$\Pr_{\rm conf}(\mathcal A,N^{\rm bad}_M=0)/\Pr_{\rm conf}(\mathcal A)\to e^{-\Lambda_k}>0$.
For any nonnegative $f$,
\[
\E_{\rm Plain}[f(Y_K)\mid \mathcal A]=\frac{\E_{\rm conf}[f(Y_K)\mathbf1_{\{N^{\rm bad}_M=0\}}\mid \mathcal A]}{\Pr_{\rm conf}(N^{\rm bad}_M=0\mid \mathcal A)}\le C\,\E_{\rm conf}[f(Y_K)\mid \mathcal A],
\]
so \eqref{F:UI} transfers directly to the Plain model. This inequality transfers only the tail
control; the limiting distribution comes from the Plain CDF transfer of
Proposition~\ref{F:prop:cdf}. Together they give all moment limits in the Plain model. Finally
convert back to $X_+$ via \eqref{F:affine}:
$X_++\eta_kt_M=(1+O(M^{-1}))(Y_K+\eta_kt)+O(M^{-1/2})+\eta_k(t_M-t)+O(t/M)$, where the last
several terms tend to zero.

\emph{Mean and variance.} The first- and second-moment conclusions give
$\E[X_+\mid \mathcal A]=-\eta_kt_M+o(1)$ and
$\Var(X_+\mid \mathcal A)=1-\eta_k^2+o(1)$. The Mills ratio satisfies $A(-t)-t=O(1/t)$, so the
mean can also be written in the form \eqref{F:moments}. This is only a rewriting; the remainder is
still $o(1)$, and it does not identify the coefficient of order $1/t_M$.

\emph{Mean under an arbitrary composition.} The event $\mathcal A$ depends only on the unsigned
mapping and is invariant under permutations of the element labels, so the exact projection
\eqref{eq:projection} gives $\E[X_s\mid \mathcal A]=\gamma\E[X_+\mid \mathcal A]$. This step does
not require a sign CLT under a rare-event conditioning.
\end{proof}

\subsubsection{Conditional CDF and variance under an arbitrary sign composition, including
$|\rho|=1$}

\begin{proposition}[Macro conditional law under an arbitrary composition]\label{F:prop:sign}
Under the assumptions of Proposition~\ref{F:prop:cdf}, let the signs be an arbitrary deterministic
composition, with $\rho_M$ not required to converge. Let
$G_{M,s}(x)=\Phi_2(x,-t_M;\eta_k\rho_M^2)/\Phi(-t_M)$; then
\begin{equation}
\sup_x|\Pr(X_s\le x\mid \mathcal A)-G_{M,s}(x)|\to0,\qquad\Var(X_s\mid \mathcal A)=1-\eta_k^2\rho_M^4+o(1).\label{F:sign}
\end{equation}
\end{proposition}

The proof of this proposition is given after Lemma~\ref{F:lem:graph}--\ref{F:lem:endpoint}. The
idea is to avoid redoing a signed tilting for each sign composition. First generate an unsigned
Plain hypergraph, then assign the given numbers of positive and negative signs uniformly; this has
the same joint law of $(X_s,\mathcal A)$ as the original model, and $\mathcal A$ selects only the
graph. Hence two facts suffice: the required graph properties still hold with probability tending
to one under $\mathcal A$; and on these deterministic good graphs the sign CLT holds uniformly.
The graph properties used in Theorem~\ref{thm:sign} hold only to $o_P(1)$ accuracy and cannot be
divided directly by the rare $\Pr(\mathcal A)$. Below, the exceptional probabilities are
strengthened to be superpolynomially small.

Use the overlap graph notation of Appendix~\ref{sec:signproof}: $m_{ij}=|S_i\cap S_j|$,
$\Omega=\sum_{i<j}m_{ij}$, $\Omega_2=\sum_{i<j}m_{ij}^2$, $\ell_i=\sum_{j\ne i}m_{ij}$,
$\Omega_\ell=\sum_i\ell_i^2$, and $\bar\ell=2\Omega/d$.

\begin{lemma}[Superpolynomial concentration of the overlap graph statistics]\label{F:lem:graph}
There exists a deterministic set $\mathcal H_M$ of graphs such that
$\Pr(\mathcal G\notin\mathcal H_M)$ is superpolynomially small and all graphs in $\mathcal H_M$
uniformly satisfy
\begin{equation}
\frac\Omega M=\frac{\mu_c^2}2+o(1),\quad\frac{\Omega_2}M=\frac{\mu_c^2}2+o(1),\quad\frac1M\sum_i(\ell_i-\bar\ell)^2=\mu_c^2+o(1),\quad\max_i\ell_i\le kL.\label{F:goodgraph}
\end{equation}
All the $o(1)$ terms can be taken as $O((\log M)^4/\sqrt M+|r_M|/\sqrt M)$. In particular
$\Pr(\mathcal G\notin\mathcal H_M\mid \mathcal A)\to0$, and this exceptional probability
multiplied by any fixed power of $M$ and divided by $\Pr(\mathcal A)$ still tends to zero.
\end{lemma}

\begin{proof}
Let $L=\lceil\log M\rceil$ and $\mathcal D_L=\{\max_vD_v\le L\}$. Since
$D_v\sim\mathrm{Bin}(d,k/M)$, $\Pr(\mathcal D_L^c)\le M(eC/L)^L$, which is superpolynomially
small.

\emph{Lipschitz bound.} For two graphs in $\mathcal D_L$, replace the edges of differing labels
one by one by the corresponding edges of the other graph. The incidence sets of all intermediate
graphs are contained in the union of the two endpoint graphs, so the vertex degrees do not exceed
$2L$. Replacing one edge in such an intermediate graph affects only the $O(L)$ edge labels that
overlap the old or the new edge, and each $m_{ij}\le k$, so both $\Omega$ and $\Omega_2$ change
by $O(L)$. For the replaced edge $\ell_i=O(L)$, whose square changes by $O(L^2)$; for the
remaining $O(L)$ affected labels, $\ell_j$ changes by $O(1)$, so each square changes by $O(L)$.
Hence the total change of $\Omega_\ell$ is $O(L^2)$. Thus all three are $CL^2$-Lipschitz in the
Hamming distance of edge labels on $\mathcal D_L$. Here the global Lipschitz bound is proved along
a path in the union of two good graphs; it is not assumed that a single-step path stays within
$\mathcal D_L$.

\emph{Extension and concentration.} For each statistic $f$, take the McShane extension
$\widetilde f(\mathcal G)=\inf_{\mathcal G'\in\mathcal D_L}\{f(\mathcal G')+CL^2d_H(\mathcal G,\mathcal G')\}$
and truncate it to the deterministic nonnegative polynomial range of $f$. The extension equals
$f$ on the good set and remains $CL^2$-Lipschitz; the bad set has superpolynomially small
probability, and both $f$ and the truncated extension are at most $O(M^3)$, so
$|\E f-\E\widetilde f|$ is superpolynomially small. Applying McDiarmid's inequality to independent
edges,
\[
\Pr\bigl(|f-\E f|>C\sqrt M(\log M)^4\bigr)\le\Pr(\mathcal D_L^c)+2e^{-c(\log M)^4},\qquad f=\Omega,\Omega_2,\Omega_\ell.
\]

\emph{Deterministic means.} For the overlap $H$ of two independent uniform $k$-subsets, write
$h_1=\E H=k^2/M$ and $h_2=\E H^2=k^2/M+k^2(k-1)^2/[M(M-1)]$. Exactly,
\[
\E\Omega=\tbinom d2h_1,\qquad\E\Omega_2=\tbinom d2h_2,\qquad\E\Omega_\ell=d\{(d-1)h_2+(d-1)(d-2)h_1^2\}.
\]
The last identity holds because, given $S_i$, the $m_{ij}$ are independent and their conditional
first and second moments do not depend on $S_i$. Taking the good set as the intersection of
$\mathcal D_L$ with the three deviation threshold events gives \eqref{F:goodgraph}; the third
expression is computed from $\Omega_\ell-4\Omega^2/d$, in which the deterministic constants of the
two terms cancel. By input (C4), $\Pr(\mathcal A)\sim p_M$ is no smaller than a fixed negative
power divided by a polylogarithm, so the conditional exceptional probability tends to zero.
\end{proof}

\begin{lemma}[Uniform sign CLT on the good graphs]\label{F:lem:unifsign}
Make the exact decomposition $X_s=\gamma X_++V$ with $\E_s[V\mid\mathcal G]=0$. If
$\rho_M\to\rho$ and $|\rho|<1$, then for each fixed $u$,
\begin{equation}
\sup_{\mathcal G\in\mathcal H_M}\Bigl|\E_s[e^{iuV}\mid\mathcal G]-e^{-(1-\rho^4)u^2/2}\Bigr|\to0.\label{F:unifsign}
\end{equation}
\end{lemma}

\begin{proof}
The sign arguments of Appendix~\ref{sec:signproof} and Appendix~\ref{app:technical} are arguments
for deterministic graph sequences, using only \eqref{F:goodgraph}, $\Omega=O(M)$,
$\Omega_\ell=O(M)$, and a fixed distance of the sign proportions from the endpoints. Concretely,
the within-group variance bound for a whole-batch flip uses only $\Omega_\ell=O(M)$; the variance
of a flipped subset and the expectation of its internal edges use only $\Omega_\ell$ and $\Omega$.
The main term has $O(M)$ local terms, the maximum degree of the dependency graph and the magnitude
of the terms are both $O(\log M)$, and the normalized connected cumulant of fixed order $r\ge3$ is
$O_r(M^{1-r/2}\,\mathrm{polylog}\,M)$; the variance of the main term is determined by
\eqref{F:goodgraph}. Hence for any sequence $\mathcal G\in\mathcal H_M$,
$\mathcal L_s(V\mid\mathcal G)\Rightarrow N(0,1-\rho^4)$. If \eqref{F:unifsign} failed, one could
extract along a subsequence a deterministic graph violating it; this sequence of graphs still
satisfies all the deterministic bounds, contradicting the conclusion above.
\end{proof}

\begin{lemma}[Graph-by-graph conditional variance and the endpoints]\label{F:lem:endpoint}
Uniformly over all feasible compositions, on $\mathcal H_M$
\begin{equation}
\Var_s(V\mid\mathcal G)=1-\rho_M^4+o(1).\label{F:pergraphvar}
\end{equation}
\end{lemma}

\begin{proof}
The unconditional $\E V^2$ cannot be computed first and then divided by the rare $\Pr(\mathcal
A)$, so the graph-by-graph exact variance is used instead. Let $a_4$ be the mean of the sign
products of four distinct labels. The finite identity
\[
\Var_s\Bigl(\sum_{i<j}m_{ij}s_is_j\Bigm|\mathcal G\Bigr)=\Omega_2(1-2\gamma+a_4)+\Omega_\ell(\gamma-a_4)+\Omega^2(a_4-\gamma^2).
\]
Setting $v=1-\rho_M^2$ and expanding this finite rational expression uniformly over all feasible
compositions gives
\[
v^2\Omega_2+\rho_M^2v(\Omega_\ell-4\Omega^2/d)+O(\Omega_2/d+\Omega_\ell/d+\Omega^2/d^2).
\]
On $\mathcal H_M$ we have $\Omega,\Omega_2,\Omega_\ell=O(M)$, so the remainder is $O(1)$;
multiplying by the normalization factor $4/[k^2(1-k/M)^2s_M^2]\sim2/(\mu_c^2M)$ gives
\eqref{F:pergraphvar}. This also gives the crude upper bound $C(v+M^{-1})$. When
$|\rho_M|\to1$, \eqref{F:pergraphvar} and the graph-by-graph Chebyshev inequality give $V\to0$,
uniformly over the good graphs; this covers arbitrary rates of approach to the endpoints.
\end{proof}

\begin{proof}[Proof of Proposition~\ref{F:prop:sign}]
Let $W^+_M=X_++\eta_kt_M$. Proposition~\ref{F:prop:cdf} gives
$W^+_M\mid \mathcal A\Rightarrow N(0,1-\eta_k^2)$. For any subsequence with $\rho_M\to\rho$ and
$|\rho|<1$, condition first on the graph: on the good graphs the inner error tends to zero
uniformly by \eqref{F:unifsign}, on the bad graphs the absolute value of the characteristic
function is at most one, and by Lemma~\ref{F:lem:graph} the probability of a bad graph under
$\mathcal A$ tends to zero. Hence
\[
\E[e^{iaW^+_M+ibV}\mid \mathcal A]-\E[e^{iaW^+_M}\mid \mathcal A]\,e^{-(1-\rho^4)b^2/2}\to0,
\]
that is, $W^+_M$ and $V$ are asymptotically independent under $\mathcal A$. For $|\rho|=1$ use
instead the conditional probability convergence of Lemma~\ref{F:lem:endpoint}.

Keep the actual shifted center:
\[
X_s+\eta_k\rho_M^2t_M=\gamma W^+_M+V+\eta_k(\rho_M^2-\gamma)t_M,
\]
where the last term is $O(t_M/d)\to0$. Here $\rho_M^2t_M$ is not replaced by $\rho^2t_M$ in the
absence of a rate. Therefore
$\mathcal L(X_s+\eta_k\rho_M^2t_M\mid \mathcal A)\Rightarrow N(0,1-\eta_k^2\rho^4)$. The
selection benchmark $G_{M,s}$ has the same limit under the same shifted center: it can be
represented as $c(Z+t_M)+\sqrt{1-c^2}\,G$, conditioned on $Z\le-t_M$, where $c=\eta_k\rho_M^2$,
and the Gaussian overshoot tends to zero. The two distributions tend to the same continuous
distribution, so their CDF distance tends to zero. When $\rho_M$ does not converge, take a
subsequence on which the CDF distance has a positive lower bound, then by compactness of $[-1,1]$
extract a sub-subsequence on which $\rho_M$ converges, contradicting the conclusion above. This
proves the first display of \eqref{F:sign}.

The variance follows from the law of total variance:
\[
\Var(X_s\mid \mathcal A)=\gamma^2\Var(X_+\mid \mathcal A)+\E[\Var_s(V\mid\mathcal G)\mid \mathcal A].
\]
On the good graphs use \eqref{F:pergraphvar}; on the bad graphs use the deterministic polynomial
upper bound of the energy and the superpolynomial accuracy of Lemma~\ref{F:lem:graph}. Substituting
$\Var(X_+\mid \mathcal A)\to1-\eta_k^2$ from Proposition~\ref{F:prop:moments} and noting
$\gamma^2-\rho_M^4=O(1/d)$ gives the second display of \eqref{F:sign}. The conditional mean still
uses the exact projection and cannot be replaced by CDF convergence.
\end{proof}

\begin{remark}[Balancing rate]
Fixing $\rho_M\to0$ in the critical window is enough to recover the standard normal. In the
moderate deviation regime, the macro-layer center $-\eta_k\rho_M^2t_M$ grows with $t_M$: if
$\rho_M^2t_M\to0$ (which an exactly balanced composition satisfies), then
$X_s\mid \mathcal A\Rightarrow N(0,1)$; if $\rho_M^2t_M\to b\in(0,\infty)$, then the limit is
$N(-\eta_kb,1)$. This concerns the macro layer. When the weight of the repeated-edge layer in $F$
tends to one, the failure-conditioned law still recovers normality.
\end{remark}

\subsubsection{The two-mechanism mixture}\label{app:two-mechanism}

The failure-conditioned distribution is asymptotically a two-layer weighted mixture, for two reasons: a
failure that is neither a repeated edge nor a macro core has probability
$O(M^{-b_k})=o(\lambda_M)$; and the probability that a macro failure and a repeated-edge failure
occur simultaneously is also $o(\lambda_M)$. Concretely, Lemma~\ref{F:lem:cutoff} gives the exact
set decomposition $F=\mathcal A\cup D\cup R$ together with the sizes of the parts. The planted
transfer of Appendix~\ref{app:lowload} gives
$\sup_x|\Pr(X_s\le x\mid D)-\Phi(x)|\to0$, which requires only that $d/M$ tends to a finite
positive number and holds equally as $\alpha_M\to\alpha_c$. Hence
\[
\Pr(X_s\le x,F)=\Pr(\mathcal A)\Pr(X_s\le x\mid \mathcal A)+\lambda_M\Phi(x)+o(\lambda_M),
\]
with the error uniform in $x$. Substituting input (C4) and Proposition~\ref{F:prop:sign} gives
the mixture display of Proposition~\ref{prop:mixture}.

\subsection{Load-axis coverage}\label{app:axis}
This subsection proves Proposition~\ref{prop:axis}. The proof of coverage is assembled by load
interval, with each interval relying on a different result: the critical window uses
Theorem~\ref{thm:upper}; the moderate deviation regime uses Proposition~\ref{prop:mixture}; the
fixed positive subcritical regime uses Proposition~\ref{prop:lowload}; when $d/M\to0$ and
$d/\sqrt M\to\infty$, it uses the small-core counting with the $d$-dependence retained, together
with the CLT after planting a repeated edge; the moving subcritical regime
($|r_M|/\sqrt{\log M}\to\infty$) uses a sandwich via the monotonicity of the failure event under
edge-addition coupling; and the overload side, as well as $d/M\to\infty$, uses the unconditional
CLT along the load axis.

To distinguish it from the mixture rule of Section~\ref{sec:rules}, this appendix denotes the upper
bound with threshold $T_n$ in Proposition~\ref{prop:axis} by $U_{\rm asym}$, and the
corresponding rule as the asymptotic mixture rule. The two differ in three respects: the asymptotic
mixture rule uses the unshifted coordinate $z_n$, uses $\lambda_n$ directly in the weight, and adds
no chi-square correction.

The mechanism of the counterexample can be seen from Proposition~\ref{prop:lowload}. On the sparse
side the failure almost always comes from a repeated edge; each repeated edge shifts
$\widehat d$ by $\pm2$, while $s_M$ is of order $d/\sqrt M$. When $d/\sqrt M\to\infty$ this shift
vanishes on the standardized scale; when $d=o(\sqrt M)$ it instead tends to infinity. Under
balanced signs the same-sign and opposite-sign repeated edges are each about half. The
opposite-sign half pushes the estimate too low, so the true value $d$ and all larger candidates
are rejected, and the probability that the upper bound misses the true value therefore tends to
$1/2$. For fixed $d=4$ with two positive and two negative signs, four of the six pairs are of
opposite sign, so this probability tends to $2/3$. Here the failure probability tends to zero, so
the positive lower bound required by Corollary~\ref{cor:balanced} fails, and balanced signs cannot
guarantee selection invariance.

In the sparse-side counterexample, $c_n\equiv0$, so the asymptotic mixture rule degenerates to the
unconditional normal threshold. In this case the failure-conditioned graph tends to a repeated
edge planted on a random pair, and the energy equals exactly $(d+2\sigma)/(1-k/M)$, where
$\sigma$ is the sign product of the repeated pair and satisfies
$\Pr(\sigma=-1)=d/[2(d-1)]$. Thus $U_{\rm asym}=d+2\sigma$. The acceptance set is nonempty, so
the fallback for an empty set also does not rescue the coverage.

The coverage of the asymptotic mixture rule depends only on whether the true candidate is rejected.
Let $\epsilon_M=\sup_x|\Pr(X\le x\mid F)-H_d(x)|$. Since $H_d$ is continuous, taking left limits
gives, for every real $q$, $|\Pr(X<q\mid F)-H_d(q)|\le\epsilon_M$; take $q=H_d^{-1}(u)$. So it
suffices to prove $\epsilon_M\to0$. By a subsequence argument it suffices to treat the following
cases separately.

Inside the critical window, $\lambda_M\to0$ and $H_d$ returns to the selection law. The moderate
deviation regime is handled by Proposition~\ref{prop:mixture}. The fixed subcritical regime has
$w_d\to0$ and is handled by Proposition~\ref{prop:lowload}.

When $r_M\to-\infty$, $|r_M|/\sqrt{\log M}\to\infty$, and $d/M\to\alpha_c$, take
$a=g_k^2C^2/2>k-2$ and $d'_M=\lfloor\alpha_cM-C\sqrt{M\log M}\rfloor$, so that eventually
$d_M\le d'_M$. Under edge-addition coupling the edge set of the core is monotone under edge
addition, so by Proposition~\ref{prop:mda},
$\Pr_{d_M}(\mathcal A)\le\Pr_{d'_M}(\mathcal A)=M^{-a+o(1)}=o(\lambda_M)$. Thus
$\Pr(F\setminus D)=o(\lambda_M)$, the conditional law is determined by the repeated-edge layer,
and $H_d\to\Phi$.

On the overload side and when $d/M\to\infty$, $\Pr(F)\to1$, and the conclusion follows from the
unconditional CLT along the load axis. This CLT holds for every sequence with
$d/\sqrt M\to\infty$. The reason is that the edge-revealing martingale of Appendix~\ref{app:lowload}
satisfies $\Var\mathcal V_d=O(d^4/M^3)$, while $s_M^4\approx4d^4/M^2$, so the variance of the
normalized bracket is $O(1/M)$; the sum of fourth moments satisfies
\begin{equation*}
\frac{\sum_j\E\Delta_j^4}{s_M^4}\le C_k\Bigl(\frac M{d^2}+\frac1d\Bigr)\to0,
\end{equation*}
so the martingale CLT gives $X\Rightarrow N(0,1)$. This conclusion does not require $d/M$ to be
bounded, and it is uniform in the signs. When a repeated edge is planted on a random pair, the
deterministic term $2\sigma$ and the cross term, divided by $s_M$, have squared errors
$O(M/d^2)$ and $O(1/d)$ respectively, which also tend to zero.

\paragraph{The $d/M\to0$ conditional transfer.} Here the dependence on $d$ in the small-core
counting of Appendix~\ref{app:lowload} is retained. In this regime every possible core size
satisfies $e\le d=o(M)$, and
\[
B_e\le(C_kd/M)^e(e/M)^{\lceil(k-2)e/2\rceil}.
\]
Retaining the rounding for $e=3$ and summing the upper bound over $e\ge4$ gives
\[
\Pr(F\setminus D)\le C_k\left(\frac{d^3}{M^{b_k+3}}+\frac{d^4}{M^{2k}}\right),
\qquad b_k=\lceil3(k-2)/2\rceil\ge k-1.
\]
By $\lambda_M\asymp_k d^2/M^k$ and the second-factorial-moment bound of repeated pairs, the
random-pair repeated-planting distribution $Q$ satisfies
\[
\bigl\|P(\cdot\mid F)-Q\bigr\|_{\rm TV}
\le C_k\left(\lambda_M+dM^{k-b_k-3}+\frac{d^2}{M^k}\right)\longrightarrow0.
\]
Hence when $d/\sqrt M\to\infty$, the planted CLT above gives $X\mid F\Rightarrow N(0,1)$. On the
other hand, $z_M(d)=-(g_k\alpha_c+o(1))\sqrt M$, while $\lambda_M\ge1/\binom Mk$, so $w_d\to0$
and $H_d\to\Phi$ uniformly. This completes the branch with load tending to zero. For a
subsequence, first take $d/M$ converging in $[0,\infty]$; when the limit is $\alpha_c$, then take
the limit of $r_M$ and of $|r_M|/\sqrt{\log M}$. This covers all load sequences.

Now construct the sparse-side counterexample. When $b=0$, each candidate has $c_n=0$, and the
asymptotic mixture rule degenerates exactly to the unconditional normal threshold
$T_n=n+s(n)\Phi^{-1}(u)$. When $d=o(\sqrt M)$, the failure-conditioned graph tends to a repeated
edge planted on a random pair of labels; the probability that the remaining base edges are pairwise
disjoint is $1-O(d^2/M)$, in which case $\widehat d=(d+2\sigma)/(1-k/M)$. For $M$ sufficiently
large, $T_n$ is strictly increasing in $n\ge2$, and
$T_{d+2\sigma}<\widehat d<T_{d+2\sigma+2}$, so $U_{\rm asym}=d+2\sigma$, hence
$\Pr(U_{\rm asym}<d\mid F)\to\Pr(\sigma=-1)=d/[2(d-1)]$.
\subsection{The second sign change of the mean}\label{app:wm}
Finally consider the conditional mean. In the fixed subcritical regime $c_F$ is positive
(Proposition~\ref{prop:lowload}); in the critical window $c_F$ is negative for $k\ge5$
(Section~\ref{sec:flip}). So for $k\ge5$, as the load decreases, $c_F$ must change sign once
more. The weighted formula of this section uses the all-same-sign measurement: set all $s_i=+1$
and write $Y=\widehat d-d=s_MX_{\rm one}$.

\begin{proposition}[Weighted mean and the sign-change layer of the mean for $k\ge5$]\label{prop:wm}
In the moderate deviation range of Section~\ref{sec:loadaxis}, with the cutoff of
Lemma~\ref{F:lem:cutoff},
\begin{equation*}
\E[Y;F]=2\lambda_M-s_M\eta_k\phi(t_M)+o\bigl(\lambda_M+s_M\phi(t_M)\bigr).
\end{equation*}
Fix $k\ge5$ and let
\begin{equation*}
t_*^2=(2k-3)\log M-2\log\frac{\alpha_ck!\sqrt\pi}{\eta_k}.
\end{equation*}
For any fixed $v\in\mathbb R$, take integers $d_M$ with $t_M^2-t_*^2\to v$; then
$c_F\binom d2\to1-e^{-v/2}$. Hence between the two layers $t^2=t_*^2\pm\delta$, $c_F$ changes
sign at least once.
\end{proposition}
The first display of the proposition is proved at the end of this subsection, and the second part
follows from the first; the meaning of the two conclusions is explained first. For any
deterministic sign composition, $\E[\widehat d-d;F]$ is exactly the all-same-sign expectation
multiplied by the sign factor $\gamma=(d\rho_M^2-1)/(d-1)$, so both the main term and the
remainder are projected by this factor. The coefficient $c_F$ is extracted from the all-same-sign
measurement, so the above conclusion about the sign-change layer is unaffected.

Why is the sign-change layer deeper than the probability-equal-weight layer? The equal-weight point
compares the failure probabilities of the two mechanisms, namely $\Phi(-t)\approx\lambda_M$,
which by the Mills ratio gives $t^2\approx(2k-4)\log M-\log\log M$. The sign change of the mean
compares instead the contributions of the two mechanisms to the mean, namely
$s_M\phi(t)\approx\lambda_M$. The factor $A(-t)=\phi(t)/\Phi(-t)$ in the macro conditional mean
cancels the Mills factor exactly, so the sign-change layer lies at
$t^2\approx(2k-3)\log M$, deeper than the equal-weight layer, and without a $\log\log M$ term. On
the sign-change layer the weight of the macro failure is only $O(1/\sqrt{M\log M})$; but the
energy shift that a macro failure produces on the original scale is $O(\sqrt{M\log M})$, and the
product is still enough to cancel the $+2$ of the repeated-edge layer. Thus the mechanism that
dominates the failure probability is not the mechanism that dominates the conditional mean.
Figure~\ref{fig:loadaxis}(a) plots the position of the two layers for $k=5$ at finite size.

The two kinds of sign change must also be distinguished. The sign change of Section~\ref{sec:flip}
fixes the load and varies $k$; the sign change here fixes $k\ge5$ and moves along the load. They
are different phenomena.

The $o(1)$ error of the mixture CDF cannot be multiplied by a growing macro center, so the
conditional mean must be controlled separately on the unnormalized first moment.

Now prove the first display of Proposition~\ref{prop:wm}. Let $N$ be the number of repeated edge
pairs and $D=\{N\ge1\}$.

\begin{proof}
The exact density of the random-pair planting gives $\E[YN]=2\lambda_M$. Let
$R=F\setminus(\mathcal A\cup D)$; there is the graph-by-graph identity
\begin{equation}
Y\mathbf1_F-Y\mathbf1_{\mathcal A}-YN=Y\mathbf1_R-YN\mathbf1_{\mathcal A}-Y(N-1)\mathbf1_{D\setminus \mathcal A}.\label{F:wm-id}
\end{equation}
The absolute first moments of the three terms on the right are estimated one by one below.

\emph{Third-kind failures.} The moment estimate over covering events in the mean expansion of
Appendix~\ref{app:lowload} uses only $d\le AM$, so it holds uniformly on the present moving load:
$\E[|Y|;R]\le C\sqrt M\,M^{-b_k}=o(\lambda_M)$, since $b_k-(k-2)\ge1$. Here the core on $R$ does
not exceed $\varepsilon_0M$ edges, and no fixed subcritical gap is needed.

\emph{Multiplicity of repeated pairs.} The joint conditioning on any two prescribed repeated pairs
involves at most four special labels, and the internal, cross, and external decomposition of the
energy still gives $\E|Y|\le C\sqrt M$; also $\E[N(N-1)]\le\lambda_M^2$, so
$\E[|Y|(N-1);D]\le C\sqrt M\lambda_M^2=o(\lambda_M)$.

\emph{Intersection of the macro layer and the repeated layer.} Let $Q$ be the repeated-planting
distribution of a random pair, whose density is exactly $N/\lambda_M$. Then
\[
\E[|Y|N;\mathcal A]=\lambda_M\E_Q[|Y|;\mathcal A]\le\lambda_M\sqrt{\E_QY^2}\sqrt{Q(\mathcal A)}\le C\lambda_M\sqrt{Mp_M}.
\]
$\E_QY^2=O(M)$ comes from the energy decomposition after repeated planting and the external
unconditional variance of Appendix~\ref{app:lowload}. $Q(\mathcal A)\sim p_M$ comes from the
fixed-planting version of input (C4): take the two identical valid $k$-subsets as finite special
edges, then apply a relative Bonferroni transfer to the remaining ordinary edges. Let
$b=s_Mt_M/\lambda_M$; by $\sqrt{p}/(1+bp)\le1/(2\sqrt b)$,
\[
\frac{\lambda_M\sqrt{Mp_M}}{\lambda_M+s_Mp_Mt_M}\le\frac12\sqrt{\frac{M\lambda_M}{s_Mt_M}}=O\Bigl(M^{(5/2-k)/2}/\sqrt{t_M}\Bigr)=o(1).
\]
When $t_M\to\infty$, $p_MA(-t_M)=\phi(t_M)\approx p_Mt_M$, so the intersection term is
$o(\lambda_M+s_M\phi(t_M))$, and no refined weighted local response is needed.

\emph{Macro main term.} By input (C4) and Proposition~\ref{F:prop:moments},
$\E[Y;\mathcal A]=s_M\Pr(\mathcal A)[-\eta_kA(-t_M)+o(1)]=-s_M\eta_k\phi(t_M)+o(s_M\phi(t_M))$.
The probability needs only relative $o(1)$ accuracy, since the target remainder is
$o(\lambda_M+s_M\phi(t_M))$.

Substituting the three estimates and the main term into \eqref{F:wm-id} gives the first display of
Proposition~\ref{prop:wm}.
\end{proof}

Combining $\Pr(F)=(p_M+\lambda_M)(1+o(1))$ with $c_F\binom d2=\E[Y\mid F]/2$ gives
\begin{equation*}
c_F\binom d2=\frac{\lambda_M-s_M\eta_k\phi(t_M)/2}{p_M+\lambda_M}+o\Bigl(\frac{\lambda_M+s_M\phi(t_M)}{p_M+\lambda_M}\Bigr).
\end{equation*}
For $k\ge5$, $t_M^2-t_*^2\to v$ gives
$s_M\eta_k\phi(t_M)/(2\lambda_M)\to e^{-v/2}$ and $p_M/\lambda_M\to0$, hence
$c_F\binom d2\to1-e^{-v/2}$.

\section{Supplementary experiments}\label{app:exp}
\subsection{Size convergence and mean reversal}\label{sec:convergence}
Figure~\ref{fig:convergence} shows the convergence diagnostics. At $k=5$ and
$z_{\rm sh}\approx-3.5$, the conditional mean over all failure samples goes from about $+0.047$ at $M=256$ to
about $-0.473$ at $M=65536$. At the largest size, the prediction in shifted coordinates is about $-0.475$, and
the standard error of the simulated mean is about $0.015$. In other words, the negative offset at large size
matches the prediction; at small size the offset has the opposite sign, and this reversal disappears as $M$ grows.

The second panel shows the residuals of the macroscopic core conditional layer, referenced against the selection
law in empirical coordinates. Positive residuals appear at both values of $k$ and shrink as the size grows; this
is consistent with the roughly $M^{-1/2}$ contraction direction of the moderate deviation regime in
Appendix~\ref{app:exp-md}.
\begin{figure}[!htbp]
\centering\includegraphics[width=\linewidth]{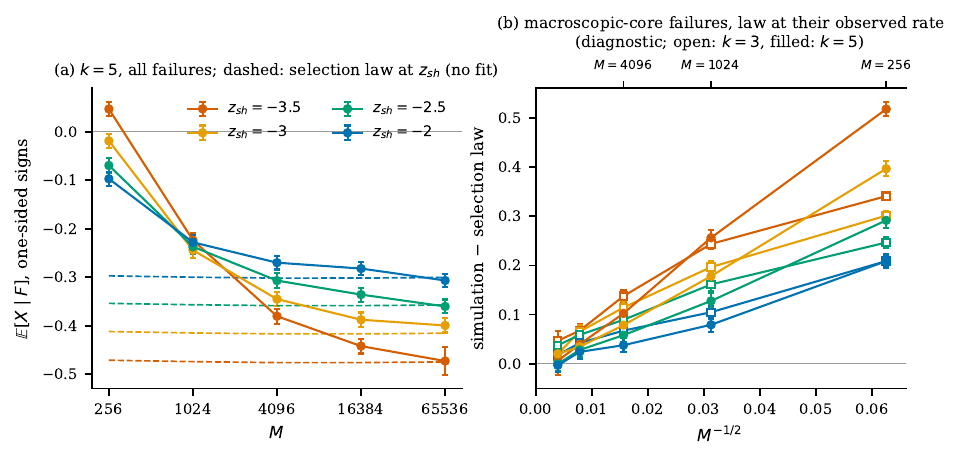}
\caption{Size convergence. (a) The conditional mean over all failure samples at $k=5$ and the analytic prediction of
the selection law; (b) residuals of the macroscopic core layer.}
\label{fig:convergence}
\end{figure}

\subsection{Range where the deep-left-tail approximation fails}
Figure~\ref{fig:leftrisk} shows the risk of the deep left tail. At $k=5$, $M=65536$, and $z_{\rm sh}\approx-3.5$, the
corrected nominal $1\%$ threshold produced 52 lower-tail observations among 4549 failure samples, an actual rate of
$1.14\%$, with a $95\%$ Wilson interval of $[0.87\%,1.50\%]$. The corresponding $k=3$ configuration gives 273 among
9021 failures, namely $3.03\%$, with interval $[2.69\%,3.40\%]$; if only the macroscopic core layer is considered,
the rate is about $0.71\%$.

This shows that at this finite size, small-core failures still affect the measurement distribution over all failure
samples. On the other hand, even when the macroscopic core dominates, calibration can still be off: at $k=5$ and
$M=256$ the deep-left-tail points are almost entirely macroscopic core failures, yet the actual rate of the nominal
$1\%$ threshold is only about $0.31\%$. The accuracy of the deep left tail is determined by two factors together: the
error of the conditional distribution itself at finite size, and the contamination from small-core failures.
\begin{figure}[!htbp]
\centering\includegraphics[width=\linewidth]{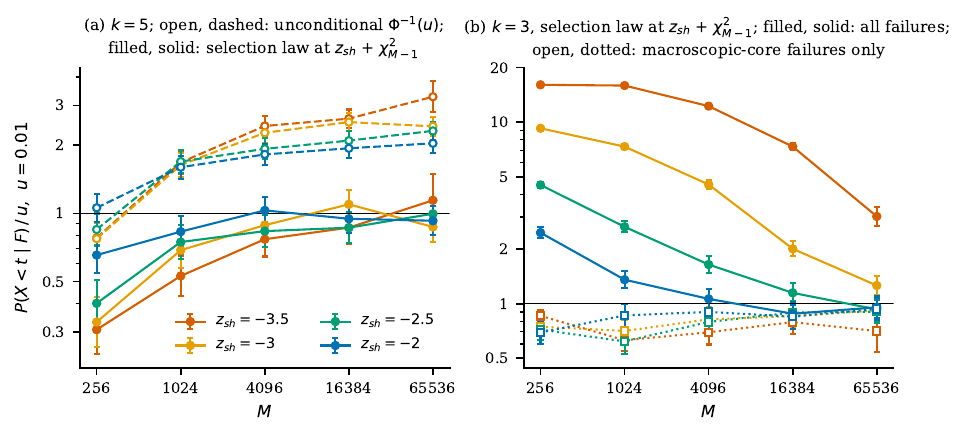}
\caption{deep-left-tail risk. (a) $k=5$, the unconditional threshold and the selection-law threshold; (b) $k=3$, all
failure samples and the macroscopic core layer.}
\label{fig:leftrisk}
\end{figure}
\subsection{The moderate-deviation conditional law}\label{app:exp-md}
\paragraph{The moderate-deviation conditional law.} In the moderate deviation regime, macroscopic failure is already a
rare event, so we use splitting to estimate its probability and conditional distribution. The GPU splitting
sampler used below was compared with direct sampling under criteria fixed in advance: over 10 configurations
with two independent pilots each, covering failure probabilities from $4.9\times10^{-7}$ to $3\times10^{-3}$ and
$M\le65536$, the largest of the 20 deviations is $|z|=2.11$. A CPU implementation of the same splitting scheme
passed an analogous comparison: across 16 configurations, all 80 comparisons satisfy
$|z|\le3.5$, the maximum $|z|$ is $2.66$, and the fraction with $|z|>2.5$ is $3.75\%$.

GPU splitting ($k=5$, all-positive signs, macroscopic event core$>0.1d$, 27 configurations) gives
the difference $\Delta\mu_{\rm sh}$ between the macroscopic-layer conditional mean and the analytic prediction of the
selection law. This difference contracts as $M^{b}$, and the fitted exponents at the three depths
$z_0=-3.5,-4.5,-5.5$ are $-0.540$, $-0.553$, and $-0.525$, with all $99\%$ intervals lying within the pre-fixed range
$[-0.75,-0.25]$. Taking $z_0=-3.5$ as an example, the difference falls from $0.157$ at $M=1024$ to $0.012$ at
$M=65536$. The measured conditional variance lies between $0.98$--$1.04$, while the selection law gives $0.984$; for
$M\ge16384$ the two differ by at most $0.015$. The ratio of the macroscopic probability to the shifted-coordinate
prediction lies between $0.76$--$1.32$, and the largest deviation of the macroscopic weight is $0.067$, all satisfying
the pre-fixed criteria.

In this $k=5$ splitting experiment, the rejection rate of the corrected selection threshold inside the macroscopic
layer is judged by a pre-fixed range $[0.7u,1.3u]$ for its point estimate, which differs from the interval test of
Section~\ref{sec:num-loadaxis}. It is the only criterion that deviates: 20 of the 78 test units fall outside it, all on the conservative
side (ratios $0.21$--$0.69$), concentrated on the near-one-sided compositions with $M\le4096$; for $M\ge32768$ no test
unit falls outside the range. This shows that at finite size the direction of the calibration error of the macroscopic
layer is conservative.

Finally, consider the rejection rate under failure at the true $d$. The mixture rule, the envelope rule, and the
corrected selection threshold have no rejected test unit, with maximum miss$/u$ equal to $1.043$, $1.024$, and
$1.043$. As a contrast, the unconditional threshold reaches up to $3.29$ times; the asymptotic mixture rule reaches up
to $1.72$ times in the crossover region at $k=5$, and does not shrink as $M$ grows.
\subsection{Other candidate rules in the coverage tests}\label{app:exp-e2e-other}
The coverage tests of Section~\ref{sec:e2e} also include two candidate rules, both of which under-cover and are
therefore not used as implementation rules. The first is the asymptotic mixture rule of Appendix~\ref{app:axis},
namely the asymptotic form of Proposition~\ref{prop:axis}, without threshold shift or chi-square correction; its
under-coverage reaches up to 1.2 times the nominal value. The second sets the threshold by $\widehat d$ and switches
between the unconditional normal and corrected selection thresholds. With a cutoff of $-6$ on the theorem coordinate
$z_0$ evaluated at $n=\widehat d$, it uses the unconditional normal threshold below the cutoff and the corrected selection
threshold otherwise; its under-coverage reaches up to 3.24 times at $k=3$. The switching
rules tested reach up to 3.14 times at $k\ge5$. As a contrast, the unconditional threshold reaches up to 3.1 times at
$k\ge5$.

\subsection{Finite-size location of the mean-sign-reversal layer}\label{app:exp-wm}
\paragraph{The mean-sign-reversal layer.} Figure~\ref{fig:loadaxis} examines Proposition~\ref{prop:wm} at finite
size. Panel (a) takes $k=5$ and all-same signs, and draws two layer lines, both in shifted coordinates: one is the
equal-probability layer $\Phi(z_{\rm sh})=\lambda_M$, and the other is the mean-sign-reversal layer
$s_M\eta_k\phi(z_{\rm sh})=2\lambda_M$. The reversal layer is always deeper; at $M=2^{16}$ the two layers sit at
$t=7.66$ and $t=8.29$. Between the two layers, failures mostly come from repeated edges, but the sign of the
conditional mean is still determined by the macroscopic core.

The dots in panel (b) are assembled from two parts: the macroscopic-layer probability and mean given by splitting,
and the exact mean $2/s_M$ of the repeated-edge layer. The two configurations at $M=4096$ cross the reversal layer.
For these two configurations, the repeated-edge contribution is estimated from simulations with a pair of identical edges
inserted, and a joint repeated bootstrap is applied to the two layers, with the following results: at $t=6.67$, $c_F\binom d2=-5.96$
with $95\%$ interval $[-6.17,-5.75]$; at $t=7.17$, $+0.48$ with interval $[+0.39,+0.57]$; the theoretical
sign-reversal point is at $t=7.15$. Small-core failures that are neither the macroscopic core nor repeated pairs can
be controlled by finite combinatorial counting, and their effect at these two points is at most $0.0022$. For
$M\ge16384$, the deepest configuration still lies before the reversal layer, with negative mean. The finite-size
offset of the macroscopic-layer mean relative to the selection law contracts roughly as $M^{-1/2}$, and after
multiplying by $s_M$ it is an $O(1)$ quantity on the original scale, so the dots in (b) lie slightly above the
theoretical curve overall.
\begin{figure}[!htbp]
\centering\includegraphics[width=\linewidth]{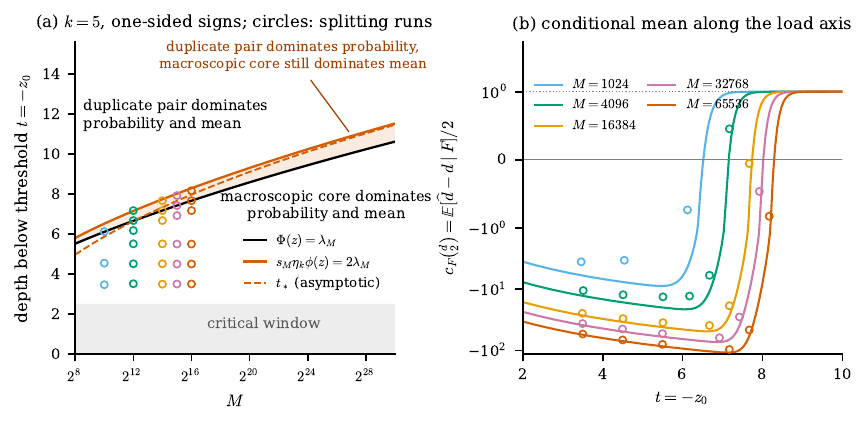}
\caption{Load axis for $k=5$ and all-same signs. (a) The equal-probability layer (black) and the mean-sign-reversal
layer (orange), the dashed line is $t_*$, the dots are splitting configurations; (b) conditional mean on the original
scale, the curve is the theoretical value, the dots are the layer-wise reconstruction.}
\label{fig:loadaxis}
\end{figure}
\subsection{Supplementary checks of the analytic constants}\label{app:exp-constants}
\emph{Window width.} The measured slope $\Delta z/\Delta r$ over failure probabilities from $0.1$ to $0.9$ exceeds the
analytic value $ka_ce^{-x_c}/\sigma_k$ by $3.6\%$ ($k=5$, $M=4096$), $2.3\%$ ($k=5$, $M=16384$), and about $2\%$
($k=3$, $M=262144$), respectively, and the gap shrinks as $M$ grows. This is an independent check of $\sigma_k^2$ and
its dynamic part $\tau_k^2$.

\emph{Dynamic variance.} At $k=3$ and $M=16384$, at the moment at which the score is taken, the measured variance of
the dynamic innovation is $0.00816$, and the partial integral $\tau^2(\varepsilon_M)$ is also $0.00816$; the share held
by the initial state is consistent with $\sigma_0^2=0.01632$.

\emph{Selection strength.} At $k=6$, $M=16384$, and failure probability $0.49$, the measured slope of the Kolmogorov
distance between failure and success samples with respect to $\rho^2$ is $0.1449$, and the first-order prediction is
$0.1397$.

\subsection{Cross-check of an independent implementation}\label{app:exp-indep}
We additionally wrote a C-language sampler that shares no code with the sampler above; at the same time we recomputed
all the constants by direct numerical integration from the definition, and the results agree to numerical precision
with the closed forms. Table~\ref{tab:indep} lists the results for $M=16384$
and $r=0$; each cell is run in eight chunks of 50{,}000 trials with fixed seeds. The theoretical value takes $c=\eta_k(d\rho^2-1)/(d-1)$, and $z$ takes the coordinate corresponding to the
measured failure probability. One can see that the means for $k=3$ and $k=5$ have opposite signs; the offset at
$\rho=0.6$ is about $0.36$ times that at $\rho=1$; and at $\rho=0$ the standard normal is recovered.
\begin{table}[!htbp]
\centering
\caption{Checks of the independent sampler ($M=16384$, $r=0$, 400 thousand trials per cell).}
\label{tab:indep}
\begin{tabular}{lrrrrrr}
\toprule
Configuration & $p_F$ & $\E[X\mid F]$ & theory & $\Var(X\mid F)$ & theory & KS (selection law)\\
\midrule
$k=3$, $\rho=1$ & 0.646 & $+0.2056$ & $+0.2025$ & 0.934 & 0.932 & 0.0022\\
$k=3$, $\rho=0.6$ & 0.646 & $+0.0741$ & $+0.0730$ & 0.986 & 0.991 & 0.0032\\
$k=3$, $\rho=0$ & 0.646 & $+0.0027$ & $0$ & 0.998 & 1.000 & 0.0013\\
$k=5$, $\rho=1$ & 0.657 & $-0.0737$ & $-0.0710$ & 0.991 & 0.991 & 0.0041\\
$k=5$, $\rho=0.6$ & 0.659 & $-0.0266$ & $-0.0254$ & 0.997 & 0.999 & 0.0032\\
$k=5$, $\rho=0$ & 0.659 & $-0.0012$ & $0$ & 1.002 & 1.000 & 0.0040\\
\bottomrule
\end{tabular}
\end{table}
The standard error of the mean is about $0.002$. The remaining deviation is of the same order as the threshold center
shift and the finite-size bias of the window width in Section~\ref{sec:threshold}, and shrinks as $M$ grows.

\bibliographystyle{plain}
\bibliography{references}
\end{document}